\documentclass[12pt,english]{article}
\pdfoutput=1
\usepackage{mathptmx}
\usepackage{helvet}
\usepackage{courier}
\usepackage[T1]{fontenc}
\usepackage[latin9]{inputenc}
\usepackage[a4paper]{geometry}
\usepackage{color}
\usepackage{babel}
\usepackage{array}
\usepackage{float}
\usepackage{booktabs}
\usepackage{multirow}
\usepackage{amsmath}
\usepackage{amsthm}
\usepackage{amssymb}
\usepackage{graphicx}
\usepackage{rotfloat}
\usepackage{setspace}
\usepackage[authoryear]{natbib}
\usepackage[unicode=true,
 bookmarks=true,bookmarksnumbered=false,bookmarksopen=false,
 breaklinks=true,pdfborder={0 0 0},pdfborderstyle={},backref=false,colorlinks=true]
 {hyperref}
\hypersetup{
 allcolors=blue}

\makeatletter

\providecommand{\tabularnewline}{\\}
\usepackage{tikz}
\usetikzlibrary{calc}

\theoremstyle{plain}
\newtheorem{assumption}{\protect\assumptionname}
\theoremstyle{plain}
\newtheorem{lem}{\protect\lemmaname}
\theoremstyle{plain}
\newtheorem{prop}{\protect\propositionname}
\theoremstyle{plain}
\newtheorem{cor}{\protect\corollaryname}

\date{}

\usepackage[position=top]{subfig}

\usepackage{bbm}
\usepackage{amssymb}
\usepackage{enumitem}

\usepackage[font=default,labelfont=bf,labelsep=period]{caption}
\usepackage{placeins}

\renewcommand{\section}{\@startsection {section}{1}{\z@}%
                                   {-3.5ex \@plus -1ex \@minus -.2ex}%
                                   {2.3ex \@plus.2ex}%
                                   {\reset@font\large\bfseries}}
\renewcommand{\subsection}{\@startsection{subsection}{2}{\z@}%
                                     {-3.25ex\@plus -1ex \@minus -.2ex}%
                                     {1.5ex \@plus .2ex}%
                                     {\reset@font\itshape}}
\renewcommand{\subsubsection}{\@startsection{subsubsection}{3}{\z@}%
                                     {-3.25ex\@plus -1ex \@minus -.2ex}%
                                     {-1em}%
                                     {\reset@font\normalsize\itshape}}

\allowdisplaybreaks

\@ifundefined{showcaptionsetup}{}{%
 \PassOptionsToPackage{caption=false}{subfig}}
\usepackage{subfig}
\makeatother

\providecommand{\assumptionname}{Assumption}
\providecommand{\corollaryname}{Corollary}
\providecommand{\lemmaname}{Lemma}
\providecommand{\propositionname}{Proposition}

\begin{document}
\title{Instrumental variables with unordered treatments: Theory and evidence from returns to fields of study\thanks{We thank Arnstein Vestre for excellent research assistance. Edwin Leuven recognizes support from the Norwegian Research Council, project no. 275906.}}
\author{Eskil Heinesen\thanks{Rockwool Foundation Research Unit. Email: \protect\href{mailto:esh@rff.dk}{esh@rff.dk}} \and Christian Hvid\thanks{Email: \protect\href{mailto:christian@hvids.eu}{christian@hvids.eu}} \and Lars Kirkeb?en\thanks{Statistics Norway. Email: \protect\href{mailto:kir@ssb.no}{kir@ssb.no}} \and Edwin Leuven\thanks{Department of Economics, University of Oslo; Statistics Norway; CESifo; IZA.  Email: \protect\href{mailto:edwin.leuven@econ.uio.no}{edwin.leuven@econ.uio.no}.} \and Magne Mogstad\thanks{University of Chicago, Department of Economics; Statistics Norway; NBER. Email: \protect\href{mailto:magne.mogstad@gmail.com\%20}{magne.mogstad@gmail.com }}}

\maketitle
\global\long\def\E{\mathbb{E}}%
\global\long\def\d#1{\mathbbm{1}_{[#1]}}%

\begin{abstract}
\begin{singlespace}
\noindent We revisit the identification argument of \citet{kirkeboen_field_2016} who showed how one may combine instruments for multiple unordered treatments with information about individuals\textquoteright{} ranking of these treatments to achieve identification while allowing for both observed and unobserved heterogeneity in treatment effects. We show that the key assumptions underlying their identification argument have testable implications. We also provide a new characterization of the bias that may arise if these assumptions are violated. Taken together, these results allow researchers not only to test the underlying assumptions, but also to argue whether the bias from violation of these assumptions are likely to be economically meaningful. Guided and motivated by these results, we estimate and compare the earnings payoffs to post-secondary fields of study in Norway and Denmark. In each country, we apply the identification argument of \citet{kirkeboen_field_2016} to data on individuals' ranking of fields of study and field-specific instruments from discontinuities in the admission systems. We empirically examine whether and why the payoffs to fields of study differ across the two countries. We find strong cross-country correlation in the payoffs to fields of study, especially after removing fields with violations of the assumptions underlying the identification argument.
\end{singlespace}
\end{abstract}
\vfill{}

\thispagestyle{empty}

\pagebreak{}

\setcounter{page}{1}
\begin{onehalfspace}

\section{Introduction}
\end{onehalfspace}

\noindent Instrumental variables (IV) estimation of treatment effects is challenging if there are multiple unordered treatments. Not only does identification require (at least) one instrument per alternative, but it is also necessary to deal with the issue that individuals who choose the same treatment may have different next-best treatments. One way to resolve this issue is to assume homogeneous treatment effects. If effects are heterogeneous across individuals (conditional on observable characteristics), then standard 2SLS does not identify the payoff to any individual or group of the population from choosing one treatment instead of another.\footnote{A number of studies in diverse fields report evidence of unobserved heterogeneity in causal effects (see, for example, the review article by \citet{mogstad2018identification}).}

We revisit the identification argument of \citet{kirkeboen_field_2016} who showed how one may combine instruments for multiple unordered treatments with information about individuals\textquoteright{} ranking of these treatment to achieve identification while allowing for both observed and unobserved heterogeneity in treatment effects.\footnote{\citet{kirkeboen_field_2016} contributes to a larger literature on identification of treatment effects in unordered choice models. \citet{heckman2006understanding} and \citet{heckman2010comparing} discuss the challenges associated with the identification and interpretation of treatment effects in such models. See also the recent work by \citet{kamat2017identification} and \citet{lee2020filtered}.} We show that the key assumptions underlying their identification argument have testable implications. We also provide a new characterization of the bias that may arise if these assumptions are violated.\footnote{Throughout the paper, we use the term \emph{bias} to describe the difference between two population quantities, namely the IV estimand and the parameter of interest, that is the positively weighted average of treatment effects for some complier group.} Taken together, these results allow researchers not only to test the underlying assumptions, but also to argue whether the bias from violation of these assumptions are likely to be economically meaningful. Guided and motivated by these results, we estimate and compare the earnings payoffs to post-secondary fields of study in Norway and Denmark.\footnote{There is a growing body of work on the payoffs to field of study or college major, reviewed in \citet{altonji2012heterogeneity,altonji2015analysis}, and \citet{kirkeboen_field_2016}. The latter study also reports IV estimates of the payoffs to fields of study from Norway. Thus, our empirical contribution is the new payoff estimates from Denmark, the examination of the IV assumptions, and the comparison of payoffs to fields of study between Norway and Denmark.} In each country, we apply the identification argument of \citet{kirkeboen_field_2016} to data on individuals' ranking of fields of study and field-specific instruments from discontinuities in the admission systems. We then empirically examine the extent to which and why the payoffs to fields of study differ across the two countries.

In Section 2, we begin by briefly reviewing IV in settings with multiple unordered treatments, laying the groundwork for our analysis. As in the analysis of binary treatments in \citet{imbens_identification_1994}, we allow for heterogeneous effects and assume that each instrument is exogenous and satisfies a monotonicity condition. Our point of departure is the key result in \citet{kirkeboen_field_2016}: IV can then be used to identify local average treatment effects (LATEs) of unordered treatments under the additional assumptions that the analyst observes individuals' next-best alternatives and an irrelevance condition on preferences.

The next two sections of the paper examine whether the additional assumptions of \citet{kirkeboen_field_2016} have testable implications and the bias that may arise if they are violated. To do so, it is useful to stratify the population into a set of instrument-dependent groups, sometimes referred to as principal strata. These groups are defined by the manner in which members of the population react to the instruments. In addition to the usual compliers, always takers, and never takers of \citet{imbens_identification_1994}, there are two so-called defier groups (both of which are distinct from the usual defier group that exists if the usual monotonicity condition fails). The first is the next-best defiers. In the context of our empirical application, this group consists of individuals who would choose their preferred field if above the admission cutoff, but otherwise choose fields other than the stated next-best alternative. The others are the irrelevance-defiers. In our context, the irrelevance assumption means that if crossing the admission cutoff to a given field does not make an individual choose that field, it should not affect her choice of other fields either.

In Section 3, we use this stratification of the population to characterize the bias in the IV estimands that may arise in the presence of next-best defiers, or irrelevance defiers, or both. It is useful to observe that the bias due to each type of defier has a product structure: It depends on the number of defiers compared to compliers, multiplied by the difference between compliers and defiers in the average payoff to choosing one type of education compared to another. Thus, there will be zero bias if there either are no defiers or if the average payoff to choosing one type of education compared to another is the same for defiers and compliers. Furthermore, the bias becomes large only if there are many defiers relative to compliers \emph{and} there are large differences in the payoff between compliers and defiers.

In Section 4, we show that the shares of next-best and irrelevance defiers can be bounded, but not point identified. We derive sharp bounds -- which are nontrivial -- and, thus, provides testable implications of the additional assumptions of \citet{kirkeboen_field_2016}. We show that these results have implications for the recent work of \citet{nibbering_clustered_2022} who propose an algorithm which aggregate fields into clusters based on estimated first-stage coefficients. The motivation for their approach is to avoid bias from irrelevance and next-best defiers. We show that their approach requires point identification of the shares of next-best and irrelevance defiers, and that it may produce biased estimates even if effects are constant across individuals (in contrast to standard 2SLS).

The last three sections of the paper take the theoretical results discussed above to the data by comparing payoff estimates for two countries, Denmark and Norway. These are two geographically and culturally close open-economies with very comparable educational institutions as well as similar tax, welfare and social benefits systems. It seems therefore natural to expect that payoffs will -- at least to a degree -- be aligned, and that differences can potentially be understood in light of violations of the assumptions underlying approach outlined above and detailed below.

In Section 5, we present the institutional background and data sources in Denmark and Norway. This section highlights the common institutional framework and data sources, documents how educational classifications and outcomes are harmonized across countries, and discusses differences that may be consequential for the analysis and results. In Section 6, we present the empirical specification that generates the payoff estimates for the two countries, following closely \citet{kirkeboen_field_2016}. Two challenges that must be met when comparing estimates from two different populations are the reference population and measurement error. Section 6 therefore also defines the population of compliers that we use to anchor the estimates, and presents an error-in-variables approach that addresses bias arising from measurement error when comparing the noisy payoff estimates.

In Section 7, we present the estimation results. We first turn to the first-stages and, building on the results from Section 4, document that in both countries the violations of irrelevance or next-best are non-trivial and appear to be of similar magnitude but of a different nature. In Norway, there is clear evidence of violations of next-best, but little if any sign of violations of irrelevance; in Denmark, the two types of violations seem to be equally frequent. The accompanying second-stages (with earnings measured eight years after application) show that, on average, the estimated annual payoff to completing a field-of-study instead of the next-best is about 2,200 USD in Denmark, while in Norway the payoff estimates are substantially larger and around 22,000 USD. The payoffs in Norway also exhibit a higher variance than in Denmark, and overall we strongly reject that the payoffs are the same. Despite these differences in levels and variation the payoffs significantly co-vary, and we estimate a correlation coefficient of 0.65 for our population of interest.

This correlation substantially increases  when we exclude the estimates with the most violations of the irrelevance and next-best assumptions. However, violations of irrelevance and next-best do not appear to explain the lower level and variation of the payoffs in Denmark compared to Norway. Additional exploratory analyses show that these across country differences are mostly driven by heterogeneity in next-best fields, and can partly be explained by differences in selectivity (as measured by students' high school test scores).

\section{IV with multiple unordered treatments\label{sec:IV-in-unordered}}

\subsection{Models and assumptions\label{sec:assumptions-notation}}

We assume individuals choose between three mutually exclusive and collectively exhaustive alternatives $d\in\{0,1,2\}$. To fix ideas we envision these as enrolling in three different fields of study. We suppress the individual index and abstract from control variables. We want to interpret IV estimates of the equation
\begin{equation}
y=\beta_{0}+\beta_{1}d_{1}+\beta_{2}d_{2}+\varepsilon\label{eq:equation-of-interest}
\end{equation}
where $y$ is an observed outcome such as earnings, and $d_{j}\equiv\d{d=j}$ is a treatment indicator. Without loss of generality we choose field 0 as reference field, so that $\beta_{1}^{IV}$ ($\beta_{2}^{IV}$) is the payoff from choosing field 1 (2) over field 0.

We suppose individuals are randomly assigned to one of three mutually exclusive and collectively exhaustive groups $Z\in\{0,1,2\}$ and let $z_{j}=\d{Z=j}$ be an indicator variable that equals 1 if an individual is assigned to group $j$ and 0 otherwise. The indicator $z_{j}$ can be thought of as an instrument shifting the costs or benefits of choosing field $j$. For each individual, this gives three potential field choices $d^{z}$ and nine potential outcomes $y^{d,z}$. We let $\mathbf{d}$ denote the column vector of treatment indicators and $\mathbf{z}$ the column vector of instruments. We define $d_{j}^{z}\equiv\d{d^{z}=j}$ to be an indicator variable that tells us whether an individual would choose field $j$ for a given value of $Z$.

As in the analysis of binary treatments in \citet{imbens_identification_1994}, we allow for heterogeneous effects and assume that each instrument satisfies the following assumptions:
\begin{assumption}
IV Assumptions\label{ass:iv-ass}

\begin{enumerate}[label={\emph{(\alph*)}}, ref={1(\alph*)}]

\item \label{ass:excl}\textbf{\emph{Exclusion:}} $y^{d,z}=y^{d}$ for all $d,z$

\item \label{ass:indep}\textbf{\emph{Independence:}} $y^{d},d^{z}\perp Z$ for all $d,z$

\item \label{ass:rank}\textbf{\emph{Rank:}} $E[\mathbf{z}\mathbf{d}^{\top}]$ has full rank

\item \label{ass:mono}\textbf{\emph{Monotonicity:}} $d_{k}^{k}\geq d_{k}^{k'}$ for each field assignment pair $k,k'$

\end{enumerate}
\end{assumption}
Given our notation and assumptions, we can link the observed and potential outcomes and choices as follows,
\begin{align}
y & =y^{0}d_{0}+y^{1}d_{1}+y^{2}d_{2}\label{eq:outcome-eq}\\
d_{j} & =d_{j}^{0}z_{0}+d_{j}^{1}z_{1}+d_{j}^{2}z_{2}\qquad\text{for }j=0,1,2\label{eq:choice-eq-0}
\end{align}

\noindent These equations represent a model with multiple unordered treatment that permits unrestricted unobserved heterogeneity in treatment effects. Extending the model (and our theoretical results) to more than three choice alternatives is straightforward.

\subsection{Principal strata}

In Table \ref{tab:taxonomy} we invoke assumptions \ref{ass:excl}--\ref{ass:mono} and characterize the principal strata, that is, the  groups of individuals defined by how their potential field choices depend on the instrument. The table considers the (sub)population with field 0 as the stated next-best alternative. For brevity, it does not include always takers of field 1 (2) (those who chose field 1 (2) irrespective of instrument value) and never takers of field 1 (2) (those who choose field 2 and 0 (1 and 0) irrespective of instrument value).

As shown in the table, there are two types of compliers, $C_{1}$ and $C_{2}$. The $C_{1}$ ($C_{2}$) compliers are individuals who choose field 1 (2) when the instrument takes value 1 (2), and the reference field 0 when the instrument takes the value 0. In addition, there are four types of defiers, irrelevance and next-best defiers of instruments 1 and 2. Irrelevance defiers $ID_{1}$ ($ID_{2}$) are individuals who choose field 2 (1) when the instrument takes value 1 (2) while choosing field 0 if the instrument takes value 0. Next-best defiers $ND_{1}$ ($ND_{2}$) are individuals who choose field 2 (1) when the instrument takes value 0 while choosing field 1 (2) if the instrument takes value 1 (2).

\begin{table}
\caption{Taxonomy of complier and defier groups with field 0 as the stated next-best alternative\label{tab:taxonomy}}
{\small{}}%
\begin{tabular*}{1\textwidth}{@{\extracolsep{\fill}}lcccccccc}
\toprule
 &  &  &  &  & \multicolumn{3}{c}{\textbf{\small{}Potential}} & \tabularnewline
\multicolumn{1}{l}{\textbf{\small{}Group}} &  & \multicolumn{3}{c}{} & \multicolumn{3}{c}{\textbf{\small{}Field Choice}} & \textbf{\small{}Characteristics}\tabularnewline
\cmidrule{6-8} \cmidrule{7-8} \cmidrule{8-8}
 &  &  &  &  & {\small{}$d^{0}$} & {\small{}$d^{1}$} & {\small{}$d^{2}$} & \tabularnewline
\midrule
{\small{}Instrument 1} &  &  &  &  &  &  &  & \tabularnewline
{\small{}- Compliers} & {\small{}$C_{1}$} &  &  &  & {\small{}0} & {\small{}1} &  & {\small{}$d_{1}^{1}-d_{1}^{0}=1\enskip\land\enskip d_{2}^{1}=d_{2}^{0}=\enskip0$}\tabularnewline
{\small{}- Irrelevance Defiers} & {\small{}$ID_{1}$} &  &  &  & {\small{}0} & {\small{}2} &  & {\small{}$d_{1}^{1}=d_{1}^{0}=0\enskip\land\enskip d_{2}^{1}-d_{2}^{0}=\enskip1$}\tabularnewline
{\small{}- Next-best Defiers} & {\small{}$ND_{1}$} &  &  &  & {\small{}2} & {\small{}1} &  & {\small{}$d_{1}^{1}-d_{1}^{0}=1\enskip\land\enskip d_{2}^{1}-d_{2}^{0}=-1$}\tabularnewline
{\small{}Instrument 2} &  &  &  &  &  &  &  & \tabularnewline
{\small{}- Compliers} & {\small{}$C_{2}$} &  &  &  & {\small{}0} &  & {\small{}2} & {\small{}$d_{2}^{2}-d_{2}^{0}=1\enskip\land\enskip d_{1}^{2}=d_{1}^{0}=\enskip0$}\tabularnewline
{\small{}- Irrelevance Defiers} & {\small{}$ID_{2}$} &  &  &  & {\small{}0} &  & {\small{}1} & {\small{}$d_{2}^{2}=d_{2}^{0}=0\enskip\land\enskip d_{1}^{2}-d_{1}^{0}=\enskip1$}\tabularnewline
{\small{}- Next-best Defiers} & {\small{}$ND_{2}$} &  &  &  & {\small{}1} &  & {\small{}2} & {\small{}$d_{2}^{2}-d_{2}^{0}=1\enskip\land\enskip d_{1}^{2}-d_{1}^{0}=-1$}\tabularnewline
\bottomrule
\end{tabular*}{\small\par}

\begin{singlespace}
\textbf{\footnotesize{}Note:}{\footnotesize{} The table characterizes compliers, irrelevance defiers and next-best defiers based on their potential treatments.}{\footnotesize\par}
\end{singlespace}
\end{table}

\subsection{Identification result}

\citet{kirkeboen_field_2016} suggest the following assumptions on the groups in Table \ref{tab:taxonomy} to obtain identification:\footnote{\citet{kirkeboen_field_2016} are imprecise about whether auxiliary assumption \ref{ass:nextbest} is imposed on everyone or only those individuals whose treatment status depends on the instrument. However, this is immaterial for their results, as well as ours. The reason is that always takers and never takers drop out of the IV estimand because their treatment status does not change with the instrument.}
\begin{assumption}
Auxiliary Assumptions\label{ass:kirkeboen}

\begin{enumerate}[label={\emph{(\alph*)}}, ref={2(\alph*)}]

\item \label{ass:irr} \textbf{\emph{Irrelevance:}} $d_{k}^{k}-d_{k}^{0}=0\implies d_{k'}^{k}=d_{k'}^{0}$ for all pairs $k,k'$

\item \label{ass:nextbest} \textbf{\emph{Next-best:}} We are able to condition on $d_{1}^{0}=d_{2}^{0}=0$ i.e. $d_{0}^{0}=1$.

\end{enumerate}
\end{assumption}
The irrelevance condition assumes that if changing z from 0 to 1 (2) does not induce an individual to choose treatment 1 (2), then it does not make her choose treatment 2 (1) either. In our context, for example, this assumption means that if crossing the admission cutoff to field 1 does not make an individual choose field 1, it does not make her choose field 2 either. The next-best alternative condition is effectively assuming that individuals' stated next-best alternative is their actual next-best alternative. The following lemma is immediate from these two assumptions:
\begin{lem}
\label{th:kirkeboen-lemma} Suppose Assumptions \ref{ass:iv-ass}--\ref{ass:kirkeboen} hold. Then $\beta_{1}^{IV},\beta_{2}^{IV}$ have a causal interpretation as positively weighted averages of treatment effects for compliers, and
\begin{align*}
\beta_{1}^{IV} & =\E[y^{1}-y^{0}\mid C_{1}]\\
\beta_{2}^{IV} & =\E[y^{2}-y^{0}\mid C_{2}]
\end{align*}
\end{lem}
\begin{proof}
For a proof, see \citet{kirkeboen_field_2016}.
\end{proof}
The core of Lemma \ref{th:kirkeboen-lemma} is that the IV estimand of $\beta_{1}$ ($\beta_{2}$) can be given an interpretation as a local average treatment effect (LATE) of an instrument-induced shift from field 0 to field 1 (2) for compliers when irrelevance and next-best defiers are assumed away.

\section{Interpretation of IV estimands if auxiliary assumptions fail\label{sec:derivation-bias}}

If Assumptions \ref{ass:irr}--\ref{ass:nextbest} do not hold, the IV estimand of $\beta_{1}$ ($\beta_{2}$) does not have a causal interpretation as a positively weighted average of treatment effects of choosing field 1 (2) over field 0. In the following, we characterize the bias that will occur in this case, and discuss in which situations the bias will be large and small.

\subsection{Assuming only next-best\label{sec:only-next-best}}

The IV estimands of $\beta_{1}$ and $\beta_{2}$ can be decomposed into a LATE for compliers and a bias term using IV moment conditions. In particular, if only next-best holds, but not irrelevance, we get the following decomposition, as shown in Appendix \ref{appx:no-auxillary-ass}.
\begin{prop}
\label{th:theorem-only-next-best} Suppose Assumptions \ref{ass:excl}--\ref{ass:mono} and \ref{ass:nextbest} hold. Then $\beta_{1}^{IV},\beta_{2}^{IV}$ do not have a causal interpretation as positively weighted averages of treatment effects for compliers,
\begin{align}
\beta_{1}^{IV}=\underbrace{\E[y^{1}-y^{0}\mid C_{1}]}_{\substack{\emph{A}}
}\quad & +\quad\underbrace{\frac{P(ID_{1})P(ID_{2})}{W'}}_{\substack{\omega_{1}}
}\times\enskip(\underbrace{\E[y^{1}-y^{0}\mid C_{1}]-\E[y^{1}-y^{0}\mid ID_{2}]}_{\substack{\Delta_{1}}
})\label{eq:full-next-best}\\[1pt]
 & -\quad\underbrace{\frac{P(ID_{1})P(C_{2})}{W'}}_{\substack{\omega_{2}}
}\enskip\times\enskip(\underbrace{\E[y^{2}-y^{0}\mid C_{2}]-\E[y^{2}-y^{0}\mid ID_{1}]}_{\substack{\Delta_{2}}
})\nonumber
\end{align}
where $W'=P(C_{1})P(C_{2})-P(ID_{1})P(ID_{2})$ and the expression for $\beta_{2}^{IV}$ follows by symmetry. \emph{A} is the complier LATE, $\omega_{1}$ and $\omega_{2}$ are defier group weights, and $\Delta_{1}$ and $\Delta_{2}$ are differences in the causal effects between compliers and irrelevance defiers.
\end{prop}
\begin{proof}
See appendix \ref{appx:no-auxillary-ass}.
\end{proof}
Imposing the constant effects assumption implies that the differences in the causal effects between defier groups ($\Delta_{1}$, $\Delta_{2}$) go to zero. In this case, $\beta_{1}^{IV}$ ($\beta_{2}^{IV}$) would recover the causal effect, $\E[y^{1}-y^{0}]$ ($\E[y^{2}-y^{0}]$). Imposing irrelevance implies that the defier weights ($\omega_{1}$, $\omega_{2}$) go to zero. In this case, $\beta_{1}^{IV}$ ($\beta_{2}^{IV}$) would recover the complier LATE, $\E[y^{1}-y^{0}\mid C_{1}]$ ($\E[y^{2}-y^{0}\mid C_{2}]$).

A central question for empirical researchers is when the bias in Proposition \ref{th:theorem-only-next-best} is likely to be large. To answer this question, it is useful to observe that the two bias terms in equation \ref{eq:full-next-best} are the products of a difference in causal effects and a defier weight consisting of the product of the propensities of irrelevance defiers divided by the difference between complier and defier propensity products.

Note that as long as $P(C_{1})P(C_{2})>2\times P(ID_{1})P(ID_{2})$ the weight $\omega_{1}$ is below 1. This will occur when there are many compliers relative to defiers. When the weight is below 1, the bias will always be smaller than the difference in causal effects. Due to the product structure ($\omega_{j}\times\Delta_{j}$) the bias due to violations of the irrelevance assumption will be very small when both $\omega_{j}$ and $\Delta_{j}$ are small. Conversely, in order for a large bias to occur, there needs to be both many defiers relative to compliers and a large difference in causal effects between the compliers and the irrelevance defiers.

We illustrate this with two examples. In both examples, we fix the LATE for compliers at \$1000. We focus on the first instrument, fixing the propensities of compliers and irrelevance defiers of instrument 2 to $P(ID_{2})=0.2$ and $P(C_{2})=0.8$, and, for simplicity, assume no always takers or never takers for any of the instruments, such that $P(C_{1})=1-P(ID_{1})$.

In Figure \ref{subfig:bias-differing-pid1-nb} we show how the bias from the first term varies with the propensity of irrelevance defiers. We let the difference in causal effects between compliers and instrument 2-defiers be fixed at three different levels: 10\%, 20\% and 50\% of the complier LATE. In Figure \ref{subfig:bias-differing-diff-eff-nb} we show the bias from the first term when varying the difference in causal effects between compliers and defiers. We let the propensity of irrelevance defiers be fixed at three different levels: low (0.1), middle (0.2), and high (0.5). The key take away is that the bias will be small even when there is a sizable number of defiers and a nontrivial difference in causal effects between the compliers and the defiers.

\begin{figure}
\subfloat[Varying Propensity\label{subfig:bias-differing-pid1-nb}]{\includegraphics[width=0.48\textwidth]{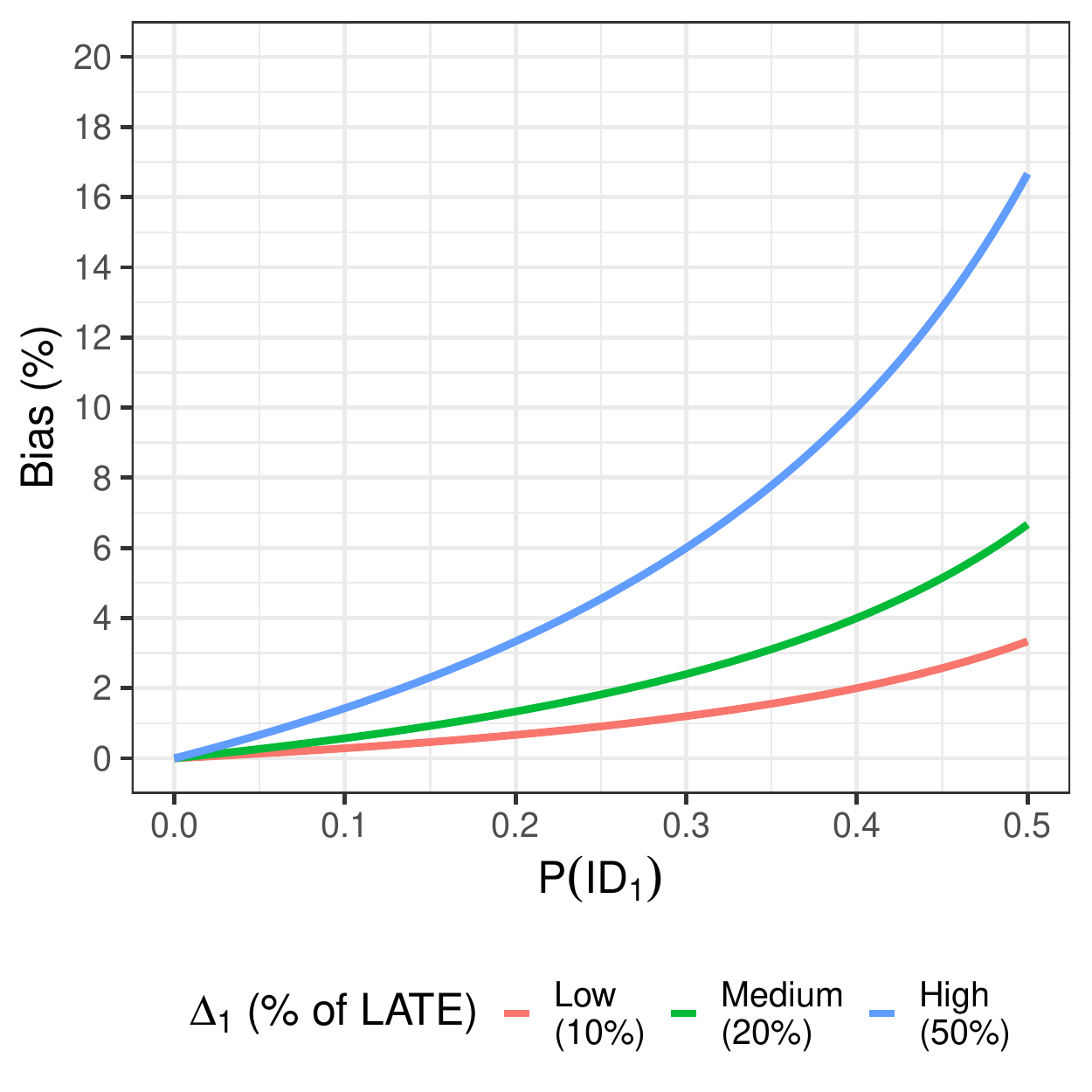}

}\subfloat[Varying Heterogeneity\label{subfig:bias-differing-diff-eff-nb}]{\includegraphics[width=0.48\textwidth]{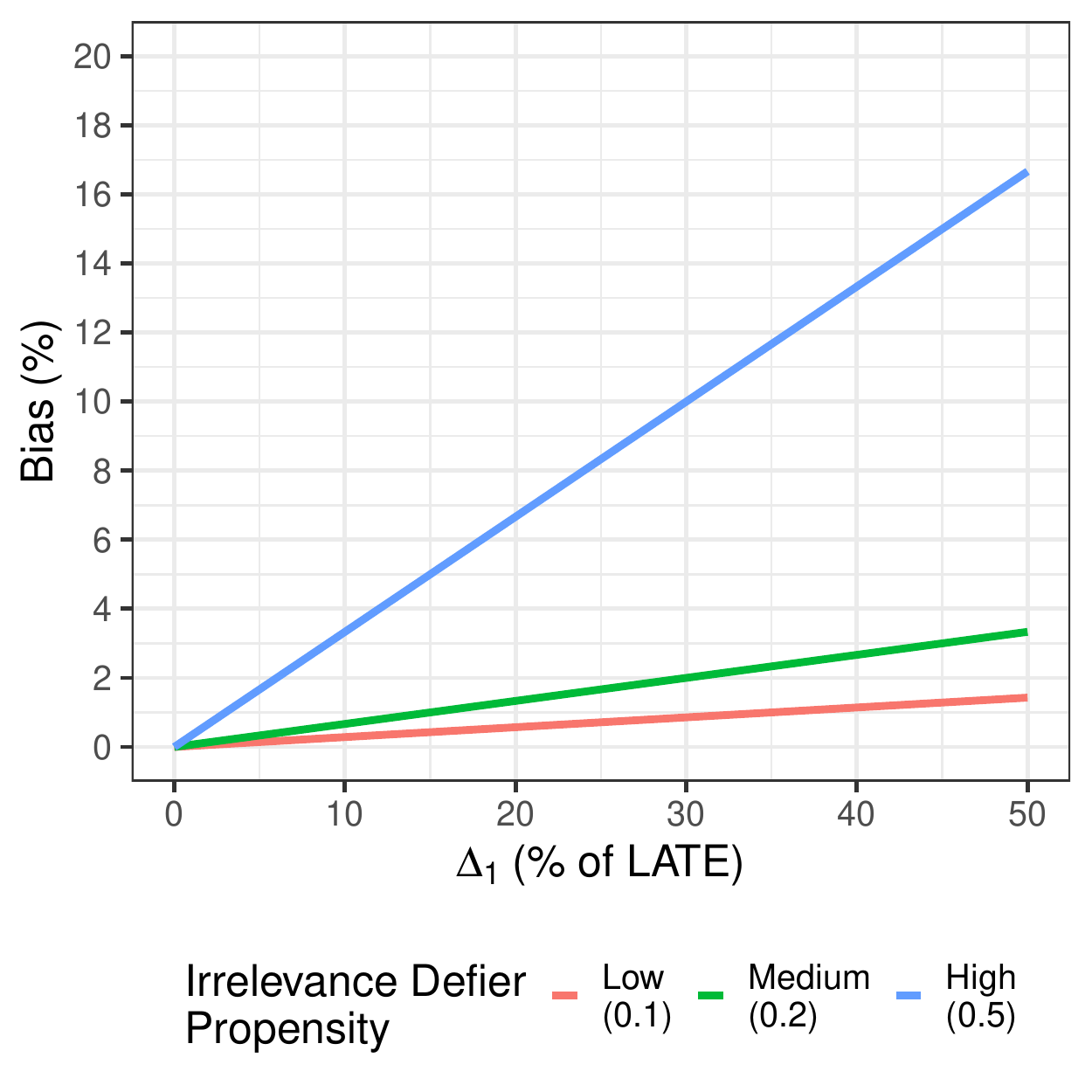}

}

\begin{singlespace}
\textbf{\footnotesize{}Note:}{\footnotesize{} Panel (a) shows the bias from irrelevance defiers for different defier propensities. The red line assumes a difference in causal effects between compliers and defiers at 10\% of the complier LATE, the green at 20\% and the blue at 50\%. Panel (b) shows the bias from irrelevance defiers for different levels of treatment effect heterogeneity. The red line assumes 10\%, the green 20\% and the blue 50\% irrelevance defiers. The number of defiers and compliers for instrument 2 is fixed at 20\% and 80\%.}{\footnotesize\par}
\end{singlespace}

\caption{Bias from irrelevance defiers under different defier weights and levels of heterogeneity.\label{fig:bias-differing-nb}}
\end{figure}

\subsection{Assuming only irrelevance\label{sec:only-irrelevance-ass}}

If irrelevance holds, but next-best is not observed, we may decompose the IV estimand into a complier LATE and a bias term.
\begin{prop}
\label{th:theorem-only-irrelevance} Suppose Assumptions \ref{ass:excl}--\ref{ass:mono} and \ref{ass:irr} hold. Then $\beta_{1}^{IV},\beta_{2}^{IV}$ do not have a causal interpretation as positively weighted averages of treatment effects for compliers,
\begin{align}
\beta_{1}^{IV}=\underbrace{\E[y^{1}-y^{0}\mid C_{1}]}_{\substack{\emph{A}}
}\quad & +\quad\underbrace{\frac{P(ND_{1})P(C_{2})}{\hat{W}}}_{\substack{\omega_{3}}
}\quad\times\enskip(\underbrace{\E[y^{1}-y^{0}\mid ND_{1}]-\E[y^{1}-y^{0}\mid C_{1}]}_{\substack{\Delta_{3}}
})\label{eq:full-irrelevance}\\[4pt]
 & -\quad\underbrace{\frac{P(ND_{1})P(C_{2})}{\hat{W}}}_{\substack{\omega_{4}}
}\quad\times\enskip(\underbrace{\E[y^{2}-y^{0}\mid ND_{1}]-\E[y^{2}-y^{0}\mid C_{2}]}_{\substack{\Delta_{4}}
})\nonumber \\[1pt]
 & +\quad\underbrace{\frac{P(ND_{1})P(ND_{2})}{\hat{W}}}_{\substack{\omega_{5}}
}\times\enskip(\underbrace{\E[y^{1}-y^{0}\mid ND_{1}]-\E[y^{1}-y^{0}\mid ND_{2}]}_{\substack{\Delta_{5}}
})\nonumber \\[1pt]
 & -\quad\underbrace{\frac{P(ND_{1})P(ND_{2})}{\hat{W}}}_{\substack{\omega_{6}}
}\times\enskip(\underbrace{\E[y^{2}-y^{0}\mid ND_{1}]-\E[y^{2}-y^{0}\mid ND_{2}]}_{\substack{\Delta_{6}}
})\nonumber
\end{align}
where $\hat{W}=P(C_{1})P(C_{2})+P(C_{1})P(ND_{2})+P(ND_{1})P(C_{2})$ and the expression for $\beta_{2}^{IV}$ follows by symmetry. \emph{A} is the complier LATE, $\omega_{3}$ through $\omega_{6}$ are defier group weights and $\Delta_{3}$ through $\Delta_{6}$ are differences in the causal effects between complier and defier groups.
\end{prop}
\begin{proof}
See appendix \ref{appx:no-auxillary-ass}.
\end{proof}
Imposing the constant effects assumption implies that the differences in causal effects between defier groups ($\Delta_{3}$ through $\Delta_{6}$) go to zero. In this case, $\beta_{1}^{IV}$ ($\beta_{2}^{IV}$) would recover the causal effect, $\E[y^{1}-y^{0}]$ ($\E[y^{2}-y^{0}]$). Observing the next-best alternative implies that the defier weights ($\omega_{3}$ through $\omega_{6}$) go to zero. In this case, $\beta_{1}^{IV}$ ($\beta_{2}^{IV}$) would recover the complier LATE, $\E[y^{1}-y^{0}\mid C_{1}]$ ($\E[y^{2}-y^{0}\mid C_{2}]$).

As in equation (\ref{eq:full-next-best}), the bias terms in equation (\ref{eq:full-irrelevance}) are the products of a difference in causal effects and a weight consisting of the product of the propensities of defiers divided by the sum of complier and defier propensity products.

Note that the weight in the first and second terms of equation (\ref{eq:full-irrelevance}) ($\omega_{3}$, $\omega_{4}$) are below 1, but that the weights for the two latter terms ($\omega_{5}$, $\omega_{6}$) can be above 1 if $P(ND_{1})P(ND_{2})>P(C_{1})P(C_{2})+P(C_{1})P(ND_{2})+P(ND_{1})P(C_{2})$. When the weight is below 1, the bias from the term will always be smaller than the difference in causal effects. Due to the product structure ($\omega_{j}\times\Delta_{j}$) the bias due to violations of the next-best assumption will be very small when both $\omega_{j}$ and $\Delta_{j}$ are small. Conversely, in order for a large bias to occur, we need both many defiers relative to compliers and a large difference in causal effects between the different groups.

We keep the same numerical example as in Section \ref{sec:only-next-best} and focus on the term $\omega_{3}\times\Delta_{3}$. In Figure \ref{subfig:bias-differing-pid1-irr} we show how the bias from this term varies with the propensity of next-best defiers. We let the difference in causal effects between compliers and defiers be fixed at three different levels: at 10\%, 20\% and 50\% of the complier LATE. In Figure \ref{subfig:bias-differing-diff-eff-irr} we show the bias when varying the difference in causal effects between compliers and defiers. We let the propensity of next-best defiers be fixed at three different levels: low (0.1), middle (0.2), and high (0.5). The key take away is as above that the bias will be small even when there is a sizable number of defiers and a nontrivial difference in causal effects between the compliers and the defiers.

\begin{figure}
\subfloat[Varying Propensity\label{subfig:bias-differing-pid1-irr}]{\includegraphics[width=0.48\textwidth]{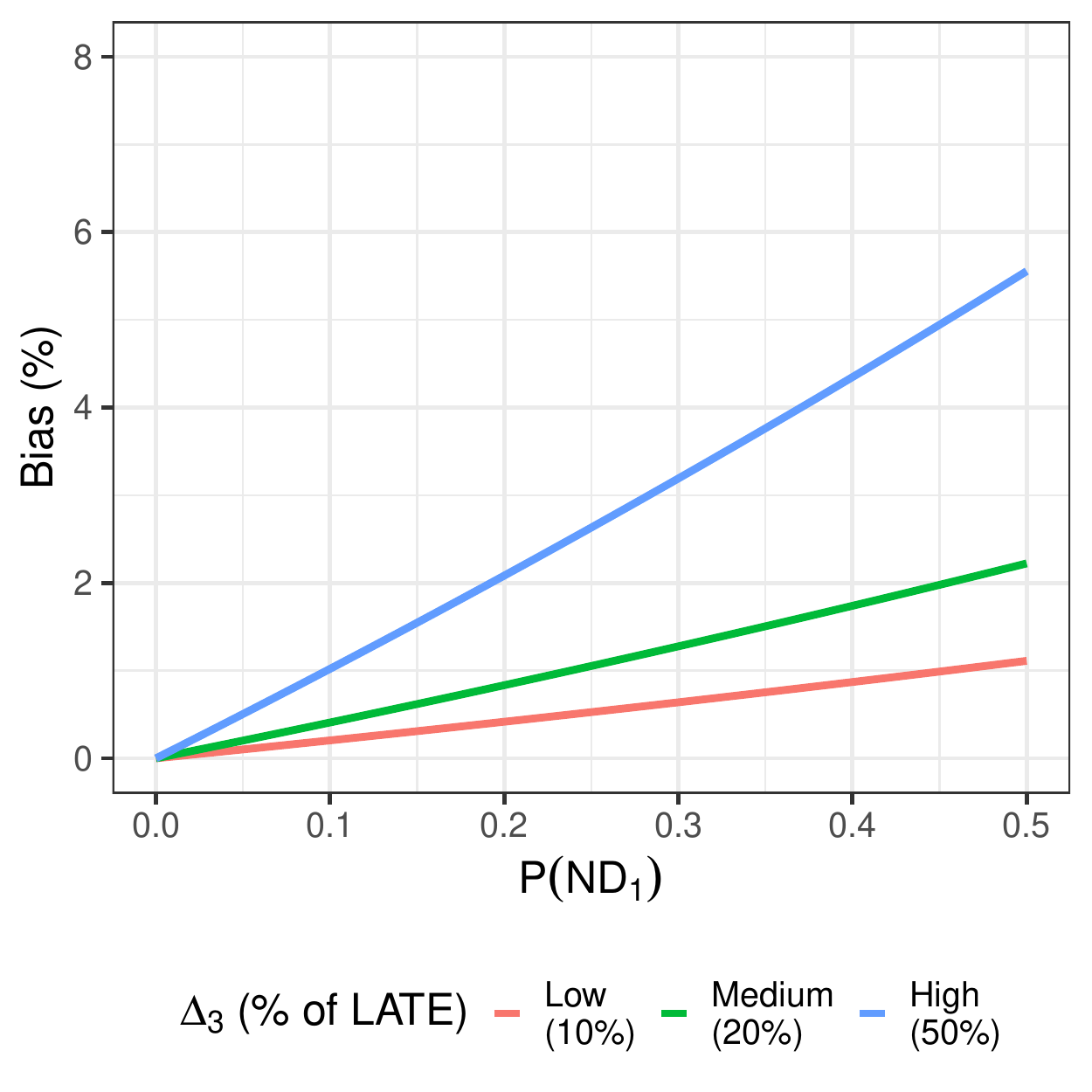}

}\subfloat[Varying Heterogeneity\label{subfig:bias-differing-diff-eff-irr}]{\includegraphics[width=0.48\textwidth]{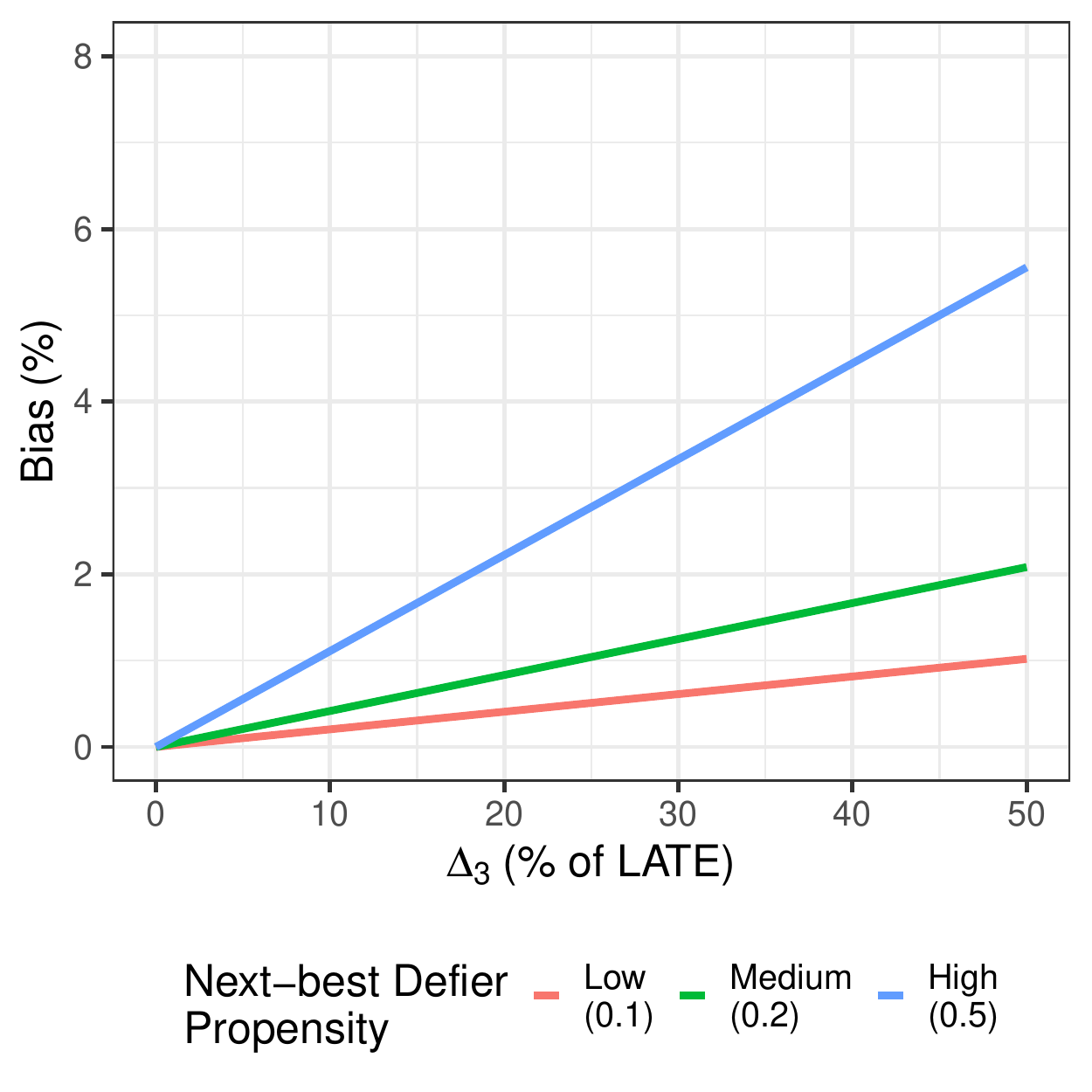}

}

\begin{singlespace}
\textbf{\footnotesize{}Note:}{\footnotesize{} Panel (a) shows one term of the bias from next-best defiers for different defier propensities. The red line assumes a difference in causal effects between compliers and defiers at 10\% of the complier LATE, the green at 20\% and the blue at 50\%. Panel (b) shows the bias from irrelevance defiers for different levels of treatment effect heterogeneity. The red line assumes 10\%, the green 20\% and the blue 50\% irrelevance defiers. The number of defiers and compliers for instrument 2 is fixed at 20\% and 80\%.}{\footnotesize\par}
\end{singlespace}

\caption{Bias from next-best defiers under different defier weights and levels of heterogeneity.\label{fig:bias-differing-irr}}
\end{figure}

\subsection{Assuming neither irrelevance nor next-best\label{sec:no-assumptions}}

If one neither makes the irrelevance assumption nor the next-best assumption, the IV estimand becomes the sum of the complier LATE, all bias terms from Propositions \ref{th:theorem-only-next-best} and \ref{th:theorem-only-irrelevance}, as well as a third set of interacted bias terms.
\begin{prop}
\label{th:theorem-no-auxillary} Suppose Assumptions \ref{ass:excl}--\ref{ass:mono} holds. Then $\beta_{1}^{IV},\beta_{2}^{IV}$ do not have a causal interpretation as positively weighted averages of treatment effects for compliers,
\begin{align}
\beta_{1}^{IV}=\underbrace{\E[y^{1}-y^{0}\mid C_{1}]}_{\substack{\emph{A}}
}\enskip & \enskip+\enskip\underbrace{\frac{P(ID_{1})P(ID_{2})}{\bar{W}}}_{\substack{\omega_{1}}
}\enskip\enskip\times\enskip(\underbrace{\E[y^{1}-y^{0}\mid C_{1}]-\E[y^{1}-y^{0}\mid ID_{2}]}_{\substack{\Delta_{1}}
})\label{eq:no-assumptions}\\[1pt]
 & \enskip-\enskip\underbrace{\frac{P(ID_{1})P(C_{2})}{\bar{W}}}_{\substack{\omega_{2}}
}\enskip\quad\times\enskip(\underbrace{\E[y^{2}-y^{0}\mid C_{2}]-\E[y^{2}-y^{0}\mid ID_{1}]}_{\substack{\Delta_{2}}
})\nonumber \\[1pt]
 & \enskip+\enskip\underbrace{\frac{P(ND_{1})P(C_{2})}{\bar{W}}}_{\substack{\omega_{3}}
}\quad\times\enskip(\underbrace{\E[y^{1}-y^{0}\mid ND_{1}]-\E[y^{1}-y^{0}\mid C_{1}]}_{\substack{\Delta_{3}}
})\nonumber \\[1pt]
 & \enskip-\enskip\underbrace{\frac{P(ND_{1})P(C_{2})}{\bar{W}}}_{\substack{\omega_{4}}
}\quad\times\enskip(\underbrace{\E[y^{2}-y^{0}\mid ND_{1}]-\E[y^{2}-y^{0}\mid C_{2}]}_{\substack{\Delta_{4}}
})\nonumber \\[4pt]
 & \enskip+\enskip\underbrace{\frac{P(ND_{1})P(ND_{2})}{\bar{W}}}_{\substack{\omega_{5}}
}\times\enskip(\underbrace{\E[y^{1}-y^{0}\mid ND_{1}]-\E[y^{1}-y^{0}\mid ND_{2}]}_{\substack{\Delta_{5}}
})\nonumber \\[1pt]
 & \enskip-\enskip\underbrace{\frac{P(ND_{1})P(ND_{2})}{\bar{W}}}_{\substack{\omega_{6}}
}\times\enskip(\underbrace{\E[y^{2}-y^{0}\mid ND_{1}]-\E[y^{2}-y^{0}\mid ND_{2}]}_{\substack{\Delta_{6}}
})\nonumber \\[1pt]
 & \enskip-\enskip\underbrace{\frac{P(ND_{1})P(ID_{2})}{\bar{W}}}_{\substack{\omega_{7}}
}\enskip\times\enskip(\underbrace{\E[y^{1}-y^{0}\mid C_{1}]-\E[y^{1}-y^{0}\mid ID_{2}]}_{\substack{\Delta_{7}}
})\nonumber \\[1pt]
 & \enskip+\enskip\underbrace{\frac{P(ID_{1})P(ND_{2})}{\bar{W}}}_{\substack{\omega_{8}}
}\enskip\times\enskip(\underbrace{\E[y^{1}-y^{0}\mid ND_{2}]-\E[y^{1}-y^{0}\mid C_{1}]}_{\substack{\Delta_{8}}
})\nonumber \\[1pt]
 & \enskip-\enskip\underbrace{\frac{P(ID_{1})P(ND_{2})}{\bar{W}}}_{\substack{\omega_{9}}
}\enskip\times\enskip(\underbrace{\E[y^{2}-y^{0}\mid ND_{2}]-\E[y^{2}-y^{0}\mid ID_{1}]}_{\substack{\Delta_{9}}
})\nonumber
\end{align}
where
\begin{align*}
\bar{W}\enskip & =\enskip P(C_{1})P(C_{2})+P(C_{1})P(ND_{2})+P(ND_{1})P(C_{2})\\[1pt]
 & \enskip+P(ND_{1})P(ID_{2})+P(ID_{1})P(ND_{2})-P(ID_{1})P(ID_{2})
\end{align*}
and the expression for $\beta_{2}^{IV}$ follows by symmetry.
\end{prop}
\begin{proof}
See appendix \ref{appx:no-auxillary-ass}.
\end{proof}
\emph{A} is the complier LATE, $\omega_{1}$ and $\omega_{2}$ are defier weights which also occur when observing the next-best alternative, $\omega_{3}$ through $\omega_{6}$ are defier weights which also occur under irrelevance and $\omega_{7}$ through $\omega_{9}$ are defier weights which occur only when neither assumption holds. $\Delta_{1}$, $\Delta_{2}$ and $\Delta_{7}$ are differences in the causal effects between irrelevance defiers and compliers, $\Delta_{3}$, $\Delta_{4}$ and $\Delta_{8}$ are differences in the causal effects between next-best defiers and compliers, while $\Delta_{5}$ and $\Delta_{6}$ are differences in causal effects between next-best defiers for the two different instruments, and $\Delta_{9}$ is the difference in the causal effects between next-best defiers of instrument 2 and irrelevance defiers of instrument 1.

Imposing the constant effects assumption implies that the differences in causal effects between defier groups ($\Delta_{1}$ through $\Delta_{9}$) go to zero. In this case, $\beta_{1}^{IV}$ ($\beta_{2}^{IV}$) would recover the causal effect, $\E[y^{1}-y^{0}]$ ($\E[y^{2}-y^{0}]$). Imposing the next-best assumption yields the result from Proposition \ref{th:theorem-only-next-best}, as weights $\omega_{3}$ through $\omega_{9}$ go to zero. Imposing the irrelevance assumption yields Proposition \ref{th:theorem-only-irrelevance}, as weights $\omega_{1}$, $\omega_{2}$ and $\omega_{7}$ through $\omega_{9}$ go to zero. Imposing both irrelevance and observing the next-best alternative make all defier weights ($\omega_{1}$ through $\omega_{9}$) go to zero. Then $\beta_{1}^{IV}$ ($\beta_{2}^{IV}$) would recover the complier LATE, $\E[y^{1}-y^{0}\mid C_{1}]$ ($\E[y^{2}-y^{0}\mid C_{2}]$).

Note that the bias in Proposition \ref{th:theorem-no-auxillary} is the sum of all bias terms from Propositions \ref{th:theorem-only-irrelevance} and \ref{th:theorem-only-next-best}, in addition to three new bias terms (except for a different denominator of the weights). These are terms following from interactions between irrelevance and next-best defiers, and rely on both types of defiers being present and having differences in causal effects between each other and with the complier group. As a result, the bias will be small unless there are relatively many of \emph{both} types of defiers and the causal effects are materially different between these groups and the compliers.

\section{Testable implications and aggregation\label{sec:empirical-test}}

\subsection{How to test the auxiliary assumptions}

The first stage equations for the IV estimates of equation (\ref{eq:equation-of-interest}) are given by:
\begin{align}
d_{1} & =\alpha_{1}^{0}+\alpha_{1}^{1}z_{1}+\alpha_{1}^{2}z_{2}+\nu_{1}\label{eq:first-stage-1}\\
d_{2} & =\alpha_{2}^{0}+\alpha_{2}^{1}z_{1}+\alpha_{2}^{2}z_{2}+\nu_{2}\label{eq:first-stage-2}
\end{align}
We now examine if it is possible to devise a test of whether the auxiliary assumptions \ref{ass:irr}--\ref{ass:nextbest} hold empirically. To do so, it is useful to characterize the quantities that the first stage coefficients recover:
\begin{lem}
\label{th:identified-quantities} Suppose Assumptions \ref{ass:excl}--\ref{ass:mono} hold. Then
\begin{align*}
\alpha_{1}^{0} & =P(AT_{1})\equiv P(OT_{2})+P(ND_{2}) & \alpha_{2}^{0} & =P(AT_{2})\equiv P(OT_{1})+P(ND_{1})\\
\alpha_{1}^{1} & =P(C_{1})+P(ND_{1}) & \alpha_{2}^{2} & =P(C_{2})+P(ND_{2})\\
\alpha_{2}^{1} & =P(ID_{1})-P(ND_{1}) & \alpha_{1}^{2} & =P(ID_{2})-P(ND_{2})
\end{align*}
where $AT_{1}$($AT_{2}$) are always-takers of field $1$ (2) when $z$ equals 0 or 1 (0 or 2), and $OT_{1}$ ($OT_{2}$) are global (for every value of the instrument) always takers of the other field 2 (1). See Appendix Table \ref{tab:taxonomy-detailed} for formal definitions of these instrument-specific strata.
\end{lem}
\begin{proof}
See appendix \ref{appx:empirical-test}.
\end{proof}
This result paves the way for the main result on the testability of the irrelevance and next best assumptions:
\begin{prop}
\label{th:testable-implications} Suppose Assumptions \ref{ass:excl}--\ref{ass:mono} hold. Then $P(ID_{1})$ and $P(ND_{1})$ are partially identified.
\begin{align*}
P(ND_{1}) & \in[\max\{0,-\alpha_{2}^{1}\},\qquad\qquad\quad\,\,\min\{\alpha_{1}^{1},\alpha_{2}^{0}\}]\\
P(ID_{1}) & \in[\max\{0,\enskip\,\alpha_{2}^{1}\},\max\{0,\alpha_{2}^{1}+\min\{\alpha_{1}^{1},\alpha_{2}^{0}\}\}]
\end{align*}
where results for $P(ID_{2})$ and $P(ND_{2})$ follow by symmetry.
\end{prop}
\begin{proof}
See appendix \ref{appx:empirical-test}.
\end{proof}
\noindent The practical implication of Proposition \ref{th:testable-implications} is that we cannot point identify the defier propensities without further assumptions. Yet, the assumptions are testable as the bounds will generally be nontrivial. Furthermore, if either assumption \ref{ass:irr} or \ref{ass:nextbest} is known to hold, the other assumption can be tested separately and $P(ID_{1})$ or $P(ND_{1})$ is point identified.
\begin{cor}
\label{cor:testable-next-best} Suppose Assumptions \ref{ass:excl}--\ref{ass:mono} and \ref{ass:nextbest} hold. Then $P(ND_{1})=P(ND_{2})=0$ and we can test whether assumption \ref{ass:irr} (irrelevance) holds, as $\alpha_{2}^{1}=P(ID_{1})$ and $\alpha_{1}^{2}=P(ID_{2})$.
\end{cor}
\begin{cor}
\label{cor:testable-irrelevance} Suppose Assumptions \ref{ass:excl}--\ref{ass:mono} and \ref{ass:irr} hold. Then $P(ID_{1})=P(ID_{2})=0$ and we can test whether assumption \ref{ass:nextbest} (next-best) holds, as $\alpha_{2}^{1}=-P(ND_{1})$ and $\alpha_{1}^{2}=-P(ND_{2})$.
\end{cor}

\subsection{How aggregation may cause violations of the exclusion restriction}

\citet{nibbering_clustered_2022} propose an algorithm which aggregates fields into clusters based on estimated first-stage coefficients. The motivation for their approach is to avoid bias from irrelevance and next-best defiers. Before discussing their approach, it is important to observe that the resulting IV estimates between such clusters will, at best, identify a positively weighted average of the causal effects of choosing one field versus a linear combination of the other fields, for example, the effects of choosing field 1 versus field 0 or 2. Hence, this approach involves moving the goalpost from clearly defined field contrasts that govern individuals' educational investments to clusters of different fields. In the discussion below, we accept at faith that such contrasts are parameters of interest.

\subsubsection{Bias From Exclusion Violation\label{sec:clustering-ass-not}}

We continue to consider the situation with three fields, discussed above. The algorithm takes as a starting point all individuals with a certain reported next-best alternative (in our case taken to be 0), and test the hypothesis that the off-diagonal coefficients, $\alpha_{2}^{1}$ and $\alpha_{1}^{2}$, are zero. If this hypothesis is rejected, the sign of the coefficient is evaluated and the treatments are clustered according to the rules laid out in Table \ref{tab:cluster-scenarios}. For example, if $\alpha_{2}^{1}$ is negative and $\alpha_{1}^{2}$ is either zero or positive, fields 0 and 2 become the control cluster and field 1 the treatment cluster. Conversely, if $\alpha_{2}^{1}$ is either zero or positive and $\alpha_{1}^{2}$ is negative, fields 0 and 1 become the control cluster and field 2 the treatment cluster.

\begin{table}
\caption{Four Possible Clustering Scenarios.\label{tab:cluster-scenarios}}
\setlength{\tabcolsep}{0pt}{\small{}}%
\begin{tabular*}{1\columnwidth}{@{\extracolsep{\fill}}lcccccccl}
\toprule
\multicolumn{1}{l}{\textbf{\small{}Scenario}} & \multicolumn{2}{c}{\textbf{\small{}Conditions}} &  & \multicolumn{3}{c}{\textbf{\small{}Clusters}} &  & \multicolumn{1}{c}{\textbf{\small{}Implied Restrictions on Defiers}}\tabularnewline
\cmidrule{2-3} \cmidrule{3-3} \cmidrule{5-7} \cmidrule{6-7} \cmidrule{7-7} \cmidrule{9-9}
 & {\small{}$\alpha_{2}^{1}$} & {\small{}$\alpha_{1}^{2}$} &  & {\small{}$S_{0}$} & {\small{}$S_{1}$} & {\small{}$S_{2}$} &  & \multicolumn{1}{c}{}\tabularnewline
\midrule
\textbf{\small{}Control} & {\small{}$<0$} & {\small{}$=0$} &  & \multirow{2}{*}{{\small{}$\{0,2\}$}} & \multirow{2}{*}{{\small{}$\{1\}$}} &  &  & {\small{}$P(ND_{1})>P(ID_{1})\geq0\land P(ID_{2})=P(ND_{2})\geq0$}\tabularnewline
\textbf{\small{}Clustering} & {\small{}$<0$} & {\small{}$>0$} &  &  &  &  &  & {\small{}$P(ND_{1})>P(ID_{1})\geq0\land P(ID_{2})>P(ND_{2})\geq0$}\tabularnewline
 & {\small{}$=0$} & {\small{}$<0$} &  & \multirow{2}{*}{{\small{}$\{0,1\}$}} & \multirow{2}{*}{{\small{}$\{2\}$}} &  &  & {\small{}$P(ID_{1})=P(ND_{1})\geq0\land P(ND_{2})>P(ID_{2})\geq0$}\tabularnewline
 & {\small{}$>0$} & {\small{}$<0$} &  &  &  &  &  & {\small{}$P(ID_{1})>P(ND_{1})\geq0\land P(ND_{2})>P(ID_{2})\geq0$}\tabularnewline
\textbf{\small{}Treatment} & {\small{}$>0$} & {\small{}$=0$} &  & \multirow{3}{*}{{\small{}$\{0\}$}} & \multirow{3}{*}{{\small{}$\{1,2\}$}} &  &  & {\small{}$P(ID_{1})>P(ND_{1})\geq0\land P(ID_{2})=P(ND_{2})\geq0$}\tabularnewline
\textbf{\small{}Clustering} & {\small{}$=0$} & {\small{}$>0$} &  &  &  &  &  & {\small{}$P(ID_{1})=P(ND_{1})\geq0\land P(ID_{2})>P(ND_{2})\geq0$}\tabularnewline
 & {\small{}$>0$} & {\small{}$>0$} &  &  &  &  &  & {\small{}$P(ID_{1})>P(ND_{1})\geq0\land P(ID_{2})>P(ND_{2})\geq0$}\tabularnewline
\textbf{\small{}No Clustering} & \multirow{1}{*}{{\small{}$=0$}} & \multirow{1}{*}{{\small{}$=0$}} &  & \multirow{1}{*}{{\small{}$\{0\}$}} & \multirow{1}{*}{{\small{}$\{1\}$}} & \multirow{1}{*}{{\small{}$\{2\}$}} &  & {\small{}$P(ID_{1})=P(ND_{1})\geq0\land P(ID_{2})=P(ND_{2})\geq0$}\tabularnewline
\multirow{2}{*}{\textbf{\small{}Undefined}{\small{}$^{*}$}} & \multirow{2}{*}{{\small{}$<0$}} & \multirow{2}{*}{{\small{}$<0$}} &  & \multirow{1}{*}{{\small{}$\{0,2\}$/}} & \multirow{1}{*}{{\small{}$\{1\}$/}} &  &  & \multirow{2}{*}{{\small{}$P(ND_{1})>P(ID_{1})\geq0\land P(ND_{2})>P(ID_{2})\geq0$}}\tabularnewline
 &  &  &  & \multirow{1}{*}{{\small{}$\{0,1\}$}} & \multirow{1}{*}{{\small{}$\{2\}$}} &  &  & \tabularnewline
\bottomrule
\end{tabular*}{\small\par}

\begin{singlespace}
\textbf{\footnotesize{}Note:}{\footnotesize{} The table shows different clusterings ensuing from the algorithm proposed by \citet{nibbering_clustered_2022} and their implied restrictions on defiers. The algorithm tests the null hypothesis of coefficients being zero. The conditions in columns two and three specify which estimates must be observed for the clustering to be chosen, where $>0$ ($<0$) indicate rejecting the null and observing a positive (negative) coefficient, while ``$=0$'' indicates not being able to reject.}{\footnotesize\par}
\end{singlespace}

{\small{}$^{*}$}{\footnotesize{}It is unclear what \citet{nibbering_clustered_2022} do when both coefficients are negative. In that case, the ordering of the coefficients will matter.}{\footnotesize\par}
\end{table}

After performing the clustering based algorithm, \citet{nibbering_clustered_2022} estimate cluster treatment effects: Let $\tilde{d}(d)=\d{d\in S_{1}}$ be the binary cluster treatment indicator and $\tilde{z}(Z)=\d{Z=d\in S_{1}}$ the cluster instrument indicator. The no clustering-scenario is equivalent to the field level. In the two other scenarios (control clustering or treatment clustering) we consider IV estimates of the equation
\[
y=\tilde{\beta}_{0}+\tilde{\beta}_{1}\tilde{d}+\varepsilon
\]
where the first stage is
\[
\tilde{d}=\pi_{0}+\pi_{1,0}\tilde{z}+\nu
\]
and $\pi_{1,0}$ is the first stage coefficient. Observed and potential outcomes and choices are linked as
\begin{align}
y & =\tilde{y}^{0}(1-\tilde{d})+\tilde{y}^{1}\tilde{d}\label{eq:outcome-cl-eq }\\
\tilde{d} & =\tilde{d}^{0}+(\tilde{d}^{1}-\tilde{d}^{0})\tilde{z}\label{eq:choice-cl-eq-1}
\end{align}
where $\tilde{d}^{j}\equiv\d{\tilde{d}^{j}=1}$ denotes the cluster-level potential treatment and $\tilde{y}^{j}$ is the potential outcome in cluster $j$. In Appendix \ref{appx:proof-violation-exclusion} we show that this IV estimand does not, under Assumptions \ref{ass:excl}--\ref{ass:mono}, have a causal interpretation as a positively weighted average of treatment effects for the cluster complier groups. This result is summarized in Proposition \ref{th:exclusion-viol}.
\begin{prop}
\label{th:exclusion-viol} Suppose Assumptions \ref{ass:excl}--\ref{ass:mono} hold.

\begin{enumerate}[label={\emph{(\alph*)}}, ref={6(\alph*)}]

\item Under control clustering, $\tilde{\beta}_{1}^{IV}$ does not have a causal interpretation as a positively weighted average of treatment effects for the cluster complier group. If the clustering is $S_{1}=\{1\}$ and $S_{0}=\{2,0\}$, we have
\begin{align*}
\tilde{\beta}_{1,0}^{IV}\enskip & =\enskip\underbrace{\frac{P(C_{1}\cup ND_{2})}{\pi_{1,0}}\E[y^{1}-y^{0}\mid C_{1}\cup ND_{2}]+\frac{P(C_{2}\cup ND_{1})}{\pi_{1,0}}\E[y^{1}-y^{2}\mid C_{2}\cup ND_{1}]}_{\substack{\text{A}}
}\\[4pt]
 & \qquad+\underbrace{\frac{P(ID_{1})}{\pi_{1,0}}}_{\substack{\tilde{\omega}_{1}}
}\enskip\enskip\underbrace{\E[y^{2}-y^{0}\mid ID_{1}]}_{\substack{\tilde{\Delta}_{1}}
}-\underbrace{\frac{P(ID_{2})}{\pi_{1,0}}}_{\substack{\tilde{\omega}_{2}}
}\enskip\enskip\underbrace{\E[y^{2}-y^{0}\mid ID_{2}]}_{\substack{\tilde{\Delta}_{1}}
}
\end{align*}
where $\pi_{1,0}=P(C_{1}\cup C_{2}\cup ND_{1}\cup ND_{2})$. \emph{A} is a positively weighted average of cluster complier LATEs, $\tilde{\omega}_{1}$ and $\tilde{\omega}_{2}$ are defier group weights, and $\tilde{\Delta}_{1}$ and $\tilde{\Delta}_{2}$ are differences in potential outcomes for irrelevance defiers in cluster $S_{0}$, i.e. never takers of the clustered treatment. The result for the clustering $S_{1}=\{2\}$ and $S_{0}=\{1,0\}$ is symmetric.

\item Under treatment clustering, $\tilde{\beta}_{1}^{IV}$ does not have a causal interpretation as a positively weighted average of treatment effects for the cluster complier group. We have
\begin{align*}
\tilde{\beta}_{1,0}^{IV}\enskip & =\enskip\underbrace{\frac{P(C_{1}\cup ID_{2})}{\pi_{1,0}}\E[y^{1}-y^{0}\mid C_{1}\cup ID_{2}]+\frac{P(C_{2}\cup ID_{1})}{\pi_{1,0}}\E[y^{2}-y^{0}\mid C_{2}\cup ID_{1}]}_{\substack{\text{A}}
}\\[4pt]
 & \qquad+\underbrace{\frac{P(ND_{1})}{\pi_{1,0}}}_{\substack{\tilde{\omega}_{3}}
}\enskip\underbrace{\E[y^{1}-y^{2}\mid ND_{1}]}_{\substack{\tilde{\Delta}_{3}}
}-\underbrace{\frac{P(ND_{2})}{\pi_{1,0}}}_{\substack{\tilde{\omega}_{4}}
}\enskip\underbrace{\E[y^{1}-y^{2}\mid ND_{2}]}_{\substack{\tilde{\Delta}_{4}}
}
\end{align*}
where $\pi_{1,0}=P(C_{1}\cup C_{2}\cup ID_{1}\cup ID_{2})$. \emph{A} is a positively weighted average of cluster complier LATEs, $\tilde{\omega}_{3}$ and $\tilde{\omega}_{4}$ are defier group weights, and $\tilde{\Delta}_{3}$ and $\tilde{\Delta}_{4}$ are differences in potential outcomes for irrelevance defiers in cluster $S_{1}$, i.e. always takers of the clustered treatment.

\end{enumerate}

\end{prop}
\begin{proof}
See Appendix \ref{appx:proof-violation-exclusion}.
\end{proof}
Imposing the irrelevance assumption under control clustering implies that the defier weights ($\tilde{\omega}_{1},\tilde{\omega}_{2}$) go to zero. In this case, $\tilde{\beta}_{1,0}^{IV}$ recovers a positively weighted average of the causal effect of choosing field 1 over 0 for compliers of instrument 1 and next-best defiers of instrument 2, and of choosing field 1 over 2 for compliers of instrument 2 and next-best defiers of instrument 1, weighted by the number of compliers and defiers. Under control clustering, this is the new parameter of interest.

Imposing the next-best assumption under treatment clustering implies that the defier weights ($\tilde{\omega}_{3},\tilde{\omega}_{4}$) go to zero. In this case, $\tilde{\beta}_{1,0}^{IV}$ recovers a positively weighted average of the causal effect of choosing field 1 over 0 for compliers of instrument 1 and irrelevance defiers of instrument 2, and of choosing field 2 over 0 for compliers of instrument 2 and irrelevance defiers of instrument 1, weighted by the number of compliers and defiers. Under treatment clustering, this is the new parameter of interest.

If neither irrelevance nor next-best assumptions hold, the IV estimand does not have a causal interpretation as a positively weighted average of treatment effects for the cluster complier group. The bias terms reflect that individuals may in response to changes in the cluster instrument be switching across fields in the treatment cluster and/or across fields in the control cluster. Such switches will generally involve changes in potential outcomes, yet no change in the cluster treatment status. Thus, the exclusion restriction at the cluster level will be violated. The reason for this bias is that the algorithm equates the sign of the off-diagonal coefficients with the presence and absence of irrelevance and nex-best defiers. As shown in Lemma \ref{th:identified-quantities}, this is wrong. The off-diagonal coefficients tell us only if there are more or less next-best defiers than irrelevance defiers. One cannot in general use the sign of $\alpha_{2}^{1}$ ($\alpha_{1}^{2}$) to show that there are no irrelevance defiers of instrument 1 (2) if $\alpha_{2}^{1}<0$ ($\alpha_{1}^{2}<0$) and no next-best defiers of instrument 1 (2) if $\alpha_{2}^{1}>0$ ($\alpha_{1}^{2}>0$).

It is also important to observe that the constant effects assumption is not sufficient for $\tilde{\beta}_{1,0}^{IV}$ to recover a positively weighted average of treatment effects between clusters 0 and 1 and obtain a causal interpretation. This result is summarized in Proposition \ref{th:constant-effects-viol}.\footnote{One exception to this negative result is the special case in which the number of defiers for each instrument happen to be equal, i.e. that $P(ID_{1})=P(ID_{2})$ under control clustering or $P(ND_{1})=P(ND_{2})$ under treatment clustering.}
\begin{prop}
\label{th:constant-effects-viol} Suppose Assumptions \ref{ass:excl}--\ref{ass:mono} hold and we further assume constant treatment effects.

\begin{enumerate}[label={\emph{(\alph*)}}, ref={6(\alph*)}]

\item Under control clustering, $\tilde{\beta}_{1}^{IV}$ does not recover the causal effect. If the clustering is $S_{1}=\{1\}$ and $S_{0}=\{2,0\}$, we have
\begin{align*}
\tilde{\beta}_{1,0}^{IV}\enskip & =\enskip\underbrace{\frac{P(C_{1}\cup ND_{2})}{\pi_{1,0}}\E[y^{1}-y^{0}]+\frac{P(C_{2}\cup ND_{1})}{\pi_{1,0}}\E[y^{1}-y^{2}]}_{\substack{\text{A}}
}\\[4pt]
 & \qquad+\underbrace{\frac{P(ID_{1})-P(ID_{2})}{\pi_{1,0}}}_{\substack{\dot{\omega}_{1}}
}\enskip\enskip\underbrace{\E[y^{2}-y^{0}]}_{\substack{\dot{\Delta}_{1}}
}
\end{align*}
where $\pi_{1,0}=P(C_{1}\cup C_{2}\cup ND_{1}\cup ND_{2})$. \emph{A} is a positively weighted average of the causal effects of choosing field 1 over 0 and of choosing field 1 over 2, $\dot{\omega}_{1}$ is a difference between defier group weights, and $\dot{\Delta}_{1}$ is the difference in potential outcomes for irrelevance defiers in cluster $S_{0}$, i.e. never takers of the clustered treatment. The result for the clustering $S_{1}=\{2\}$ and $S_{0}=\{1,0\}$ is symmetric.

\item Under treatment clustering, $\tilde{\beta}_{1}^{IV}$ does not recover the causal effect. We have
\begin{align*}
\tilde{\beta}_{1,0}^{IV}\enskip & =\enskip\underbrace{\frac{P(C_{1}\cup ID_{2})}{\pi_{1,0}}\E[y^{1}-y^{0}]+\frac{P(C_{2}\cup ID_{1})}{\pi_{1,0}}\E[y^{2}-y^{0}]}_{\substack{\text{A}}
}\\[4pt]
 & \qquad+\underbrace{\frac{P(ND_{1})-P(ND_{2})}{\pi_{1,0}}}_{\substack{\dot{\omega}_{2}}
}\enskip\underbrace{\E[y^{1}-y^{2}]}_{\substack{\dot{\Delta}_{2}}
}
\end{align*}
where $\pi_{1,0}=P(C_{1}\cup C_{2}\cup ID_{1}\cup ID_{2})$. \emph{A} is a positively weighted average of the causal effects of choosing field 1 over 0 and of choosing field 2 over 0, $\dot{\omega}_{2}$ is a difference between defier group weights, and $\dot{\Delta}_{2}$ is the difference in potential outcomes for irrelevance defiers in cluster $S_{1}$, i.e. always takers of the clustered treatment.

\end{enumerate}

\end{prop}
\begin{proof}
The constant effects assumption reduces all conditional expectations to unconditional expectations, i.e. $\E[y^{j}-y^{k}\mid G]=\E[y^{j}-y^{k}]$ for any group $G$ and any combination of fields $j,k$. The result is immediate.
\end{proof}
In contrast, the approach of \citet{kirkeboen_field_2016} recovers the causal effect under the constant effects assumption. This shows that the clustering method relies on different, not weaker assumptions than \citet{kirkeboen_field_2016}.

The following auxiliary exclusion restriction can be made to obtain identification under the clustering approach.
\begin{assumption}
Cluster Exclusion Assumptions\label{ass:within-cluster}

\begin{enumerate}[label={\emph{(\alph*)}}, ref={3(\alph*)}]

\item \label{ass:cc excl} \textbf{\emph{Control Cluster Exclusion:}} $\tilde{d}^{1}=\tilde{d}^{0}=0\implies\tilde{y}^{0,1}=\tilde{y}^{0,0}$

\item \label{ass:tc excl} \textbf{\emph{Treatment Cluster Exclusion:}} $\tilde{d}^{1}=\tilde{d}^{0}=1\implies\tilde{y}^{1,1}=\tilde{y}^{1,0}$

\end{enumerate}
\end{assumption}
Assumptions \ref{ass:cc excl} and \ref{ass:tc excl} ensure that the bias from switchers within clusters (irrelevance defiers under control clustering and next-best defiers under treatment clustering) disappear, irrespective of the number of switchers. These assumptions are homogeneity restrictions on potential outcomes across different fields, and, thus, difficult to justify. Nevertheless, if one is willing to invoke Assumptions \ref{ass:cc excl} and \ref{ass:tc excl}, one may obtain the following identification result:
\begin{prop}
\label{th:cluster-identification} Under control clustering, suppose Assumptions \ref{ass:excl}--\ref{ass:mono} and \ref{ass:cc excl} hold. $\tilde{\beta}_{1}^{IV}$ has a causal interpretation as the positively weighted average of treatment effects for cluster compliers. If the clustering is $S_{1}=\{1\}$ and $S_{0}=\{2,0\}$, we have
\begin{align*}
\tilde{\beta}_{1,0}^{IV}\enskip & =\enskip\frac{P(C_{1}\cup ND_{2})}{\pi_{1,0}}\E[y^{1}-y^{0}\mid C_{1}\cup ND_{2}]+\frac{P(C_{2}\cup ND_{1})}{\pi_{1,0}}\E[y^{1}-y^{2}\mid C_{2}\cup ND_{1}]
\end{align*}
where $\pi_{1,0}=P(C_{1}\cup C_{2}\cup ND_{1}\cup ND_{2})$. The result for clustering $S_{1}=\{2\}$ and $S_{0}=\{1,0\}$ is symmetric.

\noindent Under treatment clustering, suppose Assumptions \ref{ass:excl}--\ref{ass:mono} and \ref{ass:tc excl} hold. $\tilde{\beta}_{1}^{IV}$ has a causal interpretation as a positively weighted average of treatment effects for cluster compliers, and
\begin{align*}
\tilde{\beta}_{1,0}^{IV}\enskip & =\enskip\frac{P(C_{1}\cup ID_{1})}{\pi_{1,0}}\E[y^{1}-y^{0}\mid C_{1}\cup ID_{1}]+\frac{P(C_{2}\cup ID_{1})}{\pi_{1,0}}\E[y^{2}-y^{0}\mid C_{2}\cup ID_{1}]
\end{align*}
where $\pi_{1,0}=P(C_{1}\cup C_{2}\cup ID_{1}\cup ID_{2})$.
\end{prop}
\begin{proof}
Assumption \ref{ass:cc excl} (\ref{ass:tc excl}) eliminates the bias terms in the results from Proposition \ref{th:exclusion-viol} by letting $\tilde{\Delta}_{1},\tilde{\Delta}_{2}$ ($\tilde{\Delta}_{3},\tilde{\Delta}_{4}$) go to zero. The result is immediate.
\end{proof}

\section{Empirical analysis\label{sec:Institutions-data}}

Guided and motivated by the formal results above, we now turn to the empirical analysis of the payoffs to field of study in Norway and Denmark.

\subsection{Institutional settings}

The Danish and Norwegian post-secondary education systems are similar in many respects. Their post-secondary education sectors consist of public universities and a larger number of public and private university colleges. The vast majority of students attend a public institution, and even the private institutions are publicly funded and regulated. Universities all offer a wide selection of fields. By comparison, the university colleges rarely offer fields like Law, Medicine, Science, or Technology, but tend to offer professional degrees in fields like Engineering, Health, Business, and Teaching. Obtaining a post-secondary degree normally requires three to five years; there are no tuition fees; most students receive financial support (in the form of grants/loans) from the state.

The admission process is centralized in both countries. Applications are submitted to a central organization that handles the admission process to universities and university colleges. An applicant ranks programs (up to 15 in Norway and 8 in Denmark), each defined by a detailed field and an institution. The number of slots for each program is effectively determined by each country\textquoteright s ministry of education. For many programs, demand exceeds supply. Most slots in programs with excess demand are filled based on an application score derived from high school GPA. Offers are determined by the applicants' application score: the highest ranked applicant receives an offer for her preferred program; the second highest applicant receives an offer for her highest ranked program among the remaining programs; and so on. This is repeated until either slots run out, or applicants run out. This allocation mechanism corresponds to a so-called serial dictatorship, which is both Pareto efficient and strategy-proof \citep{svensson1999strategy} and should therefore elicit the applicants\textquoteright{} true ranking of fields at the time of application.\footnote{A possible threat to strategy-proofness is the truncation of the application list (at 15 programs in Norway and 8 in Denmark) which might induce individuals to list a safe option as their last choice. However, this is likely unimportant in practice, as less than 0.1 percent of Norwegian applicants are offered their 15th choice, and less than 1 percent of Danish applicants list eight programs.} If students want to change field or institution, they usually need to participate in next year's admission process on equal terms with other applicants.\footnote{Most programs in Denmark also have a standby (waiting) list and the GPA threshold for the standby list is typically a little lower than the main threshold. On the application form, applicants can choose whether to apply for the standby list. Applicants admitted to the standby list are guaranteed a study place the following year, but they are not considered for any of the lower-ranked programs on their application. Appendix \ref{sec:Danish-institutions} provides a more detailed discussion.}

For both countries, the exact thresholds are unpredictable at the time of application. They are not published until after the allocation process, and variation in thresholds over time is considerable. For programs with excess demand, the admission process implies that applicants scoring above a certain threshold are much more likely to receive an offer for a program they prefer compared to applicants with the same program preferences but marginally lower application score. This gives rise to credible instruments from discontinuities that effectively randomize applicants near admission cutoffs into different programs.

As explained in greater detail in \citet{kirkeboen_field_2016}, the instruments are defined around local course rankings on students' application lists. These local rankings define the \textquotedblleft preferred\textquotedblright{} and \textquotedblleft next-best\textquotedblright{} alternatives. For example, consider two fields, A and B, with A having a higher admission cutoff than B. Consider students who rank A just above B and have an application score that is either just below or just above the admission cutoff to A. These students will have A as the preferred field and B as the next-best field, no matter if A and B are ranked at the top, in the middle or at the bottom of the list. In other words, what matters for the relevance of instrument and the definition of preferred and next-best is the local ranking at which individuals are shifted in or out of a program because the application score is slightly above or below the relevant admission cutoff.

\subsection{Data and descriptive statistics}

For each country, we combine several sources of administrative data. For Norway, we use data for all applications to post-secondary education for the years 1998--2004. For Denmark, we use data for all applications to post-secondary education for the years 1994--2002. For both countries we retain the individuals\textquoteright{} first observed application and exclude those who had a post-secondary degree at the time of application. We link these applicants to the population register and other registers to obtain background information, information on completed field, and annual earnings. In our main analysis we use data on treatment (completed field) and outcome (annual earnings) eight years after application as in \citet{kirkeboen_field_2016}, and restrict the sample to those who have completed a field within eight years from application. The measure of earnings includes wage income, income from self-employment, and transfers that replace such income like short-term sickness pay and paid parental leave (but excludes unemployment benefits). Earnings are deflated using the CPI with 2011 as base year and converted to 1,000s of US dollars using the average exchange rates for the years 2010--2016 (6.5 Norwegian and 5.9 Danish crowns per US dollar).

We aggregate detailed fields into nine broad fields of study. We essentially follow the same classification of fields as in \citet{kirkeboen_field_2016}. The only difference is that Technology now covers the integrated and more vocational/professional short and long cycle degrees at university colleges and universities and consist mostly of computer science and engineering degrees. Science corresponds to more open-ended bachelor programs in different sciences such as physics, biology and mathematics, as well as agriculture, forestry and aquaculture. For the main analysis, we retain all applicants who applied to at least two broad fields, where the most preferred field has an admission cutoff. If an applicant applied to several programs within her preferred broad field we use the lowest program cutoff as the effective cutoff to the preferred field.

While the full sample of applicants is of comparable size for the two countries, the final estimation sample is smaller in Denmark than in Norway, primarily because of fewer fields with admission restrictions in Denmark. Figure \ref{fig:Distribution-of-applicants'} shows the distributions of completed field among applicants in Norway and Denmark eight years after applying. While the distributions are similar, there are some notable differences. The share of applicants completing teaching is substantially higher in Denmark than in Norway, as is the share of applicants having completed a degree with Technology. On the flip side the shares in Science, Social Science and Humanities are larger in Norway than in Denmark.\footnote{Some of the cross-country differences may be due to differences in the classification of specific fields into the nine broad fields. For instance, one reason why the share having completed Teaching is large in Denmark is that all individuals having completed a bachelor\textquoteright s degree in social education are included in Teaching regardless of the specialization (e.g., kindergarten teacher, nursery teacher, nursery nurse, child and youth worker, support worker), while some of the specializations could alternatively be classified as Other Health if they were observed as separate educations.}

\begin{figure}
\begin{centering}
\includegraphics[width=0.8\textwidth]{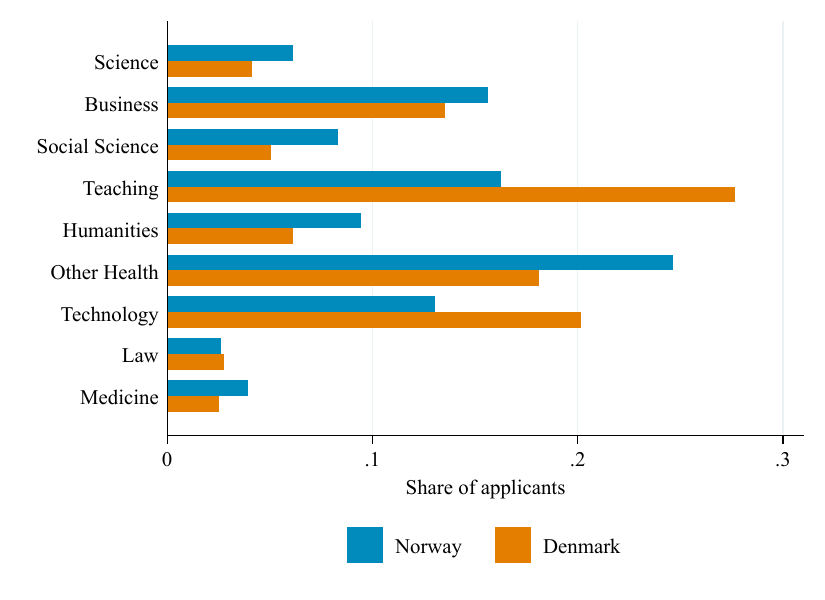}
\par\end{centering}
\begin{singlespace}
\textbf{\footnotesize{}Note:}{\footnotesize{} Figure shows number of applicants by completed field eight years after applying (conditional on having completed a field).}{\footnotesize\par}
\end{singlespace}

\caption{Distribution of applicants' completed field\label{fig:Distribution-of-applicants'}}
\end{figure}

\begin{figure}
\subfloat[GPA]{\includegraphics[width=0.5\textwidth]{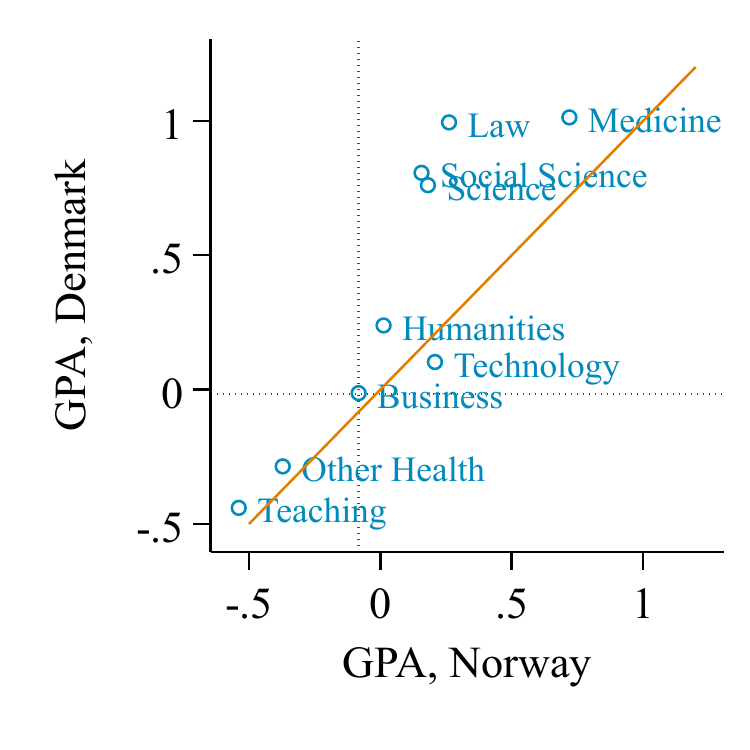}

}\subfloat[Earnings]{\includegraphics[width=0.5\textwidth]{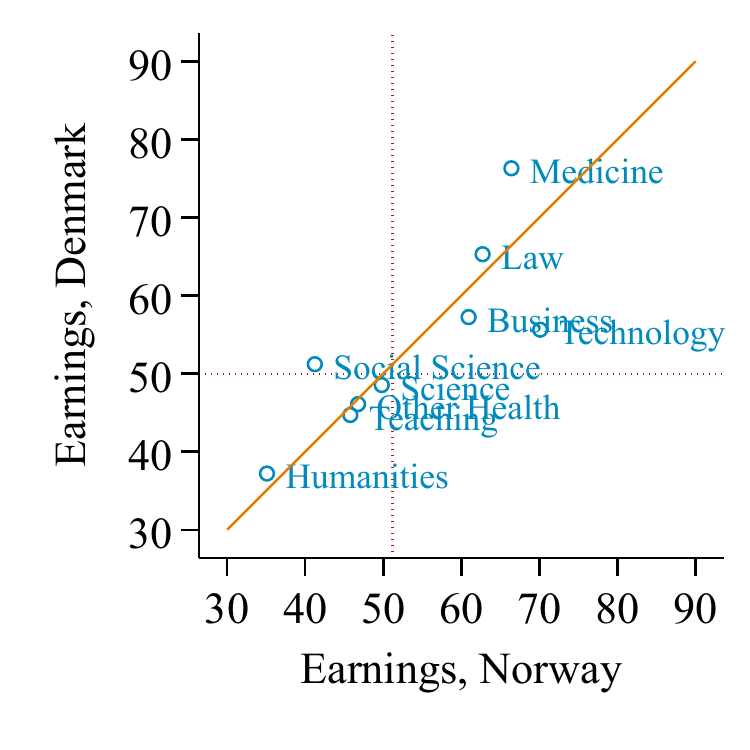}

}

\begin{singlespace}
\textbf{\footnotesize{}Note:}{\footnotesize{} Figure shows applicant-weighted average GPA and earnings. GPA is demeaned and standardized within country. Earnings are CPI adjusted and converted to USD using fixed exchange rates (see text for details) and observed eight years after applying.}{\footnotesize\par}
\end{singlespace}

\caption{Applicants' GPA and earnings in Norway and Denmark, average by country and completed field of study\label{fig:GPA-and-earnings}}
\end{figure}

As an indicator of relative selectivity, we standardize high school GPA within country and show in sub-graph (a) of Figure \ref{fig:GPA-and-earnings} the average standardized GPA by field and country. In both countries average GPA is relatively low in Teaching and Other Health. Average GPA is very high for Medicine in both countries, but Law, Social Science and Science are nearly equally selective as Medicine in Denmark.

Sub-graph (b) of Figure \ref{fig:GPA-and-earnings} compares earnings by field across country. Average earnings levels, indicated by the red dotted lines, are very similar in the two countries. Earnings in Medicine and Social Science are higher in Denmark consistent with their higher selectivity, but the same is not observed for Law and Science. Earnings are higher in Norway for Technology. In both countries, earnings are particularly low for Humanities. However, it is important to note that earnings are measured eight years after application, which is very early in the career, especially for those choosing longer programs or programs characterized by a more difficult school-to-work transition. We examine the importance of this issue in a specification check that uses earnings measured later in the working life as the outcome variable.

In Appendix Figures \ref{fig:Distribution-of-completed} and \ref{fig:GPA-and-earnings-1} we present results similar to Figures \ref{fig:Distribution-of-applicants'} and \ref{fig:GPA-and-earnings}, but not restricted to applicants. For Norway, data for Figures \ref{fig:Distribution-of-completed} and \ref{fig:GPA-and-earnings-1} consist of everybody born 1979--1983 such that we have application data for the years they are aged 19--21. Similarly, for Denmark the population sample consists of the cohorts born 1975--1981. Completed field and earnings are measured at age 28. The results for these broader populations are similar to the results for applicants in Figures \ref{fig:Distribution-of-applicants'} and \ref{fig:GPA-and-earnings}.

\section{How we estimate and compare payoffs\label{sec:Empirical-setup}}

\subsection{2SLS specification }

The identification results in Sections \ref{sec:IV-in-unordered} and \ref{sec:derivation-bias} motivate and guide the specification of the empirical model. We consider the following system of equations separately for individuals with next-best field $l$ (in the local field ranking):
\begin{align}
y & =\sum_{j\neq l}\beta_{jl}d_{j}+x'\gamma_{l}+\lambda_{l}^{k}+\epsilon_{l}\label{eq:second stage-1}\\
d_{j} & =\sum_{k\neq l}\alpha_{jl}^{k}z_{k}+x'\psi_{jl}+\eta_{jl}^{k}+u_{jl}\label{eq:firststages-1}
\end{align}
where (\ref{eq:second stage-1}) is the second stage equation, and (\ref{eq:firststages-1}) are the first-stage equations, one for each field. In these equations, $j$ denotes the completed field, $l$ denotes the stated next best alternative (in the local field ranking), and $k$ denotes the preferred field (in the local field ranking). The $l$ index is necessary in these equations, since we are now considering all possible preferred and next best fields (and not only focusing on field 0 as the stated next best alternative, as we did in the simple example in equations (\ref{eq:outcome-eq})-(\ref{eq:choice-eq-0})).

The instruments $z_{k}$ in (\ref{eq:firststages-1}) are the predicted offers for field $k$, and $z_{k}$ is therefore equal to one if $k$ is the individual's preferred field and her application score exceeds the admission cutoff for field $k$ and zero otherwise. We therefore have as many binary instruments as treatments (one $z_{k}$ for each completed field dummy $d_{j}$), and for a given individual at most one of the instruments $z_{k}$ can equal 1 (namely the one of her preferred field in the local field ranking).

Our estimation approach exploits the fuzzy regression discontinuity design implicit in the admission process described above, where individuals with application scores above the cutoff are more likely to receive an offer for their preferred field. Although the identification in this setup is ultimately local, we use 2SLS because our sample sizes do not allow for local non-parametric estimation. While the model laid out above abstracted from any control variables, we now need to include certain covariates to ensure the exogeneity of our instruments.

First, all equations include controls for the running variable. While our baseline specification controls for the application score linearly on each side of the admission cutoff, \citet{kirkeboen_field_2016} reported results from several specification checks, all of which support our main findings. Second, we control for individuals' preferences by adding fixed effects for preferring field $k$ and having $l$ as the next-best field (in the local field ranking): $\lambda_{l}^{k}$ and $\eta_{jl}^{k}$. To gain precision, we estimate the system of equations (\ref{eq:second stage-1})--(\ref{eq:firststages-1}) jointly for all completed and next-best fields, allowing for separate intercepts for preferred field and for next-best field by completed field (i.e. $\lambda_{l}^{k}=\mu^{k}+\theta_{j}$ and $\eta_{jl}^{k}=\tau_{j}^{k}+\sigma_{j}^{k}$). In a robustness check, \citet{kirkeboen_field_2016} show that their estimates are robust to allowing for separate intercepts for every interaction between preferred and next-best field. Finally, to reduce residual variance we also add controls for gender, cohort and age at application, which are pre-determined.

From the resulting 2SLS estimation of equations (\ref{eq:second stage-1})--(\ref{eq:firststages-1}) across all next-best fields, we obtain a matrix of the payoffs to field $j$ compared to $k$ for those who prefer $j$ and have $k$ as next-best field. In our baseline specification of the fields, we have 9 completed fields ($j$), 9 possible preferred fields/instruments ($k$), 8 possible next-best fields ($l$).\footnote{In both countries, the number of applicants with Medicine as next-best is very small and these are therefore omitted in our analysis. Thus, there are 9 preferred fields but only 8 next-best fields.} Because preferred field can never be the same as the next-best alternative, we get 576 (and not 648) unique first stage coefficients, $\alpha_{jl}^{k}$. Because $\sum d_{j}=1$ for each applicant, creating a within-applicant correlation between different $d_{j}$, we allow the residuals $u_{jl}$ to be clustered within applicant.

\subsection{Comparing payoff estimates\label{subsec:Comparing-estimates}}

We want to compare payoffs to field of study across two different populations:
\begin{equation}
(\beta_{jl}^{DK}-\overline{\beta}^{NO})=a_{0}+a_{1}(\beta_{jl}^{NO}-\overline{\beta}^{NO})+e_{jl}\label{eq:OLS-compare-pop}
\end{equation}
where we have re-centered the payoffs relative to the average Norwegian payoffs for interpretational convenience: it allows us to interpret the intercept $a_{0}$ as the payoff difference between Denmark and Norway at the average Norwegian payoff. The interpretation of the slope $a_{1}$ -- which quantifies the average increase in the Danish payoffs for a one unit increase in the Norwegian payoffs -- is unaffected by the centering.

There are two considerations that we need to pay attention to when taking equation (\ref{eq:OLS-compare-pop}) to the data: measurement error and across-population comparison. Unweighted estimation of (\ref{eq:OLS-compare-pop}) would assume that the estimated returns are from populations of similar size. In practice, the return estimates in the two countries will have differently sized groups, where some estimates are based on many applicants shifted by the instrument (when there are many applicants with given preferred and next-best fields and the first stage is large), while others are based on few applicants shifted (when there are less applicants in the preferred/next-best field cell or the first stage is close to zero).

To take these unequal underlying population sizes into account we will weigh our regressions with a measure of the number of applicants that are shifted. For each payoff estimate $\beta_{jl}^{c}$ in country $c$ we calculate the net number of applicants that are shifted on that margin as follows
\[
n_{jl}^{k,c}=|\alpha_{jl}^{k,c}|\cdot N_{kl}^{c}\cdot\bar{z}_{kl}^{c}
\]
where $\alpha_{jl}^{k}$ is the first-stage coefficient, $N_{kl}$ the number of applicants with preferred field $k$ and next-best field $l$, and $\bar{z}$ the share of these applicants above the cutoff. We then construct weights\footnote{It should be noted that in practice using these weights gives very similar results to using population weights $N_{kl}^{NO}+N_{kl}^{DK}$.}\textsuperscript{,}\footnote{When we study distributions of first-stage coefficients we will also  use the weights $w_{jl}$.}
\[
w_{jl}=\sum_{k}(n_{jl}^{k,NO}+n_{jl}^{k,DK})
\]

Measurement error concerns arise because rather than relating population payoffs as in (\ref{eq:OLS-compare-pop}) we will be comparing two sets of noisily estimated population payoffs:
\begin{equation}
(\hat{\beta}_{jl}^{DK}-\overline{\hat{\beta}}^{NO})=a_{0}+a_{1}(\hat{\beta}_{jl}^{NO}-\overline{\hat{\beta}}^{NO})+\tilde{e}_{jl}\label{eq:OLS-compare}
\end{equation}
It is well know that measurement error in explanatory variables results in estimation bias. Assuming classical measurement error $\hat{\beta}_{jl}^{c}=\beta_{jl}^{c}+\epsilon_{jl}^{c}$ with $\epsilon_{jl}^{c}$ i.i.d. and $\sigma_{\epsilon,c}^{2}\equiv var(\epsilon_{jl}^{c})$, we can quantify the bias as follows\footnote{Classical measurement error in the dependent variable affects the precision but not the consistency of the regression estimates.}
\[
\hat{a}_{1}=\frac{cov(\hat{\beta}_{jl}^{DK},\hat{\beta}_{jl}^{NO})}{var(\hat{\beta}_{jl}^{NO})}\rightarrow a_{1}\frac{var(\beta_{jl}^{NO})}{var(\beta_{jl}^{NO})+var(\epsilon_{jl}^{NO})}=a_{1}R_{NO}
\]
where the estimate of $a_{1}$ is attenuated by a factor $R_{NO}=1-\sigma_{\epsilon,NO}^{2}/\sigma_{\hat{\beta},NO}^{2}$ (with $\sigma_{\hat{\beta},NO}^{2}\equiv var(\hat{\beta}_{jl}^{NO})$). $R_{NO}$ quantifies the reliability of $\hat{\beta}_{jl}^{NO}$ and, provided we can estimate it, implies that we can adjust $\hat{a}_{1}$ by $1/\hat{R}_{NO}$ to recover an unbiased estimate of the true $a_{1}$.\footnote{We use the Stata command -eivreg- to perform the error-in-variable regression.} We construct an estimate of $R_{NO}$ by plugging in the variance of the payoff estimates as an estimate of $\sigma_{\hat{\beta},c}^{2}$, and using the average squared standard errors of the payoffs as an estimate of $\sigma_{\epsilon,c}^{2}$.\footnote{\citet{sullivan2001note} shows that this approach is robust to measurement error heteroskedasticity.} Finally, we can use the so-called total reliability $R_{Total}=\sqrt{R_{NO}\cdot R_{DK}}$ to construct an estimate of the correlation of the payoffs across the two countries
\[
\hat{\rho}=\rho(\hat{\beta}_{jl}^{DK},\hat{\beta}_{jl}^{NO})/\hat{R}_{Total}\rightarrow\rho\equiv\rho(\beta_{jl}^{DK},\beta_{jl}^{NO})
\]

Table \ref{tab:Descriptive-statistics-and} reports the standard deviation of the estimated payoffs, the square root of their average standard errors squared, as well as the resulting estimated reliability ratios. The first two columns report the unweighted estimates. We see that the payoff estimates vary more in Norway than in Denmark and are on average also more noisily estimated. These unweighted estimates do however not map into a population. The analysis in this paper will therefore investigate weighted results and the next two columns report the weighted reliability estimates. For shifted applicants the variability in the estimates and the average standard error is reduced, especially for the Norwegian estimates. The estimated reliability of the Norwegian payoff estimates is 0.86 compared to 0.72 for the Danish ones. Reliability is therefore high for both countries.\footnote{Using country-specific weights gives slightly higher but very similar estimates, namely a reliability of 0.89 for Norway and 0.78 for Denmark.}

\begin{table}
\begin{centering}
\caption{Descriptive statistics and reliabilities, Norwegian and Danish payoff estimates\label{tab:Descriptive-statistics-and}}
{\small{}}%
\begin{tabular*}{1\columnwidth}{@{\extracolsep{\fill}}lrrrrr}
\toprule
 & \multicolumn{2}{c}{{\small{}Unweighted}} &  & \multicolumn{2}{c}{{\small{}Weighted}}\tabularnewline
\cmidrule{2-3} \cmidrule{3-3} \cmidrule{5-6} \cmidrule{6-6}
 & {\small{}Norway} & {\small{}Denmark} &  & {\small{}Norway} & {\small{}Denmark}\tabularnewline
\midrule
{\small{}SD of payoff estimates $\hat{\beta}_{jl}^{c}$ ($\hat{\sigma}_{\hat{\beta},c}$)} & {\small{}40.4} & 20.0 &  & {\small{}31.6} & {\small{}18.6}\tabularnewline
{\small{}Square root of average $SE(\hat{\beta}_{jl}^{c})^{2}$ ($\hat{\sigma}_{\epsilon,c}$)} & {\small{}29.0} & {\small{}11.2} &  & {\small{}11.8} & {\small{}9.9}\tabularnewline
 &  &  &  &  & \tabularnewline
{\small{}Reliability ($R_{c}=1-\hat{\sigma}_{\epsilon,c}^{2}/\hat{\sigma}_{\hat{\beta},c}^{2}$)} & {\small{}0.48} & {\small{}0.71} &  & {\small{}0.86} & {\small{}0.72}\tabularnewline
 &  &  &  &  & \tabularnewline
{\small{}Total reliability ($R_{Total}=\sqrt{R_{NO}\cdot R_{DK}}$)} & \multicolumn{2}{c}{{\small{}0.58}} &  & \multicolumn{2}{c}{{\small{}0.79}}\tabularnewline
\bottomrule
\end{tabular*}{\small\par}
\par\end{centering}
\noindent\begin{minipage}[t]{1\columnwidth}%
\textbf{\footnotesize{}Note:}{\footnotesize{} See section \ref{subsec:Comparing-estimates} for the definition of the first-stage weights $\omega_{jl}$.}%
\end{minipage}
\end{table}

\section{Payoffs to fields of study\label{sec:Pay-offs}}

\subsection{Examining the violation of next-best and irrelevance}

We start by examining whether we can statistically reject the irrelevance and/or next best assumptions. As shown above, these assumptions are rejected if any of the off-diagonal first-stage coefficients are significantly different from zero. Joint tests strongly reject this null-hypothesis for both countries (cf. Tables \ref{tab:Joint-and-individual} and \ref{tab:Joint-and-individual-dk}). This implies that there are irrelevance-violators (detected by positive off-diagonal first-stage coefficients) and/or next-best-violators (detected by negative off-diagonal first-stage coefficients). We tend to detect such violation for most field of studies.

As a first indication of the relative importance of next-best vs. irrelevance violations we consider the signs of the off-diagonal coefficients that are individually significant (Tables \ref{tab:Off-diagonal-first-stages} and \ref{tab:Off-diagonal-first-stages-dk}). For Norway this reveals that few if any of the positive coefficients are individually significant, especially after adjusting for multiple testing (using the Bonferroni correction). However, a large number of the negative off-diagonal coefficients are significant, also after adjusting for multiple testing. For Norway, we therefore mostly find evidence for violations of next-best. This stands in contrast to the results for the Danish data, which are consistent with violations of the irrelevance and next-best assumptions being approximately equally frequent.

\begin{figure}
\subfloat[Norway]{\includegraphics[width=0.5\textwidth]{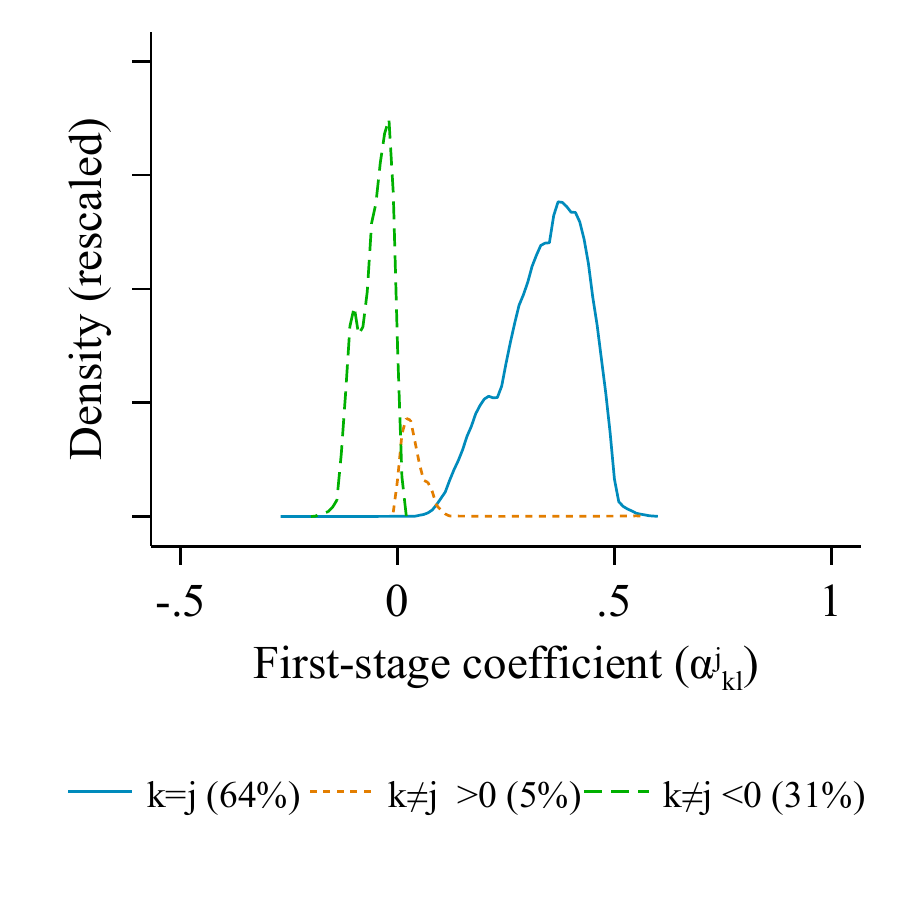}

}\subfloat[Denmark]{\includegraphics[width=0.5\textwidth]{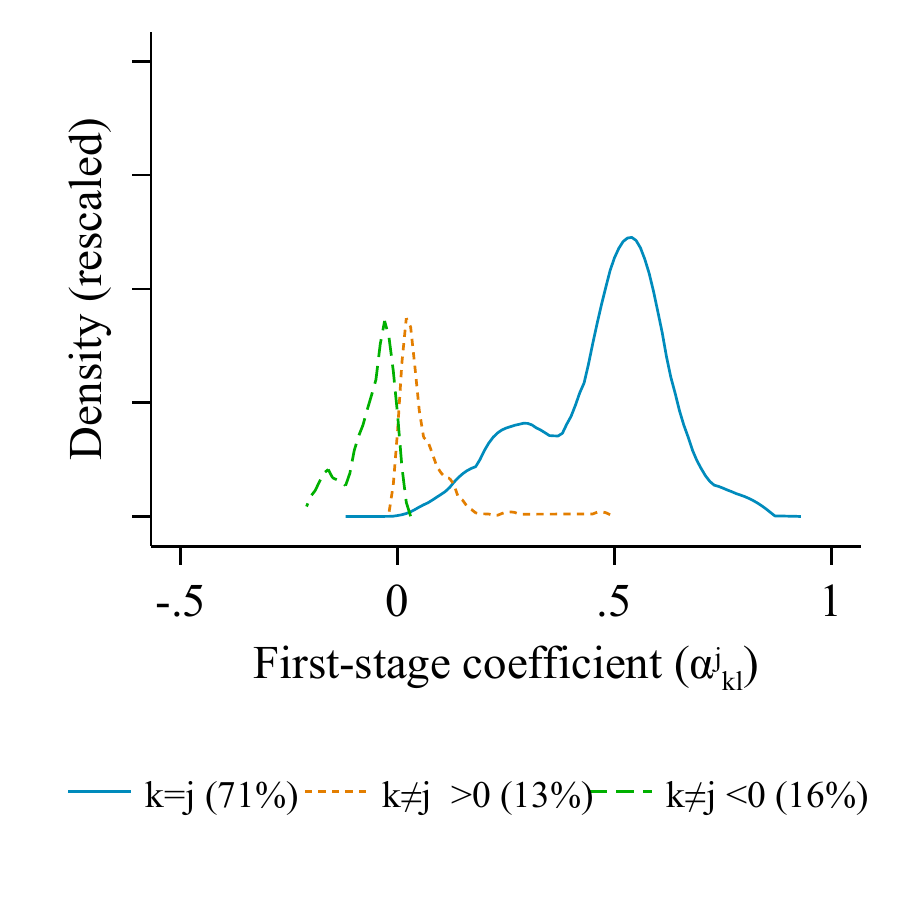}

}

\noindent\begin{minipage}[t]{1\columnwidth}%
\begin{singlespace}
\textbf{\footnotesize{}Note:}{\footnotesize{} Weighted densities that sum together to 1. First-stage coefficients $\alpha_{jl}^{k}$ where $k=j$ are ``on-diagonal''. Those with $l\neq j$ are ``off-diagonal\textquotedbl , and can be either positive (\textquotedbl >0\textquotedbl ) or negative (\textquotedbl <0\textquotedbl ). Percentage shares are indicated in parentheses. }
\end{singlespace}
\end{minipage}

\caption{Distribution of first-stage coefficients\label{fig:Distributions-FS}}
\end{figure}

With enough data any model can be rejected, no matter how minor the misspecification. We therefore gauge the empirical relevance of the violations of the irrelevance and next-best alternative conditions by quantifying the relative size of the associated applicant groups. The results are reported in Figure \ref{fig:Distributions-FS}, which shows the distribution of the relevant first-stage coefficients weighted with the number of applicants shifted and where the densities are rescaled so that they to sum to unity. The mass under each density -- reported in parenthesis in the Figure -- quantifies the relative size of the complier/defier group in question.

The left-panel of Figure \ref{fig:Distributions-FS} shows that at least 64\% of the shifted applicants in Norway are shifted at the expected (on-diagonal) margin. Of the remaining shifted applicants nearly 90\% are shifted on margins with negative coefficients. For Norway we therefore continue to find evidence for violations of next-best but not irrelevance when we take the size of the shifted applicant groups into account. The results for Denmark in the right-panel of Figure \ref{fig:Distributions-FS} show that a similar share of applicants is shifted at the diagonal. Off-diagonal the shifted applicants are however evenly distributed between positive and negative margins. This reinforces the earlier conclusion, suggesting that violations of the irrelevance and next-best assumptions are approximately equally frequent. It should be emphasized however that, depending on their sign, the (absolute values of the) off-diagonal first-stage coefficients give a lower bound on each type of violator, while the on-diagonal coefficients provide upper bounds on the compliers. We therefore conclude that in both countries the violations of irrelevance or next-best are quantitatively non-trivial, appear to be of similar magnitude, but of a different nature.

\subsection{Comparison of payoffs}

Figure \ref{fig:Distribution-of-pay-offs} reports the reliability-corrected and weighted densities for the Norwegian and Danish estimates of the payoffs of completing a field-of-study instead of the next-best.\footnote{Appendix tables \ref{tab:Payoffs No t8} and \ref{tab:Payoffs Dk t8} report the payoff estimates.} In each country, the payoffs are measured in terms of annual earnings eight years after application. On average the annual payoff in Denmark is about 2,200 USD, while in Norway the returns are substantially larger at about 22,000 USD. In addition, there is also more variation in the payoffs in Norway compared to Denmark. A joint test of equality of the payoffs across countries gives a $\chi_{64}^{2}$ statistic of 441.7 with a corresponding \emph{p}-value smaller than 0.0001. We therefore strongly reject that the payoffs are the same. In the following we investigate these differences in more detail.

\begin{figure}
\begin{centering}
\includegraphics[width=0.8\textwidth]{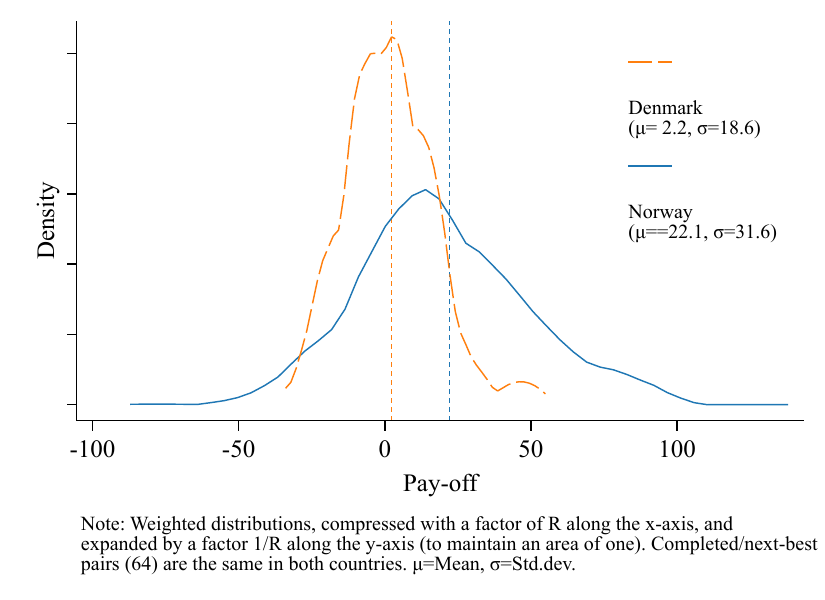}
\par\end{centering}
\caption{Distribution of payoffs by country\label{fig:Distribution-of-pay-offs}}
\end{figure}

\begin{figure}
\begin{centering}
\includegraphics[width=0.5\textwidth]{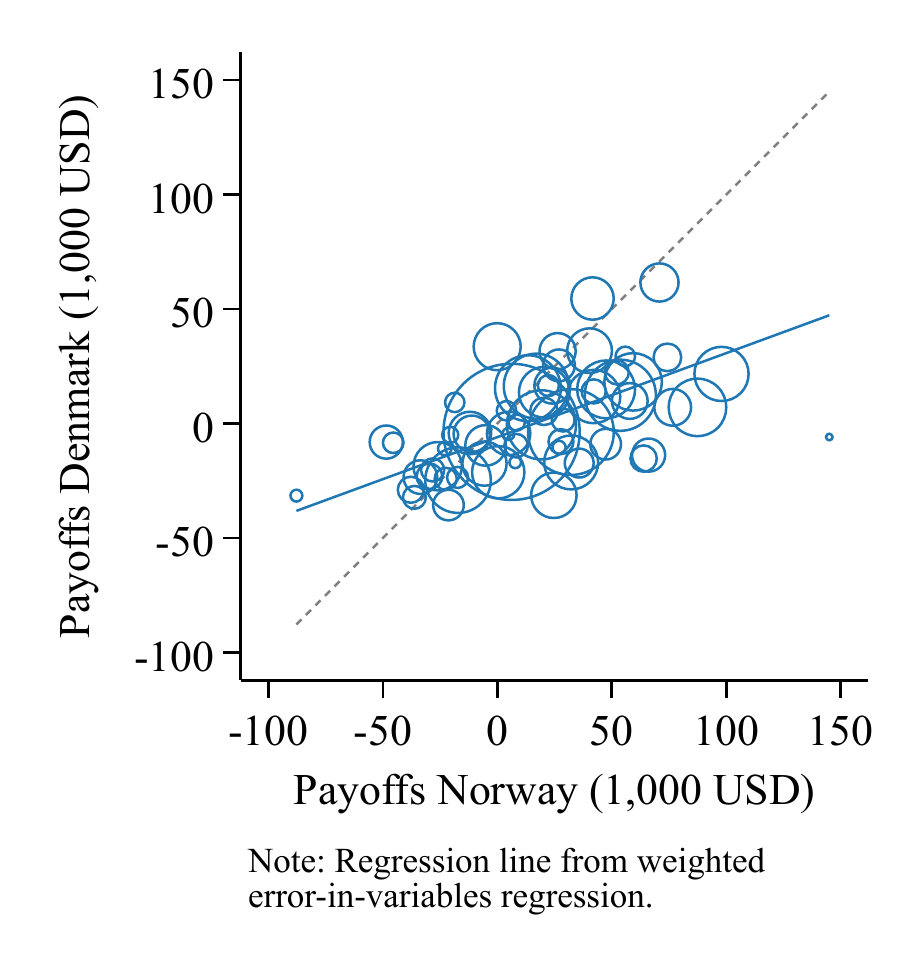}
\par\end{centering}
\caption{Payoffs in Denmark and Norway, all completed and next-best fields\label{fig:Payoffs-in-Norway}}
\end{figure}

Figure \ref{fig:Payoffs-in-Norway} starts out with comparing the Norwegian and Danish payoff estimates directly. It plots the estimates in the two countries against each other, with the size of the marker being proportional to the size of the sum of the Norwegian and Danish shifted applicant groups and, in addition to the 45-degree line, the figure also shows the regression line from the following error-in-variables regression (\ref{eq:OLS-compare}) described in section \ref{subsec:Comparing-estimates} above:

\begin{equation}
(\hat{\beta}_{jl}^{DK}-\overline{\hat{\beta}}^{NO})=a_{0}+a_{1}(\hat{\beta}_{jl}^{NO}-\overline{\hat{\beta}}^{NO})+\tilde{e}_{jl}\label{eq:OLS-compare-1}
\end{equation}
Changes in the intercept $a_{0}$ as we omit estimates with evidence of defiance shows whether the average Danish payoff become more aligned with the average Norwegian payoffs. We also report changes in the slope $a_{1}$ and the estimated correlation between the Danish and Norwegian payoffs $\rho$.

Figure \ref{fig:Regression-coefficients-and} reports the results of this exercise and shows that, consistent with the low average and lower spread of the Danish estimates in Figure \ref{fig:Distribution-of-pay-offs}, the Danish estimates increase less than one-to-one with the Norwegian estimates with an estimated slope of 0.38 (s.e. 0.07), and are on average substantially lower (the estimated payoff difference is -19.9 with a s.e. of 2.1). However, even though their levels are different, we find that the payoffs exhibit a relatively strong positive correlation of 0.65 after adjusting for measurement error.

Above we found evidence of violations of the irrelevance and next-best assumptions for both two countries. Can these violations explain the observed differences in the estimated payoffs? We investigate this question by successively removing preferred-next-best combinations with a high share of detected defiers, thus reducing the share of defiers in the sample and see how this impacts the relationships between the Norwegian and Danish payoff estimates.

We first compute, for each completed and next-best field in both countries, the share of applicants that are shifted by off-diagonal instruments. This quantifies the net-flow of irrelevance and next-best defiers at that particular margin. We then progressively drop the estimates with the largest shares of net-defiance and estimate the weighted error-in-variables regression on the resulting sub-sample of Danish payoffs on Norwegian payoffs.

\begin{figure}
\subfloat[Changes in $a_{0}$, \# estimates and shifted applicants]{\includegraphics[width=0.5\textwidth]{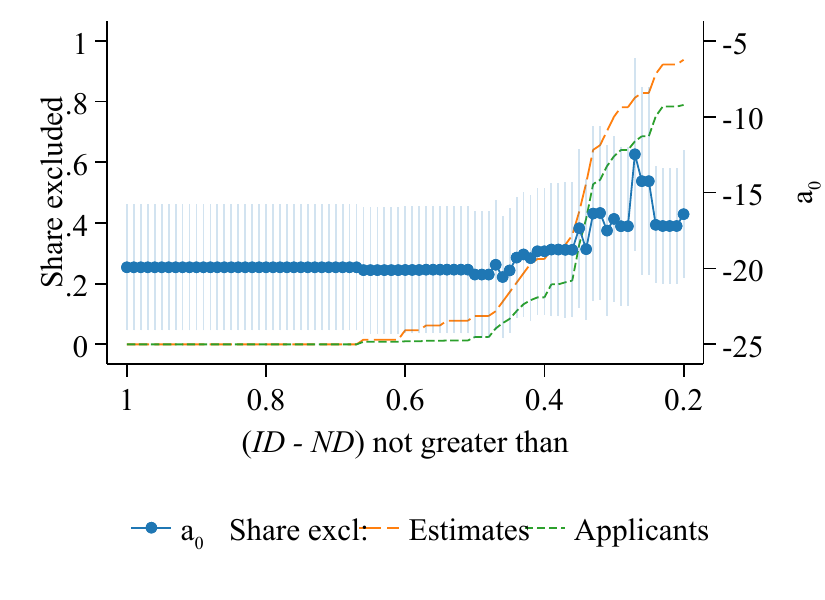}

}\subfloat[$a_{0}$ vs. exclusion of violators]{\includegraphics[width=0.5\textwidth]{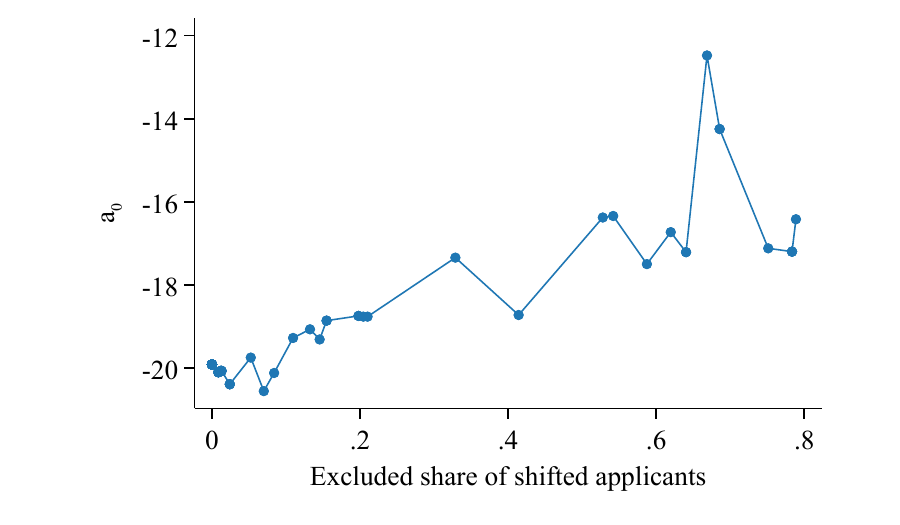}

}

\subfloat[Changes in $a_{1}$, \# estimates and shifted applicants]{\includegraphics[width=0.5\textwidth]{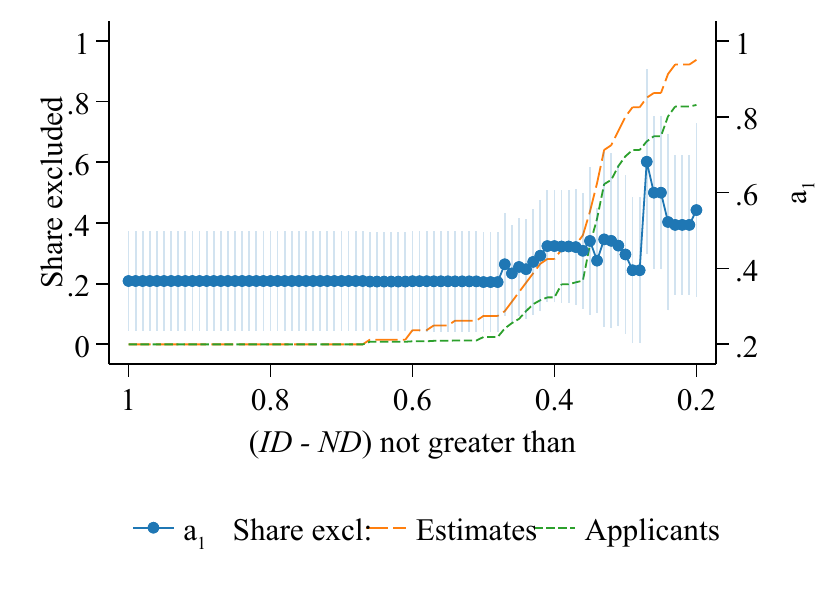}

}\subfloat[$a_{1}$ vs. exclusion of violators]{\includegraphics[width=0.5\textwidth]{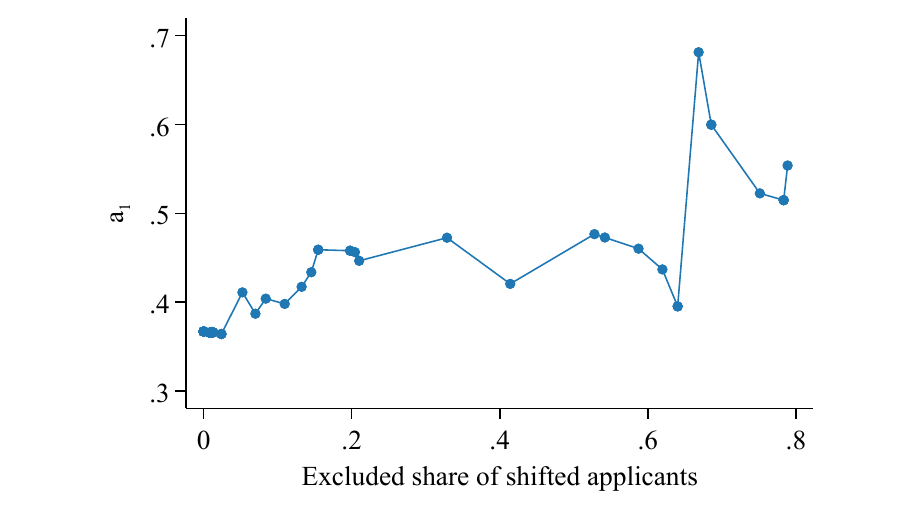}

}

\subfloat[Changes in payoff correlation $\rho$, \# estimates and shifted applicants]{\includegraphics[width=0.5\textwidth]{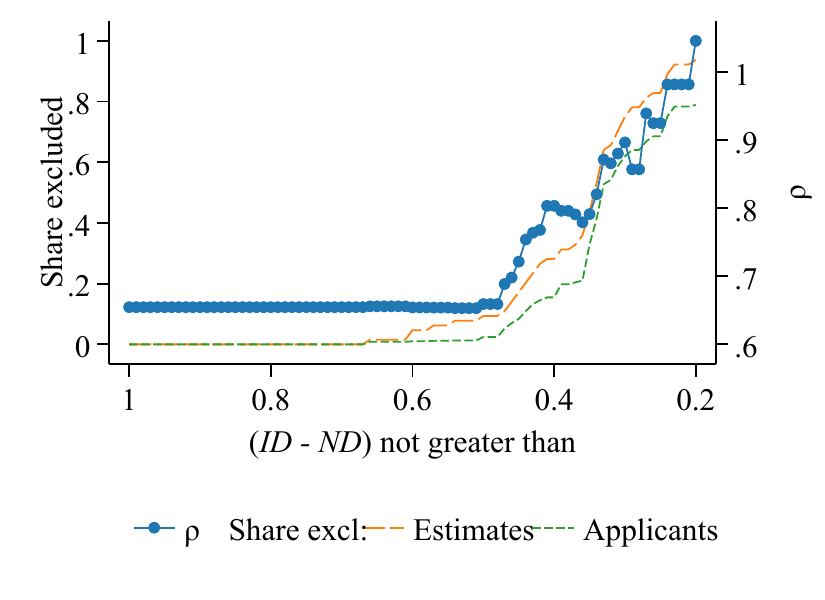}

}\subfloat[Correlation vs. exclusion of violators]{\includegraphics[width=0.5\textwidth]{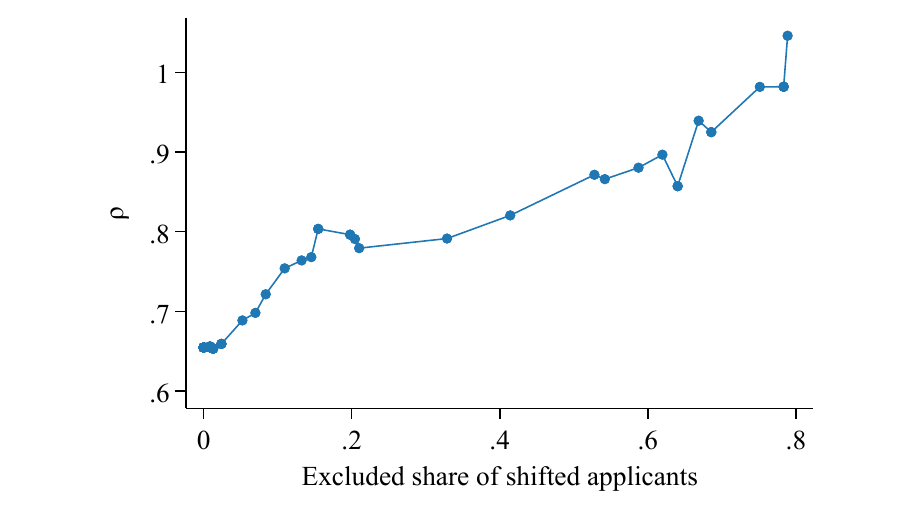}

}

\noindent\begin{minipage}[t]{1\columnwidth}%
\begin{singlespace}
{\footnotesize{}Note: Graphs show (a, b) constant term from a weighted error-in-variable-regression of Danish payoffs on Norwegian payoffs, (c, d) coefficient on Norwegian payoffs from the same regression, and (e, f) the weighted reliability-adjusted correlation between Danish and Norwegian payoff estimates. }
\end{singlespace}
\end{minipage}

\caption{Regression coefficients and correlations as a function of defiers excluded\label{fig:Regression-coefficients-and} }
\end{figure}

Sub-graph (a) shows the estimated intercepts and their confidence intervals. The x-axis in sub-graph (a) shows the maximum share of net-defiance allowed in the sample of estimates for the two countries. Reducing this share one percentage point at a time we re-estimate (\ref{eq:OLS-compare-1}). The first estimate is dropped at about 66 percent net-defiance. Then progressively more payoff estimates are excluded as we restrict the maximum share of defiers below 50 percent. We see that $a_{0}$ stays approximately constant close to -20 until we restrict the share of defiers to be below 50 percent. After this $a_{0}$ gradually increases, and reaches -17 when we restrict the sample to max 20 percent defiers.

Sub-graph (a) also shows the shares of the 64 payoff estimates and of the shifted applicants that are excluded. The pairs of completed/next-best fields that have the highest shares of net-defiers have relatively few applicants shifted on the diagonal. Restricting the maximum to 50 percent we exclude 9 percent of estimates and 2 percent of compliers. Restricting further has a stronger impact on estimates and shifted applicants retained, and when we ultimately restrict the sample to max 20 percent defiers only 4 out of 64 estimates and 21 percent of the shifted applicants are retained.

In sub-graph (b) we plot $a_{0}$ against the share of shifted applicants that are excluded. As a function of applicants excluded, $a_{0}$ rises about linearly. However, as can be seen from the confidence bands in sub-graph (a), the estimated intercepts for different samples are never significantly different.

In sub-graphs (c) and (d) we show similar results for the slope parameter $a_{1}$ from (\ref{eq:OLS-compare-1}). While $a_{1}$ increases somewhat in the beginning, it is mostly stable across the different samples. Finally, in sub-graphs (e) and (f) we show the reliability-adjusted weighted coefficient of correlation. This increases steadily with the share of compliers excluded, from 0.65 in the full sample to 1 when restricting to less than 20 percent net-defiers.

While we found above that the Norwegian and Danish payoff estimates are strongly correlated, this correlation substantially increases further when we exclude the estimates with more evidence of defiance of irrelevance and next-best. The intercept and slope from the regression (\ref{eq:OLS-compare-1}) are however relatively stable, suggesting that violations of irrelevance and next-best do not explain the lower level and variation of the payoffs in Denmark compared to Norway.

\subsection{Other explanations for differences in payoffs across the countries}

\begin{table}
\caption{Explaining payoff differences between Denmark and Norway\label{tab:other-explanation}}

\setlength{\tabcolsep}{5pt}{\small{}}%
\begin{tabular*}{1\textwidth}{@{\extracolsep{\fill}}lr@{\extracolsep{0pt}.}lr@{\extracolsep{0pt}.}lr@{\extracolsep{0pt}.}lr@{\extracolsep{0pt}.}lr@{\extracolsep{0pt}.}lr@{\extracolsep{0pt}.}lr@{\extracolsep{0pt}.}lr@{\extracolsep{0pt}.}lr@{\extracolsep{0pt}.}l}
\toprule
 & \multicolumn{12}{c}{Earnings at} & \multicolumn{2}{c}{} & \multicolumn{4}{c}{Earnings at}\tabularnewline
 & \multicolumn{12}{c}{$t=8$} & \multicolumn{2}{c}{} & \multicolumn{4}{c}{$t=13$}\tabularnewline
\cmidrule{2-19} \cmidrule{4-19} \cmidrule{6-19} \cmidrule{8-19} \cmidrule{10-19} \cmidrule{12-19} \cmidrule{14-19} \cmidrule{16-19} \cmidrule{18-19}
 & \multicolumn{2}{c}{{\small{}(1)}} & \multicolumn{2}{c}{{\small{}(2)}} & \multicolumn{2}{c}{{\small{}(3)}} & \multicolumn{2}{c}{{\small{}(4)}} & \multicolumn{2}{c}{{\small{}(5)}} & \multicolumn{2}{c}{{\small{}(6)}} & \multicolumn{2}{c}{} & \multicolumn{2}{c}{{\small{}(7)}} & \multicolumn{2}{c}{{\small{}(8)}}\tabularnewline
\midrule
{\small{}$a_{0}$} & {\small{}-19}&{\small{}92} & {\small{}-19}&{\small{}92} & {\small{}-19}&{\small{}92} & {\small{}-15}&{\small{}62} & {\small{}-20}&{\small{}39} & {\small{}-15}&{\small{}70} & \multicolumn{2}{c}{} & {\small{}-10}&{\small{}25} & {\small{}-4}&{\small{}31}\tabularnewline
 & {\small{}(2}&{\small{}12)} & {\small{}(1}&{\small{}56)} & {\small{}(1}&{\small{}71)} & {\small{}(3}&{\small{}15)} & {\small{}(2}&{\small{}01)} & {\small{}(3}&{\small{}15)} & \multicolumn{2}{c}{} & {\small{}(3}&{\small{}26)} & {\small{}(5}&{\small{}52)}\tabularnewline
{\small{}$a_{1}$} & {\small{}0}&{\small{}37} & {\small{}0}&{\small{}29} & {\small{}0}&{\small{}70} & {\small{}0}&{\small{}56} & {\small{}0}&{\small{}41} & {\small{}0}&{\small{}56} & \multicolumn{2}{c}{} & {\small{}0}&{\small{}70} & {\small{}0}&{\small{}91}\tabularnewline
 & {\small{}(0}&{\small{}07)} & {\small{}(0}&{\small{}07)} & {\small{}(0}&{\small{}13)} & {\small{}(0}&{\small{}10)} & {\small{}(0}&{\small{}08)} & {\small{}(0}&{\small{}10)} & \multicolumn{2}{c}{} & {\small{}(0}&{\small{}15)} & {\small{}(0}&{\small{}17)}\tabularnewline
 & \multicolumn{2}{c}{} & \multicolumn{2}{c}{} & \multicolumn{2}{c}{} & \multicolumn{2}{c}{} & \multicolumn{2}{c}{} & \multicolumn{2}{c}{} & \multicolumn{2}{c}{} & \multicolumn{2}{c}{} & \multicolumn{2}{c}{}\tabularnewline
{\small{}Controls $X_{jl}$} & \multicolumn{2}{c}{} & \multicolumn{2}{c}{} & \multicolumn{2}{c}{} & \multicolumn{2}{c}{} & \multicolumn{2}{c}{} & \multicolumn{2}{c}{} & \multicolumn{2}{c}{} & \multicolumn{2}{c}{} & \multicolumn{2}{c}{}\tabularnewline
{\small{}- Completed field} & \multicolumn{2}{c}{} & \multicolumn{2}{c}{{\small{}$\checkmark$}} & \multicolumn{2}{c}{} & \multicolumn{2}{c}{} & \multicolumn{2}{c}{} & \multicolumn{2}{c}{} & \multicolumn{2}{c}{} & \multicolumn{2}{c}{} & \multicolumn{2}{c}{}\tabularnewline
{\small{}- Next-best field} & \multicolumn{2}{c}{} & \multicolumn{2}{c}{} & \multicolumn{2}{c}{{\small{}$\checkmark$}} & \multicolumn{2}{c}{} & \multicolumn{2}{c}{} & \multicolumn{2}{c}{} & \multicolumn{2}{c}{} & \multicolumn{2}{c}{} & \multicolumn{2}{c}{}\tabularnewline
{\small{}- $\Delta$GPA} & \multicolumn{2}{c}{} & \multicolumn{2}{c}{} & \multicolumn{2}{c}{} & \multicolumn{2}{c}{{\small{}$\checkmark$}} & \multicolumn{2}{c}{} & \multicolumn{2}{c}{{\small{}$\checkmark$}} & \multicolumn{2}{c}{} & \multicolumn{2}{c}{} & \multicolumn{2}{c}{{\small{}$\checkmark$}}\tabularnewline
{\small{}- $\Delta$Earnings} & \multicolumn{2}{c}{} & \multicolumn{2}{c}{} & \multicolumn{2}{c}{} & \multicolumn{2}{c}{} & \multicolumn{2}{c}{{\small{}$\checkmark$}} & \multicolumn{2}{c}{{\small{}$\checkmark$}} & \multicolumn{2}{c}{} & \multicolumn{2}{c}{} & \multicolumn{2}{c}{}\tabularnewline
 & \multicolumn{2}{c}{} & \multicolumn{2}{c}{} & \multicolumn{2}{c}{} & \multicolumn{2}{c}{} & \multicolumn{2}{c}{} & \multicolumn{2}{c}{} & \multicolumn{2}{c}{} & \multicolumn{2}{c}{} & \multicolumn{2}{c}{}\tabularnewline
{\small{}$R^{2}$} & {\small{}0}&{\small{}34} & {\small{}0}&{\small{}72} & {\small{}0}&{\small{}64} & {\small{}0}&{\small{}51} & {\small{}0}&{\small{}43} & {\small{}0}&{\small{}52} & \multicolumn{2}{c}{} & {\small{}0}&{\small{}34} & {\small{}0}&{\small{}54}\tabularnewline
\bottomrule
\end{tabular*}{\small\par}

{\footnotesize{}Note: Weighted error-in-variable estimates of $(\hat{\beta}_{jl}^{DK}-\overline{\hat{\beta}}^{NO})=a_{0}+a_{1}(\hat{\beta}_{jl}^{NO}-\overline{\hat{\beta}}^{NO})+X_{jl}\gamma+\tilde{e}_{jl}$. $N=64$ for $t=8$, $N=61$ for $t=13$. Standard errors in parentheses.}{\footnotesize\par}
\end{table}

To explore other explanations for the between-country differences in payoffs, we re-estimate (\ref{eq:OLS-compare-1}) while adjusting for completed and next-best field dummies, as well as differences in average selectivity and earnings (cf. Figure \ref{fig:GPA-and-earnings}) across completed and next-best fields. Table \ref{tab:other-explanation} reports the results. The first column reproduces the basic results reported above in Figure \ref{fig:Payoffs-in-Norway} where we found that the payoff difference was about 20,000 USD, and that the payoffs in Denmark increased by less than one for each unit increase in Norway reflecting the smaller variance in the payoff distribution in Denmark.

We next investigate whether payoffs are more aligned across completed fields or across next-best fields. The next two columns of Table \ref{tab:other-explanation} therefore adjust for completed field and next-best field dummies. Keeping completed field fixed we now obtain a slope estimate of 0.29 in column (2), while keeping next-best field fixed in column (3) increases the slope substantially to 0.70. This shows that differences between next-best fields contribute more to between-country differences in payoffs than differences between completed fields.

In a next step we investigate two potential explanation of such differences. First we verify whether differential selectivity plays a role by adjusting for across country differences in average GPA in both the completed field $j$ and the next-best field $l$. This changes the interpretation of the intercept which now corresponds to the across country payoff difference keeping the average GPA the same in the completed and next-best field. The estimates in column (4) show that this reduces the payoff gap with 25\% from about -20,000 to -16,000 USD, while at the same time payoff become more evenly distributed as shown by the increase in the slope coefficient from 0.37 to 0.56. A similar exercise using average earnings in column (5) and (6) shows that this does not explain across country differences.

As noted earlier, looking at earnings eight years after application corresponds to relatively early career outcomes, especially for 5-year programs and studies that are not closely tied to a narrow set of occupations and which may therefore have longer and more complex school-to-work transitions. We therefore also compare the Danish payoff estimates 13 years after applying with the Norwegian estimates 13 years after applying. Extending the time horizon by 5 years has some impact on the estimated reliability of the Norwegian estimates which drops to 0.64, but the raw correlation between the $t=8$ and $t=13$ estimates is high (0.80).\footnote{We need to exclude three very imprecisely estimated payoffs with Law as next-best field from the $t=13$ analysis (see appendix Table \ref{tab:Payoffs No t13}) to recover a non-negative reliability estimate. } For Denmark the reliability is slightly higher at $t=13$ as is the raw correlation between the $t=8$ and $t=13$ estimates (0.86).

Figure \ref{fig:Distribution-of-pay-offs-1} shows the estimated distributions of these longer-run payoffs across the two countries. Compared to the early career payoffs, the Norwegian and Danish payoff distribution are now much more aligned in terms of location and scale. This can also be seen in column (7) of Table \ref{tab:other-explanation}. The average payoff gap between the two countries is now about 10,000 USD, and the slope coefficient has increased to 0.70.\footnote{Appendix Tables \ref{tab:Payoffs No t13} and \ref{tab:Payoffs Dk t13} report the estimates, and appendix Figure \ref{fig:Payoffs-in-Norway-1} compares the payoff estimates and reports the error-in-variables regression line.} Column (8) shows that after adjusting for differential selectivity the payoff estimates are on average aligned and we cannot reject that the intercept equals zero and the slope equals one (the corresponding F-test gives a \emph{p}-value of 0.67). However, the estimated (reliability corrected) correlation coefficient between the Norwegian and Danish payoff estimates barely moves when comparing $t=8$ vs. $t=13$ (0.66 vs 0.65).

\begin{figure}
\begin{centering}
\includegraphics[width=0.8\textwidth]{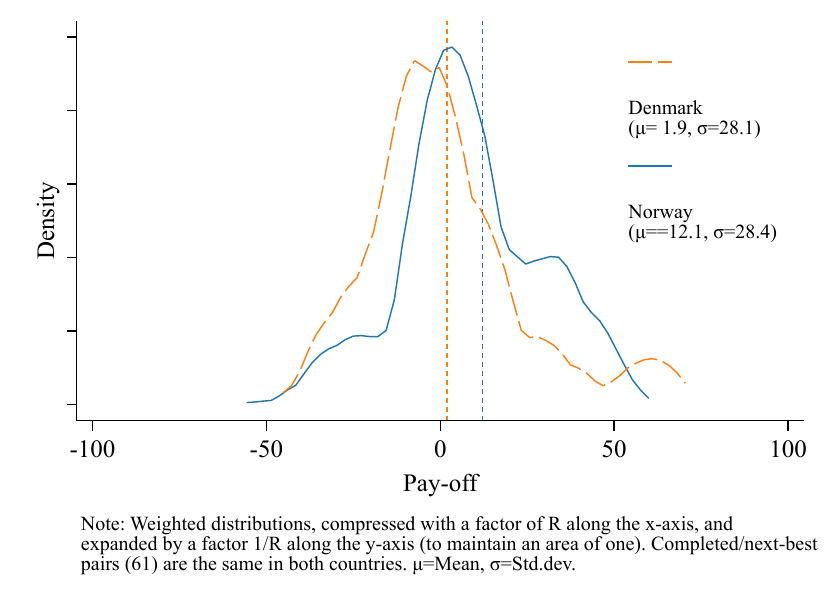}
\par\end{centering}
\caption{Distribution of longer-run payoffs by country (13 years since applying)\label{fig:Distribution-of-pay-offs-1}}
\end{figure}

To summarize, we find that payoff estimates are strongly correlated across countries but have initially different levels and dispersion. Violations of the irrelevance and next-best assumptions that underpin the empirical approach do weaken the correlation, but appear to have little consequence for the estimated level and variance differences. Over time, the level and variance difference converge across countries, but this does not affect the correlation of the payoffs. Additional exploratory analyses show that these across country differences are mostly driven by heterogeneity in next-best fields which can partly be explained by differences in selectivity.

\section{Conclusion\label{sec:Conclusion-and-summary}}

\begin{onehalfspace}
\noindent We revisited the identification argument of \citet{kirkeboen_field_2016} who showed how one may combine instruments for multiple unordered treatments with information about individuals\textquoteright{} ranking of these treatments to achieve identification while allowing for both observed and unobserved heterogeneity in treatment effects. We showed that the key assumptions underlying their identification argument have testable implications. We also provided a new characterization of the bias that may arise if these assumptions are violated. Taken together, these results allow researchers not only to test the underlying assumptions, but also to argue whether the bias from violation of these assumptions are likely to be economically meaningful.
\end{onehalfspace}

\begin{onehalfspace}
Guided and motivated by these results, we estimated and compared the earnings payoffs to post-secondary fields of study in Norway and Denmark. In each country, we applied and assessed the identification argument of \citet{kirkeboen_field_2016} to data on individuals' ranking of fields of study and field-specific instruments from discontinuities in the admission systems. We empirically examined whether and why the payoffs to fields of study differ across the two countries. We found strong cross-country correlation in the payoffs to fields of study, especially after removing fields with violations of the assumptions underlying the identification argument.
\end{onehalfspace}

While our empirical findings are specific to the context of postsecondary education in the Nordic countries, there could be lessons from our work for other settings with unordered choices. Our study highlights key challenges and possible solutions to understanding what the causal effects of these choices are. Examples can be found in observational studies that use IV to study workers\textquoteright{} selection of occupation, students' choice of education, firms\textquoteright{} decision on location, or families\textquoteright{} choice of where to live. Another example is the frequent use of IV to analyze encouragement designs in experiments where treatments are made available but take up is not universal \citep{duflo2007using}.

\bibliographystyle{apalike}
\bibliography{fieldofstudy}

\appendix

\part*{\clearpage Appendix}

\noindent
\global\long\def\thetable{A\arabic{table}}%
\setcounter{table}{0}

\global\long\def\thefigure{A\arabic{figure}}%
\setcounter{figure}{0}

\section{Proof Of Bias When Auxiliary Assumptions Fail\label{appx:no-auxillary-ass}}
\begin{proof}
We build on the notation from Section \ref{sec:assumptions-notation}. IV uses the three moment conditions:
\[
\E[\varepsilon]=0\text{,}\quad\E[\varepsilon z_{1}]=0\quad\text{and}\quad\E[\varepsilon z_{2}]=0
\]
Expressing $\varepsilon$ in terms of potential outcomes, we get:
\begin{align}
\varepsilon & =(y^{0}-\beta_{0})+(y^{1}-y^{0}-\beta_{1})d_{1}+(y^{2}-y^{0}-\beta_{2})d_{2}\label{eq:varepsilon}\\
 & =(y^{0}-\beta_{0})+(y^{1}-y^{0}-\beta_{1})(d_{1}^{0}+(d_{1}^{1}-d_{1}^{0})z_{1}+(d_{1}^{2}-d_{1}^{0})z_{2})\nonumber \\
 & \qquad+(y^{2}-y^{0}-\beta_{2})(d_{2}^{0}+(d_{2}^{1}-d_{2}^{0})z_{1}+(d_{2}^{2}-d_{2}^{0})z_{2})\nonumber
\end{align}
We substitute into the moment conditions, and solve. Under independence, we get:
\begin{align*}
\E[(y^{1}-y^{0}-\beta_{1})(d_{1}^{1}-d_{1}^{0})+(y^{2}-y^{0}-\beta_{2})(d_{2}^{1}-d_{2}^{0})] & =0\\
\E[(y^{1}-y^{0}-\beta_{1})(d_{1}^{2}-d_{1}^{0})+(y^{2}-y^{0}-\beta_{2})(d_{2}^{2}-d_{2}^{0})] & =0
\end{align*}
As shown by \citet{kirkeboen_field_2016}, this implies, for $k=1,2$, $k'=2,1$, that:
\begin{align}
 & \E[y^{k}-y^{0}-\beta_{k}\mid d_{k}^{k}-d_{k}^{0}=1,d_{k'}^{k}-d_{k'}^{0}=0]\times P[d_{k}^{k}-d_{k}^{0}=1,d_{k'}^{k}-d_{k'}^{0}=0]\label{eq:moment-cond}\\
 & +\E[(y^{k}-y^{0}-y^{k'}-y^{0})-(\beta_{k}-\beta_{k'})\mid d_{k}^{k}-d_{k}^{0}=1,d_{k'}^{k}-d_{k'}^{0}=-1]\times P[d_{k}^{k}-d_{k}^{0}=1,d_{k'}^{k}-d_{k'}^{0}=-1]\nonumber \\
 & +\E[y^{k'}-y^{0}-\beta_{k'}\mid d_{k}^{k}-d_{k}^{0}=0,d_{k'}^{k}-d_{k'}^{0}=1]\times P[d_{k}^{k}-d_{k}^{0}=0,d_{k'}^{k}-d_{k'}^{0}=1]=0%\label{eq1:instrum1}%\nonumber\nonumber{align}{align}{align}{align}{align}{align}{align}
\end{align}
where we have assumed
\[
P[d_{k}^{k}-d_{k}^{0}=-1,d_{k'}^{k}-d_{k'}^{0}=0]=P[d_{k}^{k}-d_{k}^{0}=0,d_{k'}^{k}-d_{k'}^{0}=-1]=0
\]
under monotonicity. To simplify notation, we rewrite equation \ref{eq:moment-cond} in terms of the notation from Table \ref{tab:taxonomy}:
\begin{align*}
 & \E[y^{k}-y^{0}-\beta_{k}\mid C_{k}]\times P(C_{k})\\
 & +\E[(y^{k}-y^{0}-y^{k'}-y^{0})-(\beta_{k}-\beta_{k'})\mid ND_{k}]\times P(ND_{k})\\
 & +\E[y^{k'}-y^{0}-\beta_{k'}\mid ID_{k}]\times P(ID_{k})=0
\end{align*}
We isolate $\beta_{k}$ for $k=1,2$:
\begin{align}
\beta_{k} & =\beta_{k'}\frac{P(ND_{k})-P(ID_{k})}{P(C_{k})+P(ND_{k})}+\frac{\E[y^{k}-y^{0}\mid C_{k}]P(C_{k})}{P(C_{k})+P(ND_{k})}\label{eq:solved1}\\[1em]
 & \qquad+\frac{\E[y^{k}-y^{0}-y^{k'}-y^{0}\mid ND_{k}]P(ND_{k})}{P(C_{k})+P(ND_{k})}+\frac{\E[y^{k'}-y^{0}\mid ID_{k}]P(ID_{k})}{P(C_{k})+P(ND_{k})}\nonumber
\end{align}
\end{proof}

\subsection{No Auxiliary Assumptions\label{appx:no-assumptions}}

We substitute equation (\ref{eq:solved1}) with $k=2$ into (\ref{eq:solved1}) with $k=1$ and get:
\begin{align*}
\beta_{1} & =\frac{\E[y^{1}-y^{0}\mid C_{1}]P(C_{1})}{P(C_{1})+P(ND_{1})}\\[1em]
 & \qquad+\E[y^{2}-y^{0}\mid ID_{1}]\frac{P(ID_{1})}{P(C_{1})+P(ND_{1})}+\frac{\E[y^{1}-y^{0}-y^{2}-y^{0}\mid ND_{1}]P(ND_{1})}{P(C_{1})+P(ND_{1})}\\[1em]
 & \qquad+\frac{P(ND_{1})-P(ID_{1})}{P(C_{1})+P(ND_{1})}\times\left[\frac{\E[y^{2}-y^{0}\mid C_{2}]P(C_{2})}{P(C_{2})+P(ND_{2})}+\E[y^{1}-y^{0}\mid ID_{2}]\frac{P(ID_{2})}{P(C_{2})+P(ND_{2})}\right.\\[1em]
 & \qquad\qquad\qquad\left.+\frac{\E[y^{2}-y^{0}-y^{1}-y^{0}\mid ND_{2}]P(ND_{2})}{P(C_{2})+P(ND_{2})}+\beta_{1}\frac{P(ND_{2})-P(ID_{2})}{P(C_{2})+P(ND_{2})}\right]
\end{align*}

Letting
\begin{align*}
\dot{W} & =1-\frac{(P(ND_{1})-P(ID_{1}))(P(ND_{2})-P(ID_{2}))}{(P(C_{1})+P(ND_{1}))(P(C_{2})+P(ND_{2}))}\\[1em]
 & =\frac{(P(C_{1})+P(ND_{1}))(P(C_{2})+P(ND_{2}))-(P(ND_{1})-P(ID_{1}))(P(ND_{2})-P(ID_{2}))}{(P(C_{1})+P(ND_{1}))(P(C_{2})+P(ND_{2}))}
\end{align*}
and gathering $\beta_{1}$-terms on the LHS gives:
\begin{align*}
\beta_{1}\dot{W} & =\E[y^{1}-y^{0}\mid C_{1}]\times\frac{P(C_{1})}{P(C_{1})+P(ND_{1})}\\[1em]
 & \qquad+\E[y^{2}-y^{0}\mid ID_{1}]\times\frac{P(ID_{1})}{P(C_{1})+P(ND_{1})}\\[1em]
 & \qquad+\E[y^{1}-y^{0}-y^{2}-y^{0}\mid ND_{1}]\times\frac{P(ND_{1})}{P(C_{1})+P(ND_{1})}\\[1em]
 & \qquad+\E[y^{1}-y^{0}\mid ID_{2}]\times\frac{(P(ND_{1})-P(ID_{1}))P(ID_{2})}{(P(C_{1})+P(ND_{1})(P(C_{2})+P(ND_{2}))}\\[1em]
 & \qquad+\E[y^{2}-y^{0}\mid C_{2}]\times\frac{(P(ND_{1})-P(ID_{1}))P(C_{2})}{(P(C_{1})+P(ND_{1}))(P(C_{2})+P(ND_{2}))}\\[1em]
 & \qquad+\E[y^{2}-y^{0}-y^{1}-y^{0}\mid ND_{2}]\times\frac{(P(ND_{1})-P(ID_{1}))P(ND_{2})}{(P(C_{1})+P(ND_{1}))(P(C_{2})+P(ND_{2}))}
\end{align*}
Adding and subtracting
\[
\E[y^{1}-y^{0}\mid C_{1}]\frac{P(ND_{1})}{P(C_{1})+P(ND_{1})}+\E[y^{1}-y^{0}\mid C_{1}]\frac{(P(ND_{1})-P(ID_{1}))(P(ND_{2})-P(ID_{2}))}{(P(C_{1})+P(ND_{1}))(P(C_{2})+P(ND_{2}))}
\]
on the RHS and gathering terms gives:
\begin{align*}
\beta_{1}\dot{W}=\E[y^{1}-y^{0}\mid C_{1}]\dot{W}\quad & -\quad\E[y^{1}-y^{0}\mid C_{1}]\times\frac{P(ND_{1})(P(C_{2})+P(ND_{2}))}{(P(C_{1})+P(ND_{1}))(P(C_{2})+P(ND_{2}))}\\[1em]
 & \qquad+\E[y^{1}-y^{0}\mid C_{1}]\times\frac{(P(ND_{1})-P(ID_{1}))(P(ND_{2})-P(ID_{2}))}{(P(C_{1})+P(ND_{1}))(P(C_{2})+P(ND_{2}))}\\[1em]
 & \qquad+\E[y^{2}-y^{0}\mid ID_{1}]\times\frac{P(ID_{1})(P(C_{2})+P(ND_{2}))}{(P(C_{1})+P(ND_{1}))(P(C_{2})+P(ND_{2}))}\\[1em]
 & \qquad+\E[y^{1}-y^{0}-y^{2}-y^{0}\mid ND_{1}]\times\frac{P(ND_{1})(P(C_{2})+P(ND_{2}))}{(P(C_{1})+P(ND_{1}))(P(C_{2})+P(ND_{2}))}\\[1em]
 & \qquad+\E[y^{1}-y^{0}\mid ID_{2}]\times\frac{(P(ND_{1})-P(ID_{1}))P(ID_{2})}{(P(C_{1})+P(ND_{1})(P(C_{2})+P(ND_{2}))}\\[1em]
 & \qquad+\E[y^{2}-y^{0}\mid C_{2}]\times\frac{(P(ND_{1})-P(ID_{1}))P(C_{2})}{(P(C_{1})+P(ND_{1}))(P(C_{2})+P(ND_{2}))}\\[1em]
 & \qquad+\E[y^{2}-y^{0}-y^{1}-y^{0}\mid ND_{2}]\times\frac{(P(ND_{1})-P(ID_{1}))P(ND_{2})}{(P(C_{1})+P(ND_{1}))(P(C_{2})+P(ND_{2}))}
\end{align*}
Dividing by $\dot{W}$ on both sides, and letting
\begin{align*}
\bar{W} & =(P(C_{1})+P(ND_{1}))(P(C_{2})+P(ND_{2}))-(P(ND_{1})-P(ID_{1}))(P(ND_{2})-P(ID_{2}))
\end{align*}
gives
\begin{align*}
\beta_{1}=\E[y^{1}-y^{0}\mid C_{1}]\quad & -\quad\E[y^{1}-y^{0}\mid C_{1}]\times\frac{P(ND_{1})(P(C_{2})+P(ND_{2}))}{\bar{W}}\\[4pt]
 & \qquad+\E[y^{1}-y^{0}\mid C_{1}]\times\frac{(P(ND_{1})-P(ID_{1}))(P(ND_{2})-P(ID_{2}))}{\bar{W}}\\[4pt]
 & \qquad+\E[y^{2}-y^{0}\mid ID_{1}]\times\frac{P(ID_{1})(P(C_{2})+P(ND_{2}))}{\bar{W}}\\[4pt]
 & \qquad+\E[y^{1}-y^{0}-y^{2}-y^{0}\mid ND_{1}]\times\frac{P(ND_{1})(P(C_{2})+P(ND_{2}))}{\bar{W}}\\[4pt]
 & \qquad+\E[y^{1}-y^{0}\mid ID_{2}]\times\frac{(P(ND_{1})-P(ID_{1}))P(ID_{2})}{\bar{W}}\\[4pt]
 & \qquad+\E[y^{2}-y^{0}\mid C_{2}]\times\frac{(P(ND_{1})-P(ID_{1}))P(C_{2})}{\bar{W}}\\[4pt]
 & \qquad+\E[y^{2}-y^{0}-y^{1}-y^{0}\mid ND_{2}]\times\frac{(P(ND_{1})-P(ID_{1}))P(ND_{2})}{\bar{W}}
\end{align*}

Rearranging, we get:
\begin{align}
\beta_{1}^{IV}={\E[y^{1}-y^{0}\mid C_{1}]}\quad & \quad+{\frac{P(ND_{1})P(C_{2})}{\bar{W}}}\quad\times({\E[y^{1}-y^{0}\mid ND_{1}]-\E[y^{1}-y^{0}\mid C_{1}]})\label{eq:final-exp-no-auxass}\\[4pt]
 & \quad+{\frac{P(ND_{1})P(C_{2})}{\bar{W}}}\quad\times({\E[y^{2}-y^{0}\mid C_{2}]-\E[y^{2}-y^{0}\mid ND_{1}]})\nonumber \\[4pt]
 & \quad+{\frac{P(ND_{1})P(ND_{2})}{\bar{W}}}\times({\E[y^{1}-y^{0}\mid ND_{1}]-\E[y^{1}-y^{0}\mid ND_{2}]})\nonumber \\[4pt]
 & \quad+{\frac{P(ND_{1})P(ND_{2})}{\bar{W}}}\times({\E[y^{2}-y^{0}\mid ND_{2}]-\E[y^{2}-y^{0}\mid ND_{1}]})\nonumber \\[4pt]
 & \quad+{\frac{P(ID_{1})P(ID_{2})}{\bar{W}}}\enskip\times({\E[y^{1}-y^{0}\mid C_{1}]-\E[y^{1}-y^{0}\mid ID_{2}]})\nonumber \\[4pt]
 & \quad+{\frac{P(ID_{1})P(C_{2})}{\bar{W}}}\quad\times({\E[y^{2}-y^{0}\mid ID_{1}]-\E[y^{2}-y^{0}\mid C_{2}]})\nonumber \\[4pt]
 & \quad+{\frac{P(ID_{1})P(ND_{2})}{\bar{W}}}\times({\E[y^{1}-y^{0}\mid ND_{2}]-\E[y^{1}-y^{0}\mid C_{1}]})\nonumber \\[4pt]
 & \quad+{\frac{P(ID_{1})P(ND_{2})}{\bar{W}}}\times({\E[y^{2}-y^{0}\mid ID_{1}]-\E[y^{2}-y^{0}\mid ND_{2}]})\nonumber \\[4pt]
 & \quad+{\frac{P(ND_{1})P(ID_{2})}{\bar{W}}}\times({\E[y^{1}-y^{0}\mid ID_{2}]-\E[y^{1}-y^{0}\mid C_{1}]})\nonumber
\end{align}
where we can rearrange the denominator such that
\begin{align*}
\bar{W} & =\enskip P(C_{1})P(C_{2})+P(C_{1})P(ND_{2})+P(ND_{1})P(C_{2})\\
 & \enskip+P(ND_{1})P(ID_{2})+P(ID_{1})P(ND_{2})-P(ID_{1})P(ID_{2})
\end{align*}
and the expression for $\beta_{2}^{IV}$ follows by symmetry.

\subsection{Assuming Only Next-best\label{appx:next-best}}

We now want to find an expression of the bias assuming only next-best.
\begin{proof}
Next-best ensures $P(ND_{1})=P(ND_{2})=0$. Equation \ref{eq:final-exp-no-auxass} then reduces to
\begin{align}
\beta_{1}^{IV}={\E[y^{1}-y^{0}\mid C_{1}]}\quad & \quad+{\frac{P(ID_{1})P(ID_{2})}{W'}}\times({\E[y^{1}-y^{0}\mid C_{1}]-\E[y^{1}-y^{0}\mid ID_{2}]})\label{eq:final-exp-only-sec-best}\\[4pt]
 & \quad+{\frac{P(ID_{1})P(C_{2})}{W'}}\times({\E[y^{2}-y^{0}\mid ID_{1}]-\E[y^{2}-y^{0}\mid C_{2}]})\nonumber
\end{align}
where
\begin{align*}
W' & =P(C_{1})(P(C_{2})-P(ID_{1})P(ID_{2})
\end{align*}
\end{proof}

\subsection{Assuming Only Irrelevance\label{appx:irrelevance}}

We now want to find an expression of the bias assuming only irrelevance.
\begin{proof}
Irrelevance ensures $P(ID_{1})=P(ID_{2})=0$. Equation \ref{eq:final-exp-no-auxass} then reduces to
\begin{align}
\beta_{1}^{IV}={\E[y^{1}-y^{0}\mid C_{1}]}\quad & \quad+{\frac{P(ND_{1})P(C_{2})}{\hat{W}}}\times({\E[y^{1}-y^{0}\mid ND_{1}]-\E[y^{1}-y^{0}\mid C_{1}]})\label{eq:final-exp-only-irrelevance}\\[4pt]
 & \quad+{\frac{P(ND_{1})P(C_{2})}{\hat{W}}}\times({\E[y^{2}-y^{0}\mid C_{2}]-\E[y^{2}-y^{0}\mid ND_{1}]})\nonumber \\[4pt]
 & \quad+{\frac{P(ND_{1})P(ND_{2})}{\hat{W}}}\times({\E[y^{1}-y^{0}\mid ND_{1}]-\E[y^{1}-y^{0}\mid ND_{2}]})\nonumber \\[4pt]
 & \quad+{\frac{P(ND_{1})P(ND_{2})}{\hat{W}}}\times({\E[y^{2}-y^{0}\mid ND_{2}]-\E[y^{2}-y^{0}\mid ND_{1}]})\nonumber
\end{align}
where
\begin{align*}
\hat{W} & =(P(C_{1})+P(ND_{1}))(P(C_{2})+P(ND_{2}))-P(ND_{1})P(ND_{2}))\\[2pt]
 & =P(C_{1})P(C_{2})+P(C_{1})P(ND_{2})+P(ND_{1})P(C_{2})
\end{align*}
\end{proof}
\newpage{}

\section{Proof of Testable Implications\label{appx:empirical-test}}

\subsection{First Stage Quantities}

We start by proving Proposition \ref{th:identified-quantities}
\begin{proof}
We start by introducing a richer decomposition of behavioral groups, building on Table \ref{tab:taxonomy}. This is presented in Table \ref{tab:taxonomy-detailed}.

\begin{table}
\caption{Detailed taxonomy of behavioral groups.\label{tab:taxonomy-detailed}}
{\small{}}%
\begin{tabular*}{1\columnwidth}{@{\extracolsep{\fill}}cccccccc}
\toprule
 & \multicolumn{3}{c}{\textbf{\small{}Potential Field Choice}} &  & \multicolumn{2}{c}{\textbf{\small{}Behavioral type}} & \tabularnewline
\cmidrule{2-4} \cmidrule{3-4} \cmidrule{4-4} \cmidrule{6-7} \cmidrule{7-7}
 & {\small{}$d^{0}$} & {\small{}$d^{1}$} & {\small{}$d^{2}$} &  & \textbf{\small{}$z_{1}$-stratum} & \textbf{\small{}$z_{2}$-stratum} & \tabularnewline
\midrule
 & {\small{}0} & {\small{}1} & {\small{}2} &  & {\small{}$C_{1}$} & {\small{}$C_{2}$} & \tabularnewline
 & {\small{}0} & {\small{}1} & {\small{}1} &  & {\small{}$C_{1}$} & {\small{}$ID_{2}$} & \tabularnewline
 & {\small{}0} & {\small{}1} & {\small{}0} &  & {\small{}$C_{1}$} & {\small{}$NT_{2}$} & \tabularnewline
 & {\small{}0} & {\small{}0} & {\small{}0} &  & {\small{}$NT_{1}$} & {\small{}$NT_{2}$} & \tabularnewline
 & {\small{}0} & {\small{}0} & {\small{}2} &  & {\small{}$NT_{1}$} & {\small{}$C_{2}$} & \tabularnewline
 & {\small{}2} & {\small{}2} & {\small{}2} &  & {\small{}$OT_{1}$} & {\small{}$AT_{2}$} & \tabularnewline
 & {\small{}1} & {\small{}1} & {\small{}1} &  & {\small{}$AT_{1}$} & {\small{}$OT_{2}$} & \tabularnewline
 & {\small{}1} & {\small{}1} & {\small{}2} &  & {\small{}$AT_{1}$} & {\small{}$ND_{2}$} & \tabularnewline
 & {\small{}2} & {\small{}1} & {\small{}2} &  & {\small{}$ND_{1}$} & {\small{}$AT_{2}$} & \tabularnewline
 & {\small{}0} & {\small{}2} & {\small{}2} &  & {\small{}$ID_{1}$} & {\small{}$C_{2}$} & \tabularnewline
\bottomrule
\end{tabular*}{\small\par}

\begin{singlespace}
\textbf{\footnotesize{}Note:}{\footnotesize{} The table decomposes the behavioral groups from Table \ref{tab:taxonomy} into subgroups (strata) where the $z_{1}$ stratum is the group defined by their potential field choices when changing the instrument from 0 to 1, and the $z_{2}$ stratum is correspondingly defined for an instrument change from 0 to 2. The table shows the possible behavioral responses under all states of the instrument. Note that other takers $OT_{1}$ ($OT_{2}$) refers to global always takers of field 2 (1).}{\footnotesize\par}
\end{singlespace}
\end{table}

Focusing on $k=1$, we take expectations on both sides in equation (\ref{eq:first-stage-1}). As $\E[\nu_{1}]=0$, we get:
\begin{align}
\E[d_{1}] & =\alpha_{1}^{0}+\alpha_{1}^{1}\times\E[z_{1}]+\alpha_{1}^{2}\times\E[z_{2}]\label{eq:intermediate-alphas}
\end{align}
We decompose the LHS into potential outcomes, using that $z_{0}=1-z_{1}-z_{2}$. Under independence we have:
\begin{align}
\E[d_{1}] & =\E[d_{1}^{0}]+\E[d_{1}^{1}-d_{1}^{0}]\times\E[z_{1}]+\E[d_{1}^{2}-d_{1}^{0}]\times\E[z_{2}]\label{eq:expectations-ds}
\end{align}
Using Table \ref{tab:taxonomy}, as groups are disjoint, we have
\begin{align*}
\E[d_{1}^{0}] & =P(d_{1}^{0}=1)=P(AT_{1})\\[2pt]
\E[d_{1}^{1}-d_{1}^{0}] & =P(d_{1}^{1}-d_{1}^{0}=1)=P(C_{1})+P(ND_{1})\\[2pt]
\E[d_{1}^{2}-d_{1}^{0}] & =P(d_{1}^{2}-d_{1}^{0}=1)-P(d_{1}^{2}-d_{1}^{0}=-1)=P(ID_{2})-P(ND_{2})
\end{align*}
where we in both instances have assumed monotonicity and $AT_{1}$ denotes always takers. This turns equation (\ref{eq:intermediate-alphas}) into:
\begin{align*}
 & \qquad\alpha_{1}^{0}-P(AT_{1})\\
 & +\enskip[\alpha_{1}^{1}-(P(C_{1})+P(ND_{1}))]\times\E[z_{1}]\\
 & +[\alpha_{1}^{2}-(P(ID_{2})-P(ND_{2}))]\times\E[z_{2}]=0
\end{align*}
By the rank condition (and symmetry for $k=2$), this implies:
\begin{align}
P(AT_{1}) & =\alpha_{1}^{0} & P(AT_{2}) & =\alpha_{2}^{0}\label{eq:emp-noass-1}\\
P(C_{1})+P(ND_{1}) & =\alpha_{1}^{1} & P(C_{2})+P(ND_{2}) & =\alpha_{2}^{2}\label{eq:emp-noass-2}\\
P(ID_{1})-P(ND_{1}) & =\alpha_{2}^{1} & P(ID_{2})-P(ND_{2}) & =\alpha_{1}^{2}\label{eq:emp-noass-3}
\end{align}
Since groups are disjoint we have
\begin{align}
P(C_{1})+P(AT_{1})+P(NT_{1})+P(OT_{1})+P(ID_{1})+P(ND_{1}) & =1\label{eq:emp-noass-4}\\
P(C_{2})+P(AT_{2})+P(NT_{2})+P(OT_{2})+P(ID_{2})+P(ND_{2}) & =1\label{eq:emp-noass-5}
\end{align}
By combining equation (\ref{eq:emp-noass-4}) with equations (\ref{eq:emp-noass-1})-(\ref{eq:emp-noass-3}) we get\footnote{Where we use the following $AT_{2}=OT_{1}\cup ND_{1}$ and $AT_{1}=OT_{2}\cup ND_{2}$.}
\begin{align}
P(NT_{1}) & =1-\alpha_{1}^{0}-\alpha_{2}^{0}-\alpha_{1}^{1}-\alpha_{2}^{1}\label{eq:empirical-noass-control-all-1}\\
P(NT_{2}) & =1-\alpha_{1}^{0}-\alpha_{2}^{0}-\alpha_{2}^{2}-\alpha_{1}^{2}\label{eq:empirical-noass-control-all-2}
\end{align}

\end{proof}

\subsection{Partial Identification Of Defiers}

We continue by proving Proposition \ref{th:testable-implications}
\begin{proof}
From Proposition \ref{th:identified-quantities}, we get the following information on $P(ND_{1})$:
\begin{align}
P(ND_{1})=\begin{cases}
\enskip-\alpha_{2}^{1}+P(ID_{1})\\[2pt]
\quad\alpha_{2}^{0}-P(OT_{1})\\[2pt]
\quad\alpha_{1}^{1}-P(C_{1})
\end{cases}\label{eq:allinfo-next-best}
\end{align}
where the first line follows from equation (\ref{eq:emp-noass-3}), the second from (\ref{eq:emp-noass-2}) and the third from combining equation \ref{eq:emp-noass-4} with \ref{eq:empirical-noass-control-all-1} and \ref{eq:emp-noass-1}. From equation (\ref{eq:emp-noass-3}) we know that $P(ID_{1})=\alpha_{2}^{1}+P(ND_{1})$. Combining this with the information in equation (\ref{eq:allinfo-next-best}) we have:
\begin{align*}
P(ID_{1})=\begin{cases}
\quad\alpha_{2}^{1}+P(ND_{1})\\[2pt]
\quad\alpha_{2}^{1}+\alpha_{2}^{0}-P(OT_{1})\\[2pt]
\quad\alpha_{2}^{1}+\alpha_{1}^{1}-P(C_{1})
\end{cases}
\end{align*}
This gives the following bounds on $P(ID_{1})$ and $P(ND_{1})$
\begin{align*}
P(ND_{1}) & \geq-\alpha_{2}^{1} & P(ID_{1}) & \geq\alpha_{2}^{1}\\[2pt]
P(ND_{1}) & \leq\alpha_{2}^{0} & P(ID_{1}) & \leq\alpha_{2}^{1}+\alpha_{2}^{0}\\[2pt]
P(ND_{1}) & \leq\alpha_{1}^{1} & P(ID_{1}) & \leq\alpha_{2}^{1}+\alpha_{1}^{1}
\end{align*}
where also, trivially, $P(ID_{1}),P(ND_{1})\geq0$. It follows that the bounds on $P(ID_{1})$ are:
\begin{align*}
\max\{0,-\alpha_{2}^{1}\} & \leq P(ND_{1})\leq\qquad\qquad\quad\enskip\,\min\{\alpha_{1}^{1},\alpha_{2}^{0}\}\\
\max\{0,\enskip\,\alpha_{2}^{1}\} & \leq P(ID_{1})\enskip\leq\max\{0,\alpha_{2}^{1}+\min\{\alpha_{1}^{1},\alpha_{2}^{0}\}\}
\end{align*}
and results for instrument 2 are symmetric.
\end{proof}

\subsection{Assuming Next-best}

We now prove Corollary \ref{cor:testable-next-best}.
\begin{proof}
Assuming next-best, we have $P(ND_{1})=P(ND_{2})=0$. This turns equation (\ref{eq:emp-noass-2}) into:
\begin{align*}
P(AT_{1}) & =\alpha_{1}^{0} & P(AT_{2}) & =\alpha_{2}^{0}\\
P(C_{1}) & =\alpha_{1}^{1} & P(C_{2}) & =\alpha_{2}^{2}\\
P(ID_{1}) & =\alpha_{2}^{1} & P(ID_{2}) & =\alpha_{1}^{2}
\end{align*}
\end{proof}

\subsection{Assuming Irrelevance}

Lastly, we prove Corollary \ref{cor:testable-irrelevance}
\begin{proof}
Assuming irrelevance, we have $P(ID_{1})=P(ID_{2})=0$. This turns equation (\ref{eq:emp-noass-2}) into:
\begin{align*}
P(AT_{1}) & =\alpha_{1}^{0} & P(AT_{2}) & =\alpha_{2}^{0}\\
P(C_{1}) & =\alpha_{1}^{1}+\alpha_{2}^{1} & P(C_{2}) & =\alpha_{2}^{2}+\alpha_{1}^{2}\\
P(ND_{1}) & =-\alpha_{2}^{1} & P(ND_{2}) & =-\alpha_{1}^{2}
\end{align*}
\end{proof}
\newpage{}

\section{Proof Of Violation of Exclusion Under Clustering\label{appx:proof-violation-exclusion}}

In the following, we derive an expression for the IV estimand under binary clustering, as presented in Section \ref{sec:clustering-ass-not}.

\subsection{Introduction}

As mentioned in Section \ref{sec:clustering-ass-not}, we have the binary IV estimand in our set-up as:
\[
\tilde{\beta}_{1}^{\text{IV}}=\frac{\theta_{1}}{\pi_{1}}
\]
where $\theta_{1}$ is the reduced form and $\pi_{1}$ is the first stage between when clustering treatments in two clusters, $S_{0}$ and $S_{1}$, and seeking to estimate the effect of going from the former to the latter. In the following we will derive a general expression for this estimand.

\subsubsection{First Stage}

We have the first stage given by the relation
\[
\tilde{d}=\pi_{0}+\pi_{1}\tilde{z}+\nu
\]
Taking expectations on both sides with $\E[\nu]=0$, we get
\[
\E[\tilde{d}]=\pi_{0}+\pi_{1}\times\E[\tilde{z}_{1}]
\]
Decomposing the LHS into potential outcomes using $\tilde{d}=\tilde{d}^{0}+(\tilde{d}^{1}-\tilde{d}^{0})\times\tilde{z}$ we get:
\begin{equation}
\E[\tilde{d}]=\E[\tilde{d}^{0}]+\E[\tilde{d}^{1}-\tilde{d}^{0}]\times\E[\tilde{z}]\label{eq:tsls-expectation-fs}
\end{equation}
i.e. we have
\begin{equation}
\pi_{0}+\pi_{1}\times\E[\tilde{z}]=\E[\tilde{d}^{0}]+\E[\tilde{d}^{1}-\tilde{d}^{0}]\times\E[\tilde{z}]\label{eq:tsls-condition-fs}
\end{equation}

\subsubsection{Reduced Form}

With respect to the reduced form, we have:
\[
\theta_{1}=\E[y\mid\tilde{z}=1]-\E[y\mid\tilde{z}=0]
\]
We substitute for potential outcomes with $y=\tilde{y}^{0}\times(1-\tilde{d})+\tilde{y}^{1}\times\tilde{d}$

\[
\theta_{1}=\E[\tilde{y}^{0}(1-\tilde{d})+\tilde{y}^{1}\tilde{d}\mid\tilde{z}=1]-\E[\tilde{y}^{0}(1-\tilde{d})+\tilde{y}^{1}\tilde{d}\mid\tilde{z}=0]
\]
Since we do not assume cluster-level exclusion, we need to keep potential treatments and outcomes instrument-dependent. Rearranging we get:
\begin{align*}
\theta_{1} & =\E[\tilde{y}^{0,1}\tilde{d}_{0}^{1}\mid\tilde{z}=1]+\E[\tilde{y}^{1,1}\tilde{d}_{1}^{1}\mid\tilde{z}=1]\\
 & \phantom{=}-\E[\tilde{y}^{0,0}\tilde{d}_{0}^{0}\mid\tilde{z}=0]-\E[\tilde{y}^{1,0}\tilde{d}_{1}^{0}\mid\tilde{z}=0]
\end{align*}
Rearranging, this becomes:
\begin{align}
\theta_{1}= & \E[\tilde{y}^{0,1}\mid\tilde{d}^{1}=0]P(\tilde{d}^{1}=0)+\E[\tilde{y}^{1,1}\mid\tilde{d}^{1}=1]P(\tilde{d}^{1}=1)\nonumber \\[2pt]
 & -\E[\tilde{y}^{0,0}\mid\tilde{d}^{0}=0]P(\tilde{d}^{0}=0)-\E[\tilde{y}^{1,0}\mid\tilde{d}^{0}=1]P(\tilde{d}^{0}=1)\label{eq:tsls-theta-expectation}
\end{align}
Under control clustering, we will have
\begin{align*}
\enskip S_{1}=\{1\},\enskip S_{0}=\{0,2\}\qquad\text{or}\qquad\enskip S_{1}=\{2\},\enskip S_{0}=\{0,1\}
\end{align*}
and under treatment clustering we will have
\begin{align*}
\enskip S_{1}=\{1,2\},\enskip S_{0}=\{0\}
\end{align*}
We will treat these scenarios separately, but focussing on the former control clustering scenario as these are symmetric.

\subsection{Control Clustering}

We have $\enskip S_{1}=\{1\},\enskip S_{0}=\{0,2\}$ and seek to find an expression of the first stage, reduced form and IV estimand. For brevity of notation, we use the taxonomy in Table \ref{tab:taxonomy-clustered-tsls-A} to denote complier and defier groups.

\begin{table}
\caption{Taxonomy of response groups under control clustering\label{tab:taxonomy-clustered-tsls-A}}
{\small{}}%
\begin{tabular*}{1\textwidth}{@{\extracolsep{\fill}}lccccccccc}
\toprule
\multicolumn{1}{l}{\textbf{\small{}Type}} & \multicolumn{2}{c}{\textbf{\small{}Cluster Level}} &  & \multicolumn{3}{c}{\textbf{\small{}Field Level}} &  & \multicolumn{2}{c}{\textbf{\small{}Group}}\tabularnewline
\cmidrule{2-10} \cmidrule{3-10} \cmidrule{4-10} \cmidrule{5-10} \cmidrule{6-10} \cmidrule{7-10} \cmidrule{8-10} \cmidrule{9-10} \cmidrule{10-10}
\multicolumn{1}{l}{} & {\small{}$\tilde{d}^{0}$} & \multicolumn{1}{c}{{\small{}$\tilde{d}^{1}$}} &  & \multicolumn{1}{c}{{\small{}$d^{0}$}} & \multicolumn{1}{c}{{\small{}$d^{2}$}} & \multicolumn{1}{c}{{\small{}$d^{1}$}} &  & \multicolumn{1}{c}{{\small{}Field}} & \multicolumn{1}{c}{{\small{}Cluster}}\tabularnewline
\midrule
{\small{}Compliers} & {\small{}$0$} & {\small{}$1$} &  & {\small{}$0$} & {\small{}$ $} & {\small{}$1$} &  & \multirow{4}{*}{{\small{}$\overline{C}$}} & {\small{}$C_{1}$}\tabularnewline
 & {\small{}$0$} & {\small{}$1$} &  & {\small{}$ $} & {\small{}$2$} & {\small{}$1$} &  &  & {\small{}$C_{2}$}\tabularnewline
 & {\small{}$0$} & {\small{}$1$} &  & {\small{}$2$} & {\small{}$ $} & {\small{}$1$} &  &  & {\small{}$ND_{1}$}\tabularnewline
 & {\small{}$0$} & {\small{}$1$} &  & {\small{}$ $} & {\small{}$0$} & {\small{}$1$} &  &  & {\small{}$ND_{2}$}\tabularnewline
{\small{}Never Takers} & {\small{}$0$} & {\small{}$0$} &  & {\small{}$2$} & {\small{}$ $} & {\small{}$0$} &  & \multirow{2}{*}{{\small{}$\overline{NT}$}} & {\small{}$ID_{1}$}\tabularnewline
 & {\small{}$0$} & {\small{}$0$} &  & {\small{}$ $} & {\small{}$0$} & {\small{}$2$} &  &  & {\small{}$ID_{2}$}\tabularnewline
\bottomrule
\end{tabular*}{\small\par}

\begin{singlespace}
\textbf{\footnotesize{}Note:}{\footnotesize{} The table shows potential treatments for field and cluster instruments for groups impacted by the cluster instrument under control clustering. At the field level, $d^{0}$ indicates which treatment is taken given $Z=0$,$d^{2}$ indicates which treatment is taken given $Z=2$ and $d^{1}$ indicates which treatment is taken when $Z=1$. The notation is equivalent at the cluster level. Relative to the clustered instrument, $\overline{C}$ are compliers and $\overline{NT}$ are never takers. Relative to the field instrument, $C$ are compliers, $ND$ are next-best defiers and $ID$ are irrelevance defiers, all relative to some field level instrument corresponding to a treatment in $S_{1}$.}{\footnotesize\par}
\end{singlespace}
\end{table}

\subsubsection{First Stage}

Applying the taxonomy to the expectation in equation (\ref{eq:tsls-expectation-fs}), under field level monotonicity we get:
\begin{align*}
\E[\tilde{d}^{1}-\tilde{d}^{0}] & =P[\tilde{d}^{1}-\tilde{d}^{0}=1]-P[\tilde{d}^{1}-\tilde{d}^{0}=-1]=P(\overline{C})
\end{align*}
From equation (\ref{eq:tsls-condition-fs}) we hence have by the rank condition
\[
\pi_{1,0}=P(\overline{C})
\]

\subsubsection{Reduced Form}

We use Table \ref{tab:taxonomy-clustered-tsls-A} to decompose the expectations in equation (\ref{eq:tsls-theta-expectation}). Under independence and field level monotonicity, we get:
\begin{align*}
\theta_{1} & \enskip=\enskip\E[\tilde{y}^{0,1}\mid\overline{NT}]\times P(\overline{NT})\\
 & \quad+\E[\tilde{y}^{1,1}\mid\overline{C}]\times P(\overline{C})\\
 & \quad-\E[\tilde{y}^{0,0}\mid\overline{C}\cup\overline{NT}]\times P(\overline{C}\cup\overline{NT})\\
\intertext{\text{Since sets are disjoint, we can rearrange:}}\theta_{1} & \enskip=\enskip\E[\tilde{y}^{1,1}-\tilde{y}^{0,0}\mid\overline{C}]\times P(\overline{C})\\
 & \quad-\E[\tilde{y}^{0,1}-\tilde{y}^{0,0}\mid\overline{NT}]\times P(\overline{NT})
\end{align*}
Using Table \ref{tab:taxonomy-clustered-tsls-A} to turn cluster level groups into field level groups, changing outcome indices to reflect instruments relevant to the group in question, we get:
\begin{align*}
\theta_{1} & \enskip=\enskip\E[y^{1,1}-y^{0,0}\mid{C}_{1}]\times P({C}_{1})\\
 & \quad+\E[y^{1,1}-y^{2,2}\mid{C}_{2}]\times P({C}_{2})\\
 & \quad+\E[y^{1,1}-y^{2,0}\mid{ND}_{1}]\times P({ND}_{1})\\
 & \quad+\E[y^{1,1}-y^{0,2}\mid{ND}_{2}]\times P({ND}_{2})\\
 & \quad-\E[y^{0,1}-y^{2,2}\mid{ID}_{1}]\times P({ID}_{1})\\
 & \quad-\E[y^{2,1}-y^{0,0}\mid{ID}_{2}]\times P({ID}_{2})\intertext{\text{At the field level, we assume exclusion, hence:}}\theta_{1} & \enskip=\enskip\E[y^{1}-y^{0}\mid{C}_{1}]\times P({C}_{1})\\
 & \quad+\E[y^{1}-y^{2}\mid{C}_{2}]\times P({C}_{2})\\
 & \quad+\E[y^{1}-y^{2}\mid{ND}_{1}]\times P({ND}_{1})\\
 & \quad+\E[y^{1}-y^{0}\mid{ND}_{2}]\times P({ND}_{2})\\
 & \quad+\E[y^{2}-y^{0}\mid{ID}_{1}]\times P({ID}_{1})\\
 & \quad-\E[y^{2}-y^{0}\mid{ID}_{2}]\times P({ID}_{2})
\end{align*}
We divide by the first stage and rearrange. This gives us:
\begin{align*}
\tilde{\beta}_{1}^{\text{IV}}= & \enskip\frac{P(C_{1})}{\pi_{1,0}}\enskip\underbrace{\E[y^{1}-y^{0}\mid C_{1}]}_{\substack{\text{A}}
}\enskip+\enskip\frac{P(C_{2})}{\pi_{1,0}}\enskip\underbrace{\E[y^{1}-y^{2}\mid C_{2}]}_{\substack{\text{A}}
}\\[4pt]
 & \enskip+\frac{P(ND_{1})}{\pi_{1,0}}\enskip\underbrace{\E[y^{1}-y^{2}\mid ND_{1}]}_{\substack{\text{A}}
}\enskip+\enskip\frac{P(ND_{2})}{\pi_{1,0}}\enskip\underbrace{\E[y^{1}-y^{0}\mid ND_{2}]}_{\substack{\text{A}}
}\\[4pt]
 & \enskip+\frac{P(ID_{1})}{\pi_{1,0}}\enskip\underbrace{\E[y^{2}-y^{0}\mid ID_{1}]}_{\substack{\text{B}}
}\enskip-\enskip\frac{P(ID_{2})}{\pi_{1,0}}\enskip\underbrace{\E[y^{2}-y^{0}\mid ID_{2}]}_{\substack{\text{B}}
}
\end{align*}
where
\begin{align*}
\pi_{1,0}=P(C_{1}\cup C_{2}\cup ND_{1}\cup ND_{2})
\end{align*}
This can be rewritten as:
\begin{align*}
\tilde{\beta}_{1,0}^{IV}\enskip & =\enskip{\frac{P(C_{1}\cup ND_{2})}{\pi_{1,0}}\E[y^{1}-y^{0}\mid C_{1}\cup ND_{2}]+\frac{P(C_{2}\cup ND_{1})}{\pi_{1,0}}\E[y^{1}-y^{2}\mid C_{2}\cup ND_{1}]}\\[4pt]
 & \qquad+{\frac{P(ID_{1})}{\pi_{1,0}}\enskip\enskip\E[y^{2}-y^{0}\mid ID_{1}]-\frac{P(ID_{2})}{\pi_{1,0}}\enskip\enskip\E[y^{2}-y^{0}\mid ID_{2}]}
\end{align*}

\subsection{Treatment Clustering}

We have $\enskip S_{1}=\{1,2\},\enskip S_{0}=\{0\}$ and seek to find an expression of the first stage, reduced form and IV estimand. We use the taxonomy in Table \ref{tab:taxonomy-clustered-tsls-B} to denote complier and defier groups.

\begin{table}
\caption{Taxonomy of response groups under treatment clustering.\label{tab:taxonomy-clustered-tsls-B}}
{\small{}}%
\begin{tabular*}{1\textwidth}{@{\extracolsep{\fill}}lccccccccrc}
\toprule
\multicolumn{1}{l}{\textbf{\small{}Type}} & \multicolumn{2}{c}{\textbf{\small{}Cluster Level}} &  & \multicolumn{3}{c}{\textbf{\small{}Field Level}} &  & \multicolumn{3}{c}{\textbf{\small{}Group}}\tabularnewline
\cmidrule{2-3} \cmidrule{3-3} \cmidrule{5-7} \cmidrule{6-7} \cmidrule{7-7} \cmidrule{9-11} \cmidrule{10-11} \cmidrule{11-11}
\multicolumn{1}{l}{} & \multicolumn{1}{c}{{\small{}$\tilde{d}^{0}$}} & \multicolumn{1}{c}{{\small{}$\tilde{d}^{1}$}} &  & \multicolumn{1}{c}{{\small{}$d^{0}$}} & \multicolumn{1}{c}{{\small{}$d^{1}$}} & \multicolumn{1}{c}{{\small{}$d^{2}$}} &  & {\small{}Field} &  & \multicolumn{1}{c}{{\small{}Cluster}}\tabularnewline
\midrule
{\small{}Compliers} & {\small{}$0$} & {\small{}$1$} &  & {\small{}$0$} & {\small{}$ $} & {\small{}$2$} &  & \multirow{4}{*}{{\small{}$\overline{C}$}} & \multirow{4}{*}{{\small{}$\left\{ \vphantom{\begin{array}{c}
\\
\\
\\
\\
\end{array}}\right.$}} & {\small{}$C_{1}$}\tabularnewline
 & {\small{}$0$} & {\small{}$1$} &  & {\small{}$0$} & {\small{}$1$} & {\small{}$ $} &  &  &  & {\small{}$C_{2}$}\tabularnewline
 & {\small{}$0$} & {\small{}$1$} &  & {\small{}$0$} & {\small{}$ $} & {\small{}$1$} &  &  &  & {\small{}$ID_{1}$}\tabularnewline
 & {\small{}$0$} & {\small{}$1$} &  & {\small{}$0$} & {\small{}$2$} & {\small{}$ $} &  &  &  & {\small{}$ID_{2}$}\tabularnewline
{\small{}Always Takers} & {\small{}$1$} & {\small{}$1$} &  & {\small{}$1$} & {\small{}$ $} & {\small{}$2$} &  & \multirow{2}{*}{{\small{}$\overline{AT}$}} & \multirow{2}{*}{{\small{}$\left\{ \vphantom{\begin{array}{c}
\\
\\
\end{array}}\right.$}} & {\small{}$ND_{1}$}\tabularnewline
 & {\small{}$1$} & {\small{}$1$} &  & {\small{}$2$} & {\small{}$1$} & {\small{}$ $} &  &  &  & {\small{}$ND_{2}$}\tabularnewline
\bottomrule
\end{tabular*}{\small\par}

\smallskip{}

\begin{singlespace}
\textbf{\footnotesize{}Note:}{\footnotesize{} The table shows potential treatments for field and cluster instruments for groups impacted by the cluster instrument under treatment clustering. At the field level, $d^{0}$ indicates which treatment is taken given $Z=0$, $d^{1}$ indicates which treatment is taken when $Z=1$ and $d^{2}$ indicates which treatment is taken given $Z=2$. The notation is equivalent at the cluster level. Relative to the clustered instrument, $\overline{C}$ are compliers and $\overline{AT}$ are always takers. Relative to the field instrument, $C$ are compliers, $ID$ are irrelevance defiers and $ND$ are next-best defiers.}{\footnotesize\par}
\end{singlespace}
\end{table}

\subsubsection{First Stage}

Applying the taxonomy to the expectation in equation (\ref{eq:tsls-expectation-fs}), under field level monotonicity we get:
\begin{align*}
\E[\tilde{d}^{1}-\tilde{d}^{0}] & =P[\tilde{d}^{1}-\tilde{d}^{0}=1]-P[\tilde{d}^{1}-\tilde{d}^{0}=-1]=P(\overline{C})
\end{align*}
From equation (\ref{eq:tsls-condition-fs}) we hence have by the rank condition
\[
\pi_{1,0}=P(\overline{C})
\]

\subsubsection{Reduced Form}

We use Table \ref{tab:taxonomy-clustered-tsls-B} to decompose the expectations in equation (\ref{eq:tsls-theta-expectation}). Under independence and field level monotonicity, we get:
\begin{align*}
\theta_{1} & \enskip=\enskip\E[\tilde{y}^{1,1}\mid\overline{C}]\times P(\overline{C})\\
 & \quad+\E[\tilde{y}^{1,1}\mid\overline{AT}]\times P(\overline{AT})\\
 & \quad-\E[\tilde{y}^{0,0}\mid\overline{C}]\times P(\overline{C})\\
 & \quad-\E[\tilde{y}^{1,0}\mid\overline{AT}]\times P(\overline{AT})\intertext{\text{This rearranges to:}}\theta_{1} & \enskip=\enskip\E[\tilde{y}^{1,1}-\tilde{y}^{0,0}\mid\overline{C}]\times P(\overline{C})\\
 & \quad-\E[\tilde{y}^{0,1}-\tilde{y}^{0,0}\mid\overline{AT}]\times P(\overline{AT})
\end{align*}
Using Table \ref{tab:taxonomy-clustered-tsls-A} to turn cluster level groups into field level groups, further using that groups are disjoint, and changing outcome indices to reflect instruments relevant to the group in question, we get:
\begin{align*}
\theta_{1} & \enskip=\enskip\E[y^{2,2}-y^{0,0}\mid{C}_{1}]\times P({C}_{1})\\
 & \quad+\E[y^{1,1}-y^{0,0}\mid{C}_{2}]\times P({C}_{2})\\
 & \quad+\E[y^{1,2}-y^{0,0}\mid{ID}_{1}]\times P({ID}_{1})\\
 & \quad+\E[y^{2,1}-y^{0,0}\mid{ID}_{2}]\times P({ID}_{2})\\
 & \quad-\E[y^{2,2}-y^{1,0}\mid{ND}_{1}]\times P({ND}_{1})\\
 & \quad-\E[y^{1,1}-y^{2,0}\mid{ND}_{2}]\times P({ND}_{2})\intertext{\text{At the field level, we assume exclusion, hence:}}\theta_{1} & \enskip=\enskip\E[y^{2}-y^{0}\mid{C}_{1}]\times P({C}_{1})\\
 & \quad+\E[y^{1}-y^{0}\mid{C}_{2}]\times P({C}_{2})\\
 & \quad+\E[y^{1}-y^{0}\mid{ID}_{1}]\times P({ID}_{1})\\
 & \quad+\E[y^{2}-y^{0}\mid{ID}_{2}]\times P({ID}_{2})\\
 & \quad-\E[y^{2}-y^{1}\mid{ND}_{1}]\times P({ND}_{1})\\
 & \quad-\E[y^{1}-y^{2}\mid{ND}_{2}]\times P({ND}_{2})
\end{align*}

We divide by the first stage and rearrange. This gives us:
\begin{align*}
\tilde{\beta}_{1}^{\text{IV}}\enskip & =\enskip\frac{P(C_{1})}{\pi_{1,0}}\enskip{\E[y^{2}-y^{0}\mid C_{1}]}\enskip+\enskip\frac{P(C_{2})}{\pi_{1,0}}\enskip{\E[y^{1}-y^{0}\mid C_{2}]}\\[4pt]
 & \enskip+\frac{P(ID_{1})}{\pi_{1,0}}\enskip{\E[y^{1}-y^{0}\mid ID_{1}]}\enskip+\enskip\frac{P(ID_{2})}{\pi_{1,0}}\enskip{\E[y^{2}-y^{0}\mid ID_{2}]}\\[4pt]
 & \enskip+\frac{P(ND_{1})}{\pi_{1,0}}\enskip{\E[y^{1}-y^{2}\mid ND_{1}]}\enskip-\enskip\frac{P(ND_{2})}{\pi_{1,0}}\enskip{\E[y^{1}-y^{2}\mid ND_{2}]}
\end{align*}
where
\begin{align*}
\pi_{1,0}=P(C_{1}\cup C_{2}\cup ID_{1}\cup ID_{2})
\end{align*}
This may be rewritten to
\begin{align*}
\tilde{\beta}_{1,0}^{IV}\enskip & =\enskip\frac{P(C_{1}\cup ID_{1})}{\pi_{1,0}}\E[y^{1}-y^{0}\mid C_{1}\cup ID_{1}]+\frac{P(C_{2}\cup ID_{1})}{\pi_{1,0}}\E[y^{2}-y^{0}\mid C_{2}\cup ID_{1}]\\[4pt]
 & \qquad+\frac{P(ND_{1})}{\pi_{1,0}}\enskip\E[y^{1}-y^{2}\mid ND_{1}]-\frac{P(ND_{2})}{\pi_{1,0}}\enskip\E[y^{1}-y^{2}\mid ND_{2}]
\end{align*}
\clearpage

\section{Examining violations of next-best and irrelevance and payoffs}

\begin{table}[H]
\caption{Joint test of irrelevance and next-best by completed field, Norway\label{tab:Joint-and-individual}}

{\small{}}%
\begin{tabular*}{1\columnwidth}{@{\extracolsep{\fill}}llr@{\extracolsep{0pt}.}lr@{\extracolsep{0pt}.}lr@{\extracolsep{0pt}.}llccc}
\toprule
 &  & \multicolumn{6}{c}{{\small{}on-diagonal ($\alpha_{jl}^{j}$)}} &  & \multicolumn{3}{c}{{\small{}off-diagonal ($\alpha_{jl}^{k},j\notin\{k,l\}$)}}\tabularnewline
{\small{}Completed field} &  & \multicolumn{2}{c}{{\small{}first stages}} & \multicolumn{2}{c}{{\small{}F-statistic}} & \multicolumn{2}{c}{{\small{}p-value}} &  & {\small{}first stages} & {\small{}F-statistic} & {\small{}p-value}\tabularnewline
\midrule
{\small{}Science} &  & \multicolumn{2}{c}{{\small{}7}} & {\small{}28}&{\small{}2} & {\small{}<0}&{\small{}01} &  & {\small{}49} & {\small{}4.7} & {\small{}<0.01}\tabularnewline
{\small{}Business} &  & \multicolumn{2}{c}{{\small{}7}} & {\small{}42}&{\small{}7} & {\small{}<0}&{\small{}01} &  & {\small{}49} & {\small{}7.6} & {\small{}<0.01}\tabularnewline
{\small{}Social Science} &  & \multicolumn{2}{c}{{\small{}7}} & {\small{}26}&{\small{}2} & {\small{}<0}&{\small{}01} &  & {\small{}49} & {\small{}8.3} & {\small{}<0.01}\tabularnewline
{\small{}Teaching} &  & \multicolumn{2}{c}{{\small{}7}} & {\small{}124}&{\small{}0} & {\small{}<0}&{\small{}01} &  & {\small{}49} & {\small{}5.2} & {\small{}<0.01}\tabularnewline
{\small{}Humanities} &  & \multicolumn{2}{c}{{\small{}7}} & {\small{}20}&{\small{}6} & {\small{}<0}&{\small{}01} &  & {\small{}49} & {\small{}5.9} & {\small{}<0.01}\tabularnewline
{\small{}Other Health} &  & \multicolumn{2}{c}{{\small{}7}} & {\small{}420}&{\small{}8} & {\small{}<0}&{\small{}01} &  & {\small{}49} & {\small{}4.7} & {\small{}<0.01}\tabularnewline
{\small{}Technology} &  & \multicolumn{2}{c}{{\small{}7}} & {\small{}50}&{\small{}6} & {\small{}<0}&{\small{}01} &  & {\small{}49} & {\small{}4.3} & {\small{}<0.01}\tabularnewline
{\small{}Law} &  & \multicolumn{2}{c}{{\small{}7}} & {\small{}94}&{\small{}9} & {\small{}<0}&{\small{}01} &  & {\small{}49} & {\small{}4.9} & {\small{}<0.01}\tabularnewline
{\small{}Medicine} &  & \multicolumn{2}{c}{{\small{}8}} & {\small{}93}&{\small{}6} & {\small{}<0}&{\small{}01} &  & {\small{}56} & {\small{}5.7} & {\small{}<0.01}\tabularnewline
 &  & \multicolumn{2}{c}{} & \multicolumn{2}{c}{} & \multicolumn{2}{c}{} &  &  &  & \tabularnewline
\emph{\small{}All} &  & \multicolumn{2}{c}{\emph{\small{}64}} & \emph{\small{}100}&\emph{\small{}65} & \emph{\small{}<0}&\emph{\small{}01} &  & \emph{\small{}448} & \emph{\small{}10.57} & \emph{\small{}<0.01}\tabularnewline
\bottomrule
\end{tabular*}{\small\par}
\end{table}

\begin{table}[H]
\caption{Off-diagonal first stages by completed field, sign and significance, Norway\label{tab:Off-diagonal-first-stages}}

{\small{}}%
\begin{tabular*}{1\columnwidth}{@{\extracolsep{\fill}}llr@{\extracolsep{0pt}.}lr@{\extracolsep{0pt}.}lr@{\extracolsep{0pt}.}lr@{\extracolsep{0pt}.}lr@{\extracolsep{0pt}.}lr@{\extracolsep{0pt}.}lr@{\extracolsep{0pt}.}l}
\toprule
 &  & \multicolumn{14}{c}{{\small{}off-diagonal, \# firsts stages that are}}\tabularnewline
\cmidrule{3-16} \cmidrule{5-16} \cmidrule{7-16} \cmidrule{9-16} \cmidrule{11-16} \cmidrule{13-16} \cmidrule{15-16}
 &  & \multicolumn{6}{c}{{\small{}>0}} & \multicolumn{2}{c}{} & \multicolumn{6}{c}{{\small{}<0}}\tabularnewline
\cmidrule{3-8} \cmidrule{5-8} \cmidrule{7-8} \cmidrule{11-16} \cmidrule{13-16} \cmidrule{15-16}
{\small{}Completed} &  & \multicolumn{2}{c}{} & \multicolumn{2}{c}{} & \multicolumn{2}{c}{{\small{}Multi}} & \multicolumn{2}{c}{} & \multicolumn{2}{c}{} & \multicolumn{2}{c}{} & \multicolumn{2}{c}{{\small{}Multi}}\tabularnewline
{\small{}field} &  & \multicolumn{2}{c}{{\small{}All}} & {\small{}Sign}& & {\small{}sign}& & \multicolumn{2}{c}{} & \multicolumn{2}{c}{{\small{}All}} & {\small{}Sign}& & {\small{}sign}&\tabularnewline
\midrule
{\small{}Science} &  & \multicolumn{2}{c}{{\small{}18}} & \multicolumn{2}{c}{{\small{}2}} & \multicolumn{2}{c}{} & \multicolumn{2}{c}{} & \multicolumn{2}{c}{{\small{}31}} & \multicolumn{2}{c}{{\small{}17}} & \multicolumn{2}{c}{{\small{}6}}\tabularnewline
{\small{}Business} &  & \multicolumn{2}{c}{{\small{}5}} & \multicolumn{2}{c}{} & \multicolumn{2}{c}{} & \multicolumn{2}{c}{} & \multicolumn{2}{c}{{\small{}44}} & \multicolumn{2}{c}{{\small{}21}} & \multicolumn{2}{c}{{\small{}8}}\tabularnewline
{\small{}Social Science} &  & \multicolumn{2}{c}{{\small{}4}} & \multicolumn{2}{c}{{\small{}1}} & \multicolumn{2}{c}{} & \multicolumn{2}{c}{} & \multicolumn{2}{c}{{\small{}45}} & \multicolumn{2}{c}{{\small{}24}} & \multicolumn{2}{c}{{\small{}12}}\tabularnewline
{\small{}Teaching} &  & \multicolumn{2}{c}{{\small{}37}} & \multicolumn{2}{c}{{\small{}13}} & \multicolumn{2}{c}{} & \multicolumn{2}{c}{} & \multicolumn{2}{c}{{\small{}12}} & \multicolumn{2}{c}{{\small{}6}} & \multicolumn{2}{c}{{\small{}5}}\tabularnewline
{\small{}Humanities} &  & \multicolumn{2}{c}{{\small{}12}} & \multicolumn{2}{c}{{\small{}3}} & \multicolumn{2}{c}{{\small{}1}} & \multicolumn{2}{c}{} & \multicolumn{2}{c}{{\small{}37}} & \multicolumn{2}{c}{{\small{}15}} & \multicolumn{2}{c}{{\small{}8}}\tabularnewline
{\small{}Other Health} &  & \multicolumn{2}{c}{{\small{}30}} & \multicolumn{2}{c}{{\small{}9}} & \multicolumn{2}{c}{} & \multicolumn{2}{c}{} & \multicolumn{2}{c}{{\small{}19}} & \multicolumn{2}{c}{{\small{}9}} & \multicolumn{2}{c}{{\small{}6}}\tabularnewline
{\small{}Technology} &  & \multicolumn{2}{c}{{\small{}11}} & \multicolumn{2}{c}{{\small{}1}} & \multicolumn{2}{c}{} & \multicolumn{2}{c}{} & \multicolumn{2}{c}{{\small{}38}} & \multicolumn{2}{c}{{\small{}14}} & \multicolumn{2}{c}{{\small{}5}}\tabularnewline
{\small{}Law} &  & \multicolumn{2}{c}{{\small{}11}} & \multicolumn{2}{c}{} & \multicolumn{2}{c}{} & \multicolumn{2}{c}{} & \multicolumn{2}{c}{{\small{}38}} & \multicolumn{2}{c}{{\small{}18}} & \multicolumn{2}{c}{{\small{}4}}\tabularnewline
{\small{}Medicine} &  & \multicolumn{2}{c}{{\small{}18}} & \multicolumn{2}{c}{{\small{}1}} & \multicolumn{2}{c}{} & \multicolumn{2}{c}{} & \multicolumn{2}{c}{{\small{}38}} & \multicolumn{2}{c}{{\small{}20}} & \multicolumn{2}{c}{{\small{}7}}\tabularnewline
 &  & \multicolumn{2}{c}{} & \multicolumn{2}{c}{} & \multicolumn{2}{c}{} & \multicolumn{2}{c}{} & \multicolumn{2}{c}{} & \multicolumn{2}{c}{} & \multicolumn{2}{c}{}\tabularnewline
\emph{\small{}All} &  & \multicolumn{2}{c}{\emph{\small{}146}} & \multicolumn{2}{c}{\emph{\small{}30}} & \multicolumn{2}{c}{\emph{\small{}1}} & \multicolumn{2}{c}{} & \multicolumn{2}{c}{\emph{\small{}302}} & \multicolumn{2}{c}{\emph{\small{}144}} & \multicolumn{2}{c}{\emph{\small{}61}}\tabularnewline
\bottomrule
\end{tabular*}{\small\par}

\noindent\begin{minipage}[t]{1\columnwidth}%
{\footnotesize{}Note: Significant are first stages with $p<.05$, multi test-significant are first stages with $p<.05/512$ (i.e. significant with a Bonferroni adjustment for multiple testing).}%
\end{minipage}
\end{table}

\begin{table}
\caption{Joint test of irrelevance and next-best by completed field, Denmark\label{tab:Joint-and-individual-dk}}

{\small{}}%
\begin{tabular*}{1\columnwidth}{@{\extracolsep{\fill}}lr@{\extracolsep{0pt}.}lr@{\extracolsep{0pt}.}lr@{\extracolsep{0pt}.}llccr@{\extracolsep{0pt}.}l}
\toprule
 & \multicolumn{6}{c}{{\small{}on-diagonal ($\alpha_{jl}^{j}$ )}} &  & \multicolumn{4}{c}{{\small{}off-diagonal ($\alpha_{jl}^{k},j\notin\{k,l\}$)}}\tabularnewline
\cmidrule{2-7} \cmidrule{4-7} \cmidrule{6-7} \cmidrule{9-12} \cmidrule{10-12} \cmidrule{11-12}
{\small{}Completed field} & \multicolumn{2}{c}{{\small{}first stages}} & \multicolumn{2}{c}{{\small{}F-statistic}} & \multicolumn{2}{c}{{\small{}p-value}} &  & {\small{}first stages} & {\small{}F-statistic} & \multicolumn{2}{c}{{\small{}p-value}}\tabularnewline
\midrule
{\small{}Science} & \multicolumn{2}{c}{{\small{}7}} & {\small{}95}&{\small{}3} & {\small{}<0}&{\small{}01} &  & {\small{}49} & {\small{}2.1} & {\small{}<0}&{\small{}01}\tabularnewline
{\small{}Business} & \multicolumn{2}{c}{{\small{}7}} & {\small{}2}&{\small{}8} & {\small{}<0}&{\small{}01} &  & {\small{}49} & {\small{}1.9} & {\small{}<0}&{\small{}01}\tabularnewline
{\small{}Social Science} & \multicolumn{2}{c}{{\small{}7}} & {\small{}111}&{\small{}5} & {\small{}<0}&{\small{}01} &  & {\small{}49} & {\small{}3.9} & {\small{}<0}&{\small{}01}\tabularnewline
{\small{}Teaching} & \multicolumn{2}{c}{{\small{}7}} & {\small{}12}&{\small{}7} & {\small{}<0}&{\small{}01} &  & {\small{}49} & {\small{}2.6} & {\small{}<0}&{\small{}01}\tabularnewline
{\small{}Humanities} & \multicolumn{2}{c}{{\small{}7}} & {\small{}33}&{\small{}7} & {\small{}<0}&{\small{}01} &  & {\small{}49} & {\small{}1.8} & {\small{}<0}&{\small{}01}\tabularnewline
{\small{}Other Health} & \multicolumn{2}{c}{{\small{}7}} & {\small{}61}&{\small{}5} & {\small{}<0}&{\small{}01} &  & {\small{}49} & {\small{}2.8} & {\small{}<0}&{\small{}01}\tabularnewline
{\small{}Technology} & \multicolumn{2}{c}{{\small{}7}} & {\small{}38}&{\small{}2} & {\small{}<0}&{\small{}01} &  & {\small{}49} & {\small{}2.2} & {\small{}<0}&{\small{}01}\tabularnewline
{\small{}Law} & \multicolumn{2}{c}{{\small{}7}} & {\small{}63}&{\small{}3} & {\small{}<0}&{\small{}01} &  & {\small{}49} & {\small{}1.4} & {\small{}0}&{\small{}04}\tabularnewline
{\small{}Medicine} & \multicolumn{2}{c}{{\small{}8}} & {\small{}98}&{\small{}8} & {\small{}<0}&{\small{}01} &  & {\small{}56} & {\small{}1.6} & {\small{}<0}&{\small{}01}\tabularnewline
 & \multicolumn{2}{c}{} & \multicolumn{2}{c}{} & \multicolumn{2}{c}{} &  &  &  & \multicolumn{2}{c}{}\tabularnewline
\emph{\small{}All} & \multicolumn{2}{c}{\emph{\small{}64}} & \emph{\small{}57}&\emph{\small{}9} & \emph{\small{}<0}&\emph{\small{}01} &  & \emph{\small{}448} & \emph{\small{}3.3} & \emph{\small{}<0}&\emph{\small{}01}\tabularnewline
\bottomrule
\end{tabular*}{\small\par}
\end{table}

\begin{table}
\caption{Off-diagonal first stages by completed field, sign and significance, Denmark\label{tab:Off-diagonal-first-stages-dk}}

{\small{}}%
\begin{tabular*}{1\columnwidth}{@{\extracolsep{\fill}}lr@{\extracolsep{0pt}.}lr@{\extracolsep{0pt}.}lr@{\extracolsep{0pt}.}lr@{\extracolsep{0pt}.}lr@{\extracolsep{0pt}.}lr@{\extracolsep{0pt}.}lr@{\extracolsep{0pt}.}l}
\toprule
 & \multicolumn{14}{c}{{\small{}off-diagonal, \# firsts stages that are}}\tabularnewline
\cmidrule{2-15} \cmidrule{4-15} \cmidrule{6-15} \cmidrule{8-15} \cmidrule{10-15} \cmidrule{12-15} \cmidrule{14-15}
 & \multicolumn{6}{c}{{\small{}>0}} & \multicolumn{2}{c}{} & \multicolumn{6}{c}{{\small{}<0}}\tabularnewline
\cmidrule{2-7} \cmidrule{4-7} \cmidrule{6-7} \cmidrule{10-15} \cmidrule{12-15} \cmidrule{14-15}
 & \multicolumn{2}{c}{} & \multicolumn{2}{c}{} & \multicolumn{2}{c}{{\small{}Multi}} & \multicolumn{2}{c}{} & \multicolumn{2}{c}{} & \multicolumn{2}{c}{} & \multicolumn{2}{c}{{\small{}Multi}}\tabularnewline
{\small{}Completed field} & \multicolumn{2}{c}{{\small{}All}} & {\small{}Sign}& & {\small{}sign}& & \multicolumn{2}{c}{} & \multicolumn{2}{c}{{\small{}All}} & {\small{}Sign}& & {\small{}sign}&\tabularnewline
\midrule
{\small{}Science} & \multicolumn{2}{c}{{\small{}26}} & \multicolumn{2}{c}{{\small{}7}} & \multicolumn{2}{c}{} & \multicolumn{2}{c}{} & \multicolumn{2}{c}{{\small{}23}} & \multicolumn{2}{c}{{\small{}8}} & \multicolumn{2}{c}{{\small{}1}}\tabularnewline
{\small{}Business} & \multicolumn{2}{c}{{\small{}23}} & \multicolumn{2}{c}{{\small{}5}} & \multicolumn{2}{c}{{\small{}1}} & \multicolumn{2}{c}{} & \multicolumn{2}{c}{{\small{}26}} & \multicolumn{2}{c}{{\small{}6}} & \multicolumn{2}{c}{{\small{}1}}\tabularnewline
{\small{}Social Science} & \multicolumn{2}{c}{{\small{}29}} & \multicolumn{2}{c}{{\small{}8}} & \multicolumn{2}{c}{{\small{}2}} & \multicolumn{2}{c}{} & \multicolumn{2}{c}{{\small{}20}} & \multicolumn{2}{c}{{\small{}10}} & \multicolumn{2}{c}{{\small{}5}}\tabularnewline
{\small{}Teaching} & \multicolumn{2}{c}{{\small{}30}} & \multicolumn{2}{c}{{\small{}4}} & \multicolumn{2}{c}{} & \multicolumn{2}{c}{} & \multicolumn{2}{c}{{\small{}19}} & \multicolumn{2}{c}{{\small{}4}} & \multicolumn{2}{c}{{\small{}1}}\tabularnewline
{\small{}Humanities} & \multicolumn{2}{c}{{\small{}32}} & \multicolumn{2}{c}{{\small{}8}} & \multicolumn{2}{c}{} & \multicolumn{2}{c}{} & \multicolumn{2}{c}{{\small{}17}} & \multicolumn{2}{c}{{\small{}2}} & \multicolumn{2}{c}{{\small{}            }}\tabularnewline
{\small{}Other Health} & \multicolumn{2}{c}{{\small{}30}} & \multicolumn{2}{c}{{\small{}7}} & \multicolumn{2}{c}{{\small{}3}} & \multicolumn{2}{c}{} & \multicolumn{2}{c}{{\small{}19}} & \multicolumn{2}{c}{{\small{}2}} & \multicolumn{2}{c}{{\small{}            }}\tabularnewline
{\small{}Technology} & \multicolumn{2}{c}{{\small{}26}} & \multicolumn{2}{c}{{\small{}2}} & \multicolumn{2}{c}{} & \multicolumn{2}{c}{} & \multicolumn{2}{c}{{\small{}23}} & \multicolumn{2}{c}{{\small{}7}} & \multicolumn{2}{c}{{\small{}            }}\tabularnewline
{\small{}Law} & \multicolumn{2}{c}{{\small{}22}} & \multicolumn{2}{c}{{\small{}2}} & \multicolumn{2}{c}{} & \multicolumn{2}{c}{} & \multicolumn{2}{c}{{\small{}27}} & \multicolumn{2}{c}{{\small{}6}} & \multicolumn{2}{c}{{\small{}            }}\tabularnewline
{\small{}Medicine} & \multicolumn{2}{c}{{\small{}24}} & \multicolumn{2}{c}{{\small{}2}} & \multicolumn{2}{c}{} & \multicolumn{2}{c}{} & \multicolumn{2}{c}{{\small{}32}} & \multicolumn{2}{c}{{\small{}13}} & \multicolumn{2}{c}{{\small{}1}}\tabularnewline
 & \multicolumn{2}{c}{} & \multicolumn{2}{c}{} & \multicolumn{2}{c}{} & \multicolumn{2}{c}{} & \multicolumn{2}{c}{} & \multicolumn{2}{c}{} & \multicolumn{2}{c}{}\tabularnewline
\emph{\small{}Sum} & \multicolumn{2}{c}{\emph{\small{}242}} & \multicolumn{2}{c}{\emph{\small{}45}} & \multicolumn{2}{c}{\emph{\small{}6}} & \multicolumn{2}{c}{} & \multicolumn{2}{c}{\emph{\small{}206}} & \multicolumn{2}{c}{\emph{\small{}58}} & \multicolumn{2}{c}{\emph{\small{}9}}\tabularnewline
\bottomrule
\end{tabular*}{\small\par}

\emph{\footnotesize{}Note:}{\footnotesize{} Significant are first stages with $p<.05$, multi test-significant are first stages with $p<.05/512$ (i.e. significant with Bonferroni adjustment for multiple testing).}{\footnotesize\par}
\end{table}

\clearpage

\section{The Danish institutional setting and implications for empirical specification\label{sec:Danish-institutions}}

\paragraph*{Danish admission institutions}

For programs with restricted admission, student places are allocated through two quotas.

The majority of places are allocated through Quota 1 based on applicants\textquoteright{} GPA from high school, although some programs have additional specific requirements, e.g., a high-level maths course from high school. In our data period grades were awarded on a 10-point scale with integer values between 0 and 13 (omitting values 1, 2, 4 and 12). The high school GPA is based on grades in all the subjects on the student\textquoteright s study program, and it is recorded to 1 decimal place. All applicants with a GPA strictly above the Quota 1 threshold level are admitted provided they also meet any specific entry requirements. Because of the coarse GPA measure, there will often not be sufficient places for all applicants with a GPA exactly equal to the threshold, and in this case the oldest are typically admitted first. In our sample, about one third of the applicants with a GPA exactly equal to the threshold are not admitted to their preferred program. It is important to note that the minimum GPAs needed to be admitted (the GPA thresholds) are published after the student places are allocated and that the variation in thresholds over time is considerable, and thus applicants cannot predict the exact thresholds.

Most programs also have a standby (waiting) list and the GPA threshold for the standby list is typically a little lower than the Quota 1 threshold. On the application form, applicants can choose to apply for the standby list as well. If some of the applicants admitted under Quota 1 drop out before the course starts or in the very early days of the course, then their places are offered to applicants on the standby list. In any event, applicants admitted to the standby list are guaranteed a study place the following year. Applicants who are admitted to a standby list are not considered for any of the lower-ranked programs on their application.

Some of the available student places are reserved for admission via Quota 2 where applicants are assessed based on other criteria besides their GPA from upper secondary school, e.g. specific admission tests, admission interviews, or vocational qualifications. The institutions have considerable discretion in deciding these criteria and the relative size of Quotas 1 and 2 for their programs. Quota 2 applicants are automatically considered for Quota 1, so if they meet the Quota 1 GPA criterion (and any additional specific criteria), they are admitted via Quota 1.

We do not observe whether students apply for Quota 2 or standby, but for those who are offered a program, we know whether admission was via Quota 1, standby or Quota 2.

\paragraph{}

\paragraph{Specification of instrument}

The coarse application score variable and the standby list option are handled by using the following instrument:
\begin{align*}
z & =[GPA>Q1]+[Std\leq GPA\leq Q1]\times\text{Offer}
\end{align*}
where $Q1$ and $Stb$ are the Quota 1 and standby admission thresholds of the preferred field of the applicant, respectively. The first term on the right-hand side captures the effect of having an application score strictly above the Quota 1 threshold to the preferred field k. When $GPA=Q1$ or $GPA=Std$ the offer is a random draw (conditional on age which is controlled for in the analysis, and for the standby threshold also conditional on having applied for standby), and $[GPA\in\{Q1,Std\}]\times\text{Offer}$ is therefore a valid instrument. When $[Std<GPA<Q1]$ those who applied to the standby list will receive an offer.\clearpage

\section{Extra results}

\begin{figure}[H]
\begin{centering}
\subfloat[Number in population]{\includegraphics[width=0.5\textwidth]{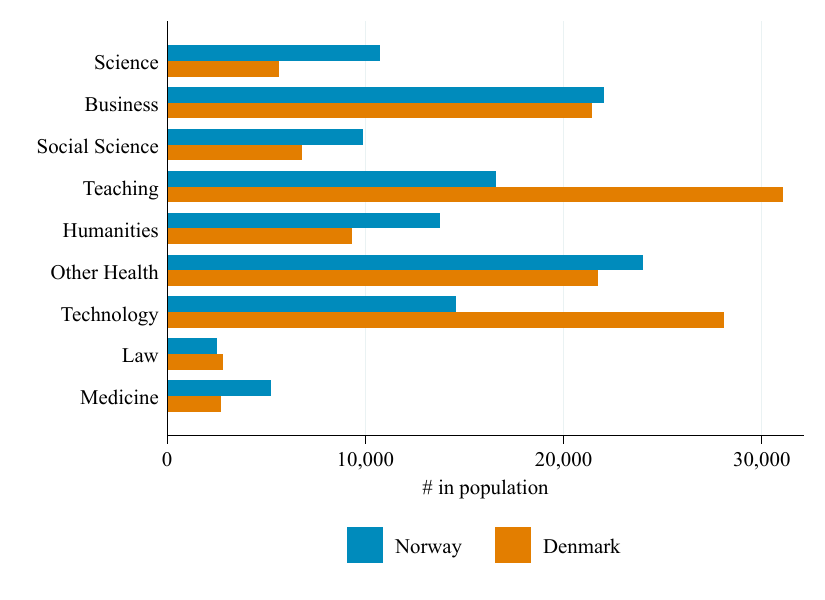}

}\subfloat[Share in population]{\includegraphics[width=0.5\textwidth]{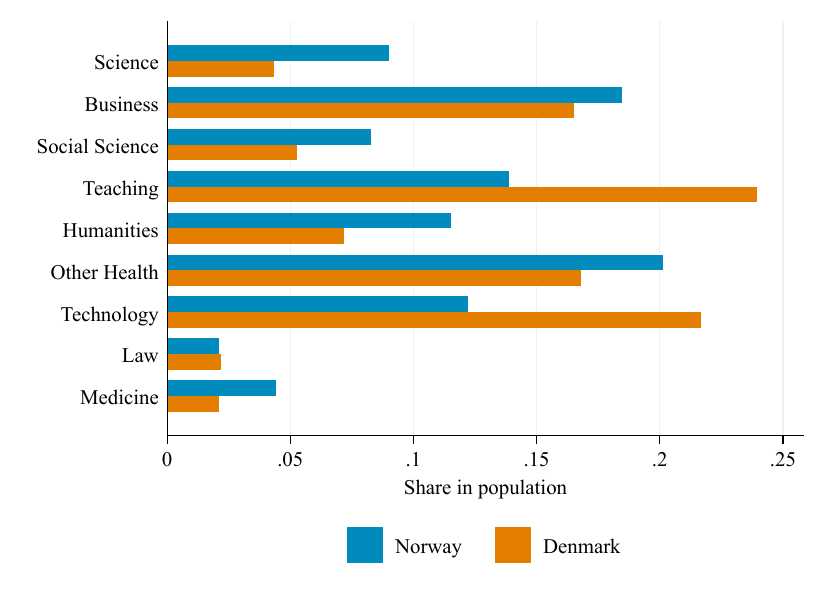}

}
\par\end{centering}
\noindent\begin{minipage}[t]{1\columnwidth}%
\begin{singlespace}
{\footnotesize{}Note: Figures show the distribution of completed fields in the population for those born 1979-1983 (Norway) or 1975-81 (Denmark) that have completed any higher education. Completed field is measured at age 28.}
\end{singlespace}
\end{minipage}

\caption{Distribution of completed field in the population\label{fig:Distribution-of-completed}}
\end{figure}

\begin{figure}[H]
\begin{centering}
\subfloat[GPA]{\includegraphics[width=0.5\textwidth]{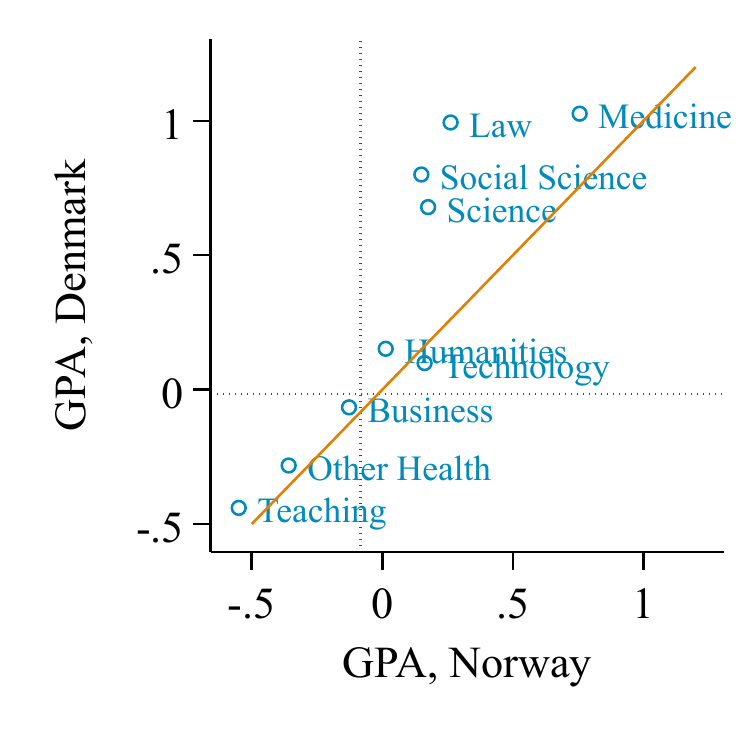}

}\subfloat[Earnings]{\includegraphics[width=0.5\textwidth]{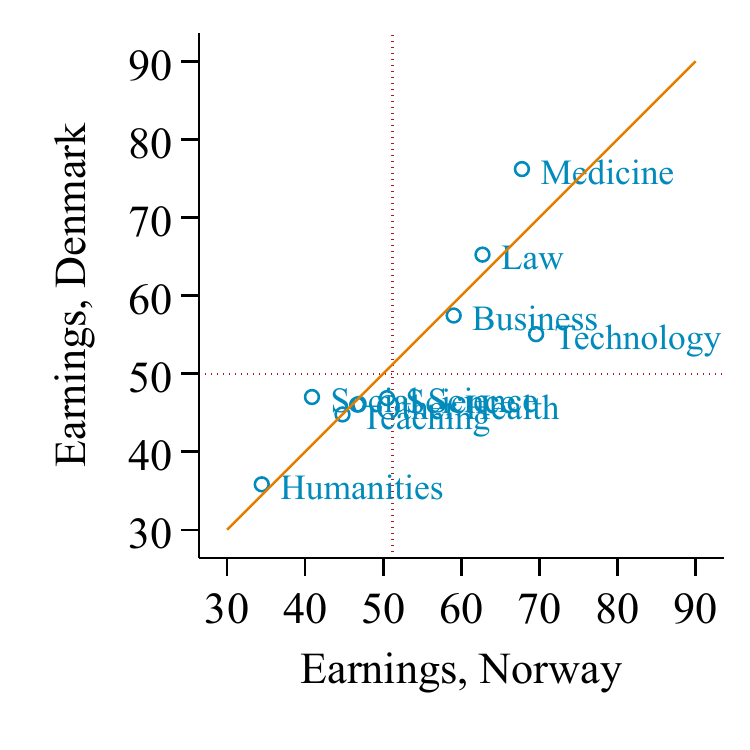}

}
\par\end{centering}
\noindent\begin{minipage}[t]{1\columnwidth}%
\begin{singlespace}
{\footnotesize{}Note: Figure shows population-weighted average GPA of applicants and population earnings, based on the populations in Figure \ref{fig:Distribution-of-completed}. GPA is demeaned within country, but not otherwise standardized. Earnings are measured at age 28.}
\end{singlespace}
\end{minipage}

\caption{GPA and earnings in Norway and Denmark, average by country and field\label{fig:GPA-and-earnings-1}}
\end{figure}

\begin{figure}
\subfloat[Numbers of applicants shifted by preferred field]{\includegraphics[width=0.5\textwidth]{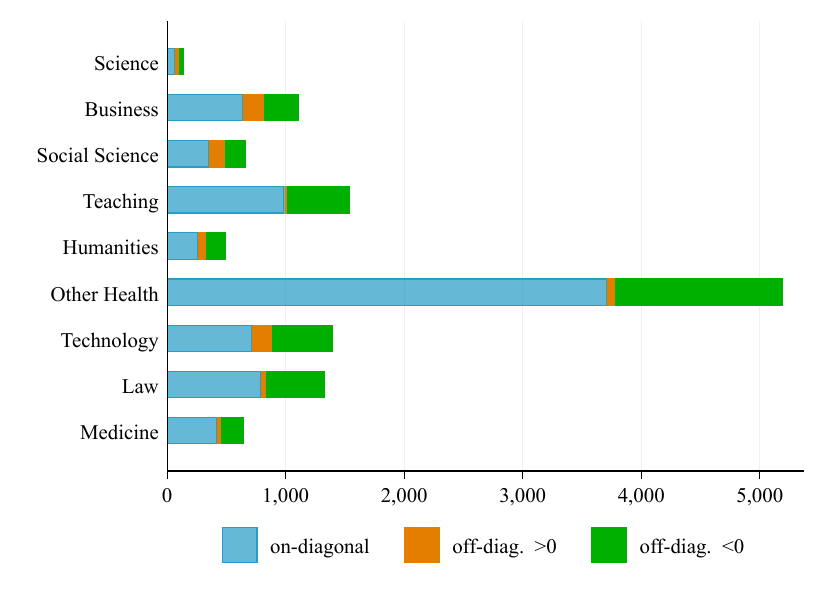}

}\subfloat[Share of applicants shifted by preferred field]{\includegraphics[width=0.5\textwidth]{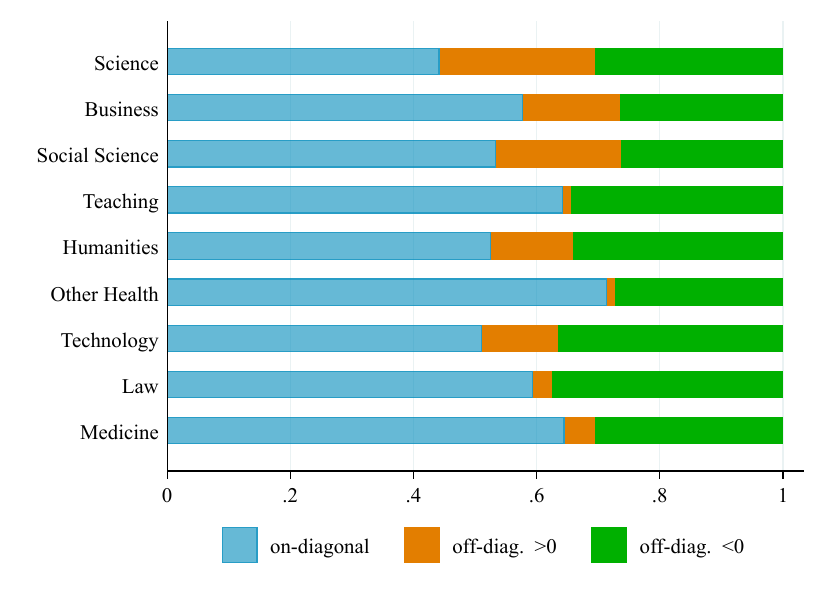}

}

\subfloat[Numbers of applicants shifted by next-best field]{\includegraphics[width=0.5\textwidth]{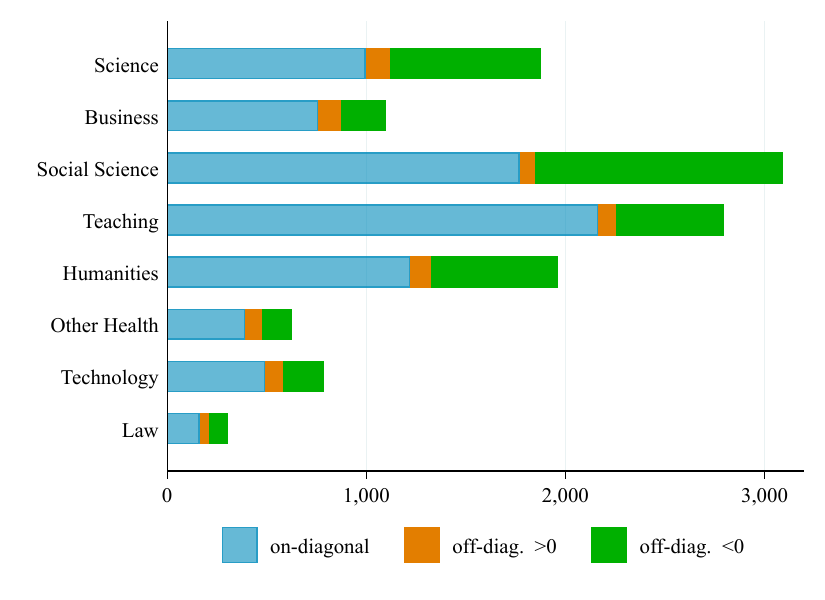}

}\subfloat[Share of applicants shifted by next-best field]{\includegraphics[width=0.5\textwidth]{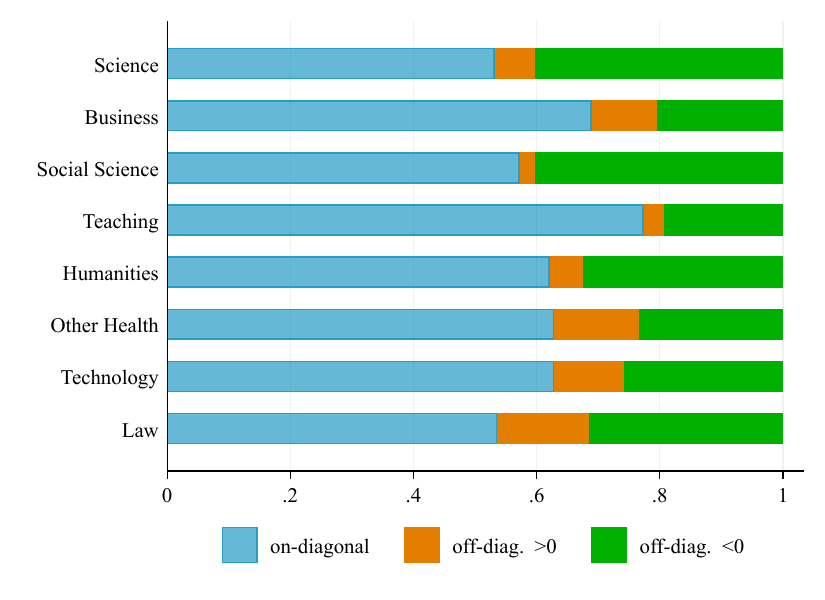}

}

\caption{Numbers and shares of applicants shifted by the instrument by violating irrelevance or not and by preferred/stated next-best field, Norway\label{fig:Numbers-and-shares}}
\end{figure}

\begin{figure}
\subfloat[Numbers of applicants shifted by preferred field]{\begin{centering}
\includegraphics[width=0.5\textwidth]{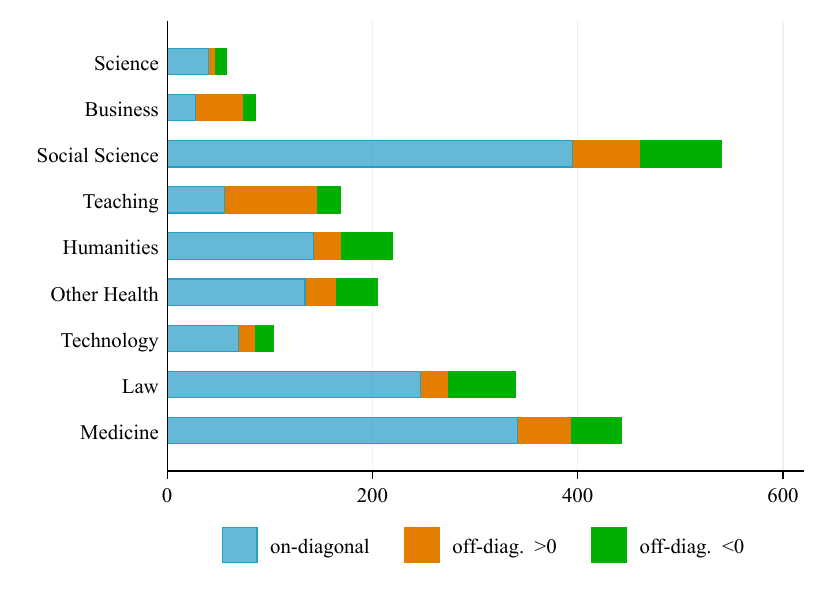}
\par\end{centering}
}\subfloat[Share of applicants shifted by preferred field]{\begin{centering}
\includegraphics[width=0.5\textwidth]{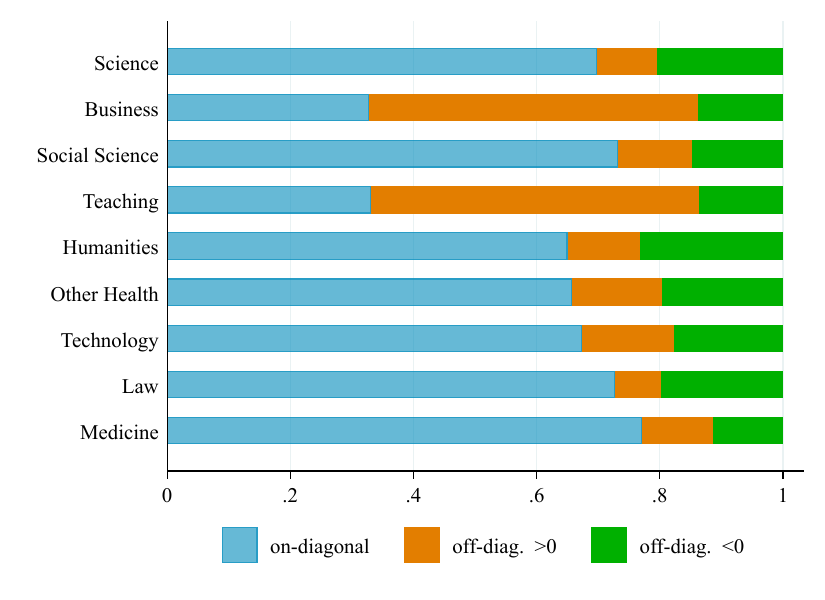}
\par\end{centering}
}

\subfloat[Numbers of applicants shifted by next-best field]{\begin{centering}
\includegraphics[width=0.5\textwidth]{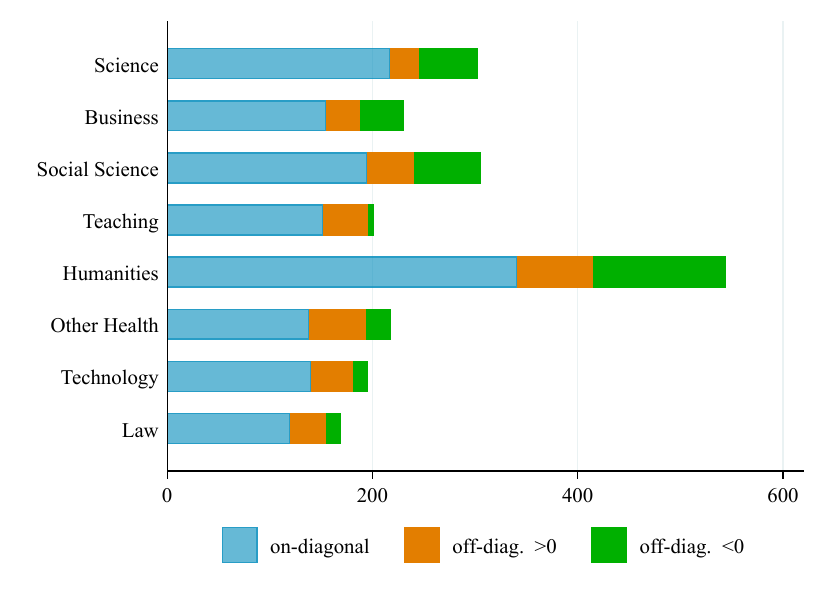}
\par\end{centering}
}\subfloat[Share of applicants shifted by next-best field]{\begin{centering}
\includegraphics[width=0.5\textwidth]{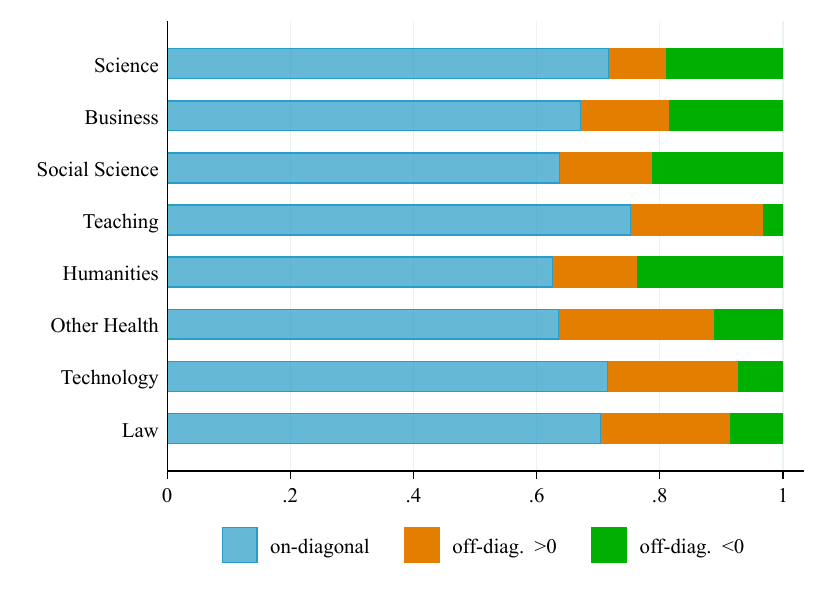}
\par\end{centering}
}

\caption{Numbers and shares of applicants shifted by the instrument by violating irrelevance or not and by preferred/stated next-best field, Denmark\label{fig:Numbers-and-shares-1}}
\end{figure}

\begin{figure}
\begin{centering}
\includegraphics[width=0.5\textwidth]{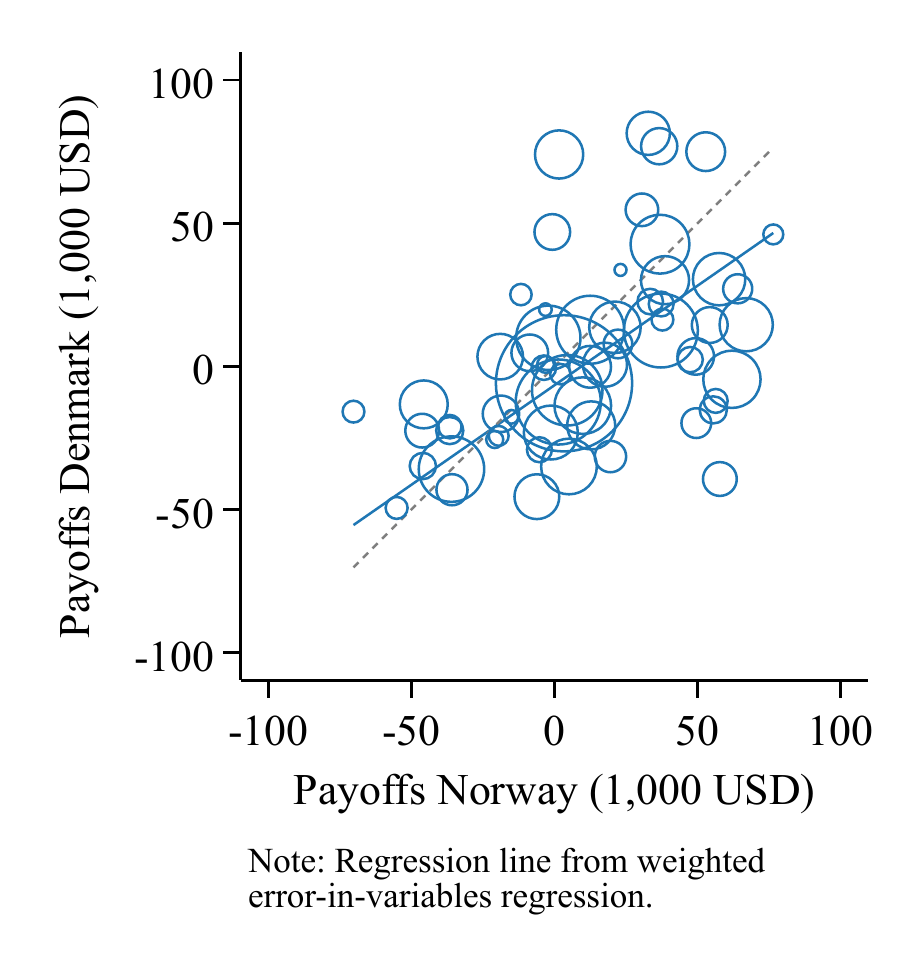}
\par\end{centering}
\caption{Payoffs in Norway and Denmark 13 years after applying, all completed and next-best fields\label{fig:Payoffs-in-Norway-1}}
\end{figure}

\begin{sidewaystable}
\caption{Payoff estimates Norway (\$1,000), 8 years after application\label{tab:Payoffs No t8}}

{\small{}}%
\begin{tabular*}{1\textwidth}{@{\extracolsep{\fill}}lcccccccc}
\toprule
 & \multicolumn{8}{c}{{\small{}Next-best field}}\tabularnewline
\cmidrule{2-9} \cmidrule{3-9} \cmidrule{4-9} \cmidrule{5-9} \cmidrule{6-9} \cmidrule{7-9} \cmidrule{8-9} \cmidrule{9-9}
 & {\small{}Science} & {\small{}Business} & {\small{}Social Science} & {\small{}Teaching} & {\small{}Humanities} & {\small{}Other Health} & {\small{}Technology} & {\small{}Law}\tabularnewline
\midrule
{\small{}Completed field} &  &  &  &  &  &  &  & \tabularnewline
{\small{}Science} &  & {\small{}-22.9} & {\small{}27.0} & {\small{}4.8} & {\small{}8.4} & {\small{}4.0} & {\small{}-17.3} & {\small{}145.0}\tabularnewline
 &  & {\small{}(17.9)} & {\small{}(25.8)} & {\small{}(20.1)} & {\small{}(15.6)} & {\small{}(15.5)} & {\small{}(16.3)} & {\small{}(167.1)}\tabularnewline
{\small{}Business} & {\small{}70.8} &  & {\small{}59.3} & {\small{}26.3} & {\small{}41.6} & {\small{}26.9} & {\small{}-0.0} & {\small{}24.0}\tabularnewline
 & {\small{}(10.5)} &  & {\small{}(9.5)} & {\small{}(6.6)} & {\small{}(7.9)} & {\small{}(9.0)} & {\small{}(6.1)} & {\small{}(32.5)}\tabularnewline
{\small{}Social Science} & {\small{}63.9} & {\small{}-33.6} &  & {\small{}5.0} & {\small{}13.4} & {\small{}-11.1} & {\small{}-45.4} & {\small{}-48.5}\tabularnewline
 & {\small{}(18.6)} & {\small{}(10.3)} &  & {\small{}(10.5)} & {\small{}(6.9)} & {\small{}(12.4)} & {\small{}(23.3)} & {\small{}(53.4)}\tabularnewline
{\small{}Teaching} & {\small{}47.3} & {\small{}-21.3} & {\small{}32.3} &  & {\small{}19.3} & {\small{}0.4} & {\small{}-29.6} & {\small{}7.8}\tabularnewline
 & {\small{}(9.7)} & {\small{}(5.5)} & {\small{}(6.9)} &  & {\small{}(4.4)} & {\small{}(4.8)} & {\small{}(6.7)} & {\small{}(46.6)}\tabularnewline
{\small{}Humanities} & {\small{}8.2} & {\small{}-36.2} & {\small{}24.7} & {\small{}-5.1} &  & {\small{}-22.5} & {\small{}-37.7} & {\small{}-87.7}\tabularnewline
 & {\small{}(15.2)} & {\small{}(8.5)} & {\small{}(9.9)} & {\small{}(7.9)} &  & {\small{}(11.3)} & {\small{}(7.8)} & {\small{}(87.6)}\tabularnewline
{\small{}Other Health} & {\small{}47.5} & {\small{}-17.1} & {\small{}32.2} & {\small{}6.3} & {\small{}17.0} &  & {\small{}-26.0} & {\small{}-28.2}\tabularnewline
 & {\small{}(9.5)} & {\small{}(3.8)} & {\small{}(6.8)} & {\small{}(2.3)} & {\small{}(4.8)} &  & {\small{}(4.7)} & {\small{}(39.9)}\tabularnewline
{\small{}Technology} & {\small{}87.4} & {\small{}-5.6} & {\small{}66.1} & {\small{}35.8} & {\small{}57.9} & {\small{}27.8} &  & {\small{}-20.5}\tabularnewline
 & {\small{}(11.9)} & {\small{}(7.1)} & {\small{}(9.5)} & {\small{}(8.0)} & {\small{}(9.0)} & {\small{}(8.3)} &  & {\small{}(57.0)}\tabularnewline
{\small{}Law} & {\small{}52.2} & {\small{}-11.9} & {\small{}53.7} & {\small{}28.8} & {\small{}40.3} & {\small{}22.3} & {\small{}-18.6} & \tabularnewline
 & {\small{}(11.0)} & {\small{}(7.9)} & {\small{}(7.5)} & {\small{}(10.6)} & {\small{}(5.9)} & {\small{}(10.7)} & {\small{}(9.6)} & \tabularnewline
{\small{}Medicine} & {\small{}97.9} & {\small{}20.4} & {\small{}76.6} & {\small{}55.9} & {\small{}74.3} & {\small{}42.1} & {\small{}20.5} & {\small{}42.1}\tabularnewline
 & {\small{}(11.3)} & {\small{}(9.1)} & {\small{}(9.4)} & {\small{}(8.5)} & {\small{}(8.3)} & {\small{}(6.6)} & {\small{}(5.8)} & {\small{}(28.8)}\tabularnewline
 &  &  &  &  &  &  &  & \tabularnewline
{\small{}N} & {\small{}5,320} & {\small{}4,477} & {\small{}11,250} & {\small{}11,254} & {\small{}8,539} & {\small{}3,371} & {\small{}3,612} & {\small{}1,403}\tabularnewline
\bottomrule
\end{tabular*}{\small\par}

\medskip{}

\begin{singlespace}
\textbf{\footnotesize{}Note:}{\footnotesize{} 2SLS estimation of equations (\ref{eq:second stage-1})--\ref{eq:firststages-1} results in a matrix of payoffs to field $j$ as compared to $k$ for those who prefer $j$ and have $k$ as next-best field. Each cell is a 2SLS estimate (with standard errors in parenthesis) of the payoff to a given pair of preferred field and next-best field. The rows represent completed fields and the columns represent next-best fields.}{\footnotesize\par}
\end{singlespace}
\end{sidewaystable}

\begin{sidewaystable}

\caption{Payoff estimates Denmark (\$1,000), 8 years after application\label{tab:Payoffs Dk t8}}

{\small{}}%
\begin{tabular*}{1\textwidth}{@{\extracolsep{\fill}}lcccccccc}
\toprule
 & \multicolumn{8}{c}{{\small{}Next-best field}}\tabularnewline
\cmidrule{2-9} \cmidrule{3-9} \cmidrule{4-9} \cmidrule{5-9} \cmidrule{6-9} \cmidrule{7-9} \cmidrule{8-9} \cmidrule{9-9}
 & {\small{}Science} & {\small{}Business} & {\small{}Social Science} & {\small{}Teaching} & {\small{}Humanities} & {\small{}Other Health} & {\small{}Technology} & {\small{}Law}\tabularnewline
\midrule
{\small{}Completed field} &  &  &  &  &  &  &  & \tabularnewline
{\small{}Science} &  & {\small{}-10.8} & {\small{}-10.4} & {\small{}-4.5} & {\small{}0.5} & {\small{}5.6} & {\small{}-23.5} & {\small{}-5.8}\tabularnewline
 &  & {\small{}(9.6)} & {\small{}(10.0)} & {\small{}(16.8)} & {\small{}(10.0)} & {\small{}(7.9)} & {\small{}(9.4)} & {\small{}(6.7)}\tabularnewline
{\small{}Business} & {\small{}61.7} &  & {\small{}18.3} & {\small{}31.7} & {\small{}54.7} & {\small{}25.5} & {\small{}33.6} & {\small{}16.6}\tabularnewline
 & {\small{}(22.9)} &  & {\small{}(17.4)} & {\small{}(26.2)} & {\small{}(20.5)} & {\small{}(29.6)} & {\small{}(23.3)} & {\small{}(18.6)}\tabularnewline
{\small{}Social Science} & {\small{}-15.2} & {\small{}-23.3} &  & {\small{}-4.3} & {\small{}15.5} & {\small{}-4.4} & {\small{}-8.3} & {\small{}-8.0}\tabularnewline
 & {\small{}(10.2)} & {\small{}(8.1)} &  & {\small{}(4.3)} & {\small{}(5.0)} & {\small{}(4.4)} & {\small{}(8.6)} & {\small{}(4.0)}\tabularnewline
{\small{}Teaching} & {\small{}-9.0} & {\small{}-35.5} & {\small{}-16.9} &  & {\small{}-0.5} & {\small{}-21.2} & {\small{}-23.0} & {\small{}-16.9}\tabularnewline
 & {\small{}(10.3)} & {\small{}(11.1)} & {\small{}(8.1)} &  & {\small{}(8.8)} & {\small{}(8.3)} & {\small{}(8.4)} & {\small{}(9.5)}\tabularnewline
{\small{}Humanities} & {\small{}-9.6} & {\small{}-32.2} & {\small{}-31.3} & {\small{}-9.5} &  & {\small{}-24.1} & {\small{}-28.9} & {\small{}-31.4}\tabularnewline
 & {\small{}(8.3)} & {\small{}(7.7)} & {\small{}(10.0)} & {\small{}(5.8)} &  & {\small{}(10.7)} & {\small{}(10.9)} & {\small{}(9.8)}\tabularnewline
{\small{}Other Health} & {\small{}15.0} & {\small{}-24.7} & {\small{}-3.7} & {\small{}-3.6} & {\small{}16.4} &  & {\small{}-18.6} & {\small{}-20.2}\tabularnewline
 & {\small{}(10.6)} & {\small{}(9.3)} & {\small{}(7.2)} & {\small{}(2.9)} & {\small{}(5.3)} &  & {\small{}(8.9)} & {\small{}(6.8)}\tabularnewline
{\small{}Technology} & {\small{}7.1} & {\small{}-17.4} & {\small{}-13.7} & {\small{}-17.2} & {\small{}9.9} & {\small{}-7.8} &  & {\small{}-4.9}\tabularnewline
 & {\small{}(7.9)} & {\small{}(9.5)} & {\small{}(9.6)} & {\small{}(19.1)} & {\small{}(7.9)} & {\small{}(8.8)} &  & {\small{}(16.6)}\tabularnewline
{\small{}Law} & {\small{}22.4} & {\small{}-4.1} & {\small{}12.4} & {\small{}1.5} & {\small{}31.9} & {\small{}16.6} & {\small{}9.3} & \tabularnewline
 & {\small{}(6.2)} & {\small{}(5.2)} & {\small{}(6.3)} & {\small{}(8.8)} & {\small{}(5.6)} & {\small{}(4.7)} & {\small{}(8.6)} & \tabularnewline
{\small{}Medicine} & {\small{}21.7} & {\small{}5.4} & {\small{}7.1} & {\small{}29.4} & {\small{}28.9} & {\small{}11.9} & {\small{}13.7} & {\small{}14.1}\tabularnewline
 & {\small{}(5.1)} & {\small{}(8.2)} & {\small{}(6.9)} & {\small{}(8.0)} & {\small{}(6.1)} & {\small{}(4.3)} & {\small{}(4.5)} & {\small{}(5.9)}\tabularnewline
 &  &  &  &  &  &  &  & \tabularnewline
{\small{}Total} & {\small{}1,254} & {\small{}1,076} & {\small{}1,427} & {\small{}1,359} & {\small{}2,761} & {\small{}1,521} & {\small{}738} & {\small{}432}\tabularnewline
\bottomrule
\end{tabular*}{\small\par}

\medskip{}

\begin{singlespace}
\textbf{\footnotesize{}Note:}{\footnotesize{} 2SLS estimation of equations (\ref{eq:second stage-1})--\ref{eq:firststages-1} results in a matrix of payoffs to field $j$ as compared to $k$ for those who prefer $j$ and have $k$ as next-best field. Each cell is a 2SLS estimate (with standard errors in parenthesis) of the payoff to a given pair of preferred field and next-best field. The rows represent completed fields and the columns represent next-best fields.}{\footnotesize\par}
\end{singlespace}
\end{sidewaystable}

\begin{sidewaystable}
\caption{Payoff estimates Norway (\$1,000), 13 years after application\label{tab:Payoffs No t13}}

{\small{}}%
\begin{tabular*}{1\textwidth}{@{\extracolsep{\fill}}lcccccccc}
\toprule
 & \multicolumn{8}{c}{{\small{}Next-best field}}\tabularnewline
\cmidrule{2-9} \cmidrule{3-9} \cmidrule{4-9} \cmidrule{5-9} \cmidrule{6-9} \cmidrule{7-9} \cmidrule{8-9} \cmidrule{9-9}
 & {\small{}Science} & {\small{}Business} & {\small{}Social Science} & {\small{}Teaching} & {\small{}Humanities} & {\small{}Other Health} & {\small{}Technology} & {\small{}Law}\tabularnewline
\midrule
{\small{}Completed field} &  &  &  &  &  &  &  & \tabularnewline
{\small{}Science} &  & {\small{}-15.1} & {\small{}-3.2} & {\small{}23.1} & {\small{}-3.1} & {\small{}2.1} & {\small{}-19.5} & {\small{}-10.3}\tabularnewline
 &  & {\small{}(22.0)} & {\small{}(29.6)} & {\small{}(27.9)} & {\small{}(20.6)} & {\small{}(22.8)} & {\small{}(19.7)} & {\small{}(185.7)}\tabularnewline
{\small{}Business} & {\small{}52.8} &  & {\small{}36.9} & {\small{}36.6} & {\small{}32.8} & {\small{}30.5} & {\small{}1.6} & {\small{}-0.8}\tabularnewline
 & {\small{}(18.5)} &  & {\small{}(16.6)} & {\small{}(12.0)} & {\small{}(14.8)} & {\small{}(13.2)} & {\small{}(9.9)} & {\small{}(63.1)}\tabularnewline
{\small{}Social Science} & {\small{}55.5} & {\small{}-46.2} &  & {\small{}17.6} & {\small{}12.5} & {\small{}-8.6} & {\small{}-70.2} & {\small{}-18.7}\tabularnewline
 & {\small{}(26.8)} & {\small{}(15.9)} &  & {\small{}(15.6)} & {\small{}(10.2)} & {\small{}(19.5)} & {\small{}(29.2)} & {\small{}(71.3)}\tabularnewline
{\small{}Teaching} & {\small{}19.5} & {\small{}-35.8} & {\small{}5.1} &  & {\small{}4.4} & {\small{}-1.3} & {\small{}-46.0} & {\small{}9.8}\tabularnewline
 & {\small{}(16.0)} & {\small{}(8.9)} & {\small{}(10.9)} &  & {\small{}(6.4)} & {\small{}(7.6)} & {\small{}(10.4)} & {\small{}(189.2)}\tabularnewline
{\small{}Humanities} & {\small{}-5.3} & {\small{}-36.6} & {\small{}12.8} & {\small{}12.4} &  & {\small{}-3.7} & {\small{}-36.6} & {\small{}-182.4}\tabularnewline
 & {\small{}(21.8)} & {\small{}(12.9)} & {\small{}(15.6)} & {\small{}(10.9)} &  & {\small{}(12.5)} & {\small{}(10.8)} & {\small{}(495.3)}\tabularnewline
{\small{}Other Health} & {\small{}9.9} & {\small{}-36.0} & {\small{}1.3} & {\small{}3.3} & {\small{}-2.2} &  & {\small{}-45.7} & {\small{}-55.2}\tabularnewline
 & {\small{}(15.8)} & {\small{}(4.9)} & {\small{}(11.3)} & {\small{}(3.3)} & {\small{}(7.9)} &  & {\small{}(6.3)} & {\small{}(56.7)}\tabularnewline
{\small{}Technology} & {\small{}62.0} & {\small{}-6.2} & {\small{}57.8} & {\small{}49.5} & {\small{}49.3} & {\small{}56.3} &  & {\small{}-20.9}\tabularnewline
 & {\small{}(19.1)} & {\small{}(9.2)} & {\small{}(16.0)} & {\small{}(11.4)} & {\small{}(12.1)} & {\small{}(19.2)} &  & {\small{}(49.2)}\tabularnewline
{\small{}Law} & {\small{}33.4} & {\small{}-19.0} & {\small{}37.1} & {\small{}37.3} & {\small{}38.6} & {\small{}37.7} & {\small{}-11.7} & \tabularnewline
 & {\small{}(16.4)} & {\small{}(10.6)} & {\small{}(12.8)} & {\small{}(9.5)} & {\small{}(9.4)} & {\small{}(15.4)} & {\small{}(11.9)} & \tabularnewline
{\small{}Medicine} & {\small{}67.0} & {\small{}22.2} & {\small{}54.2} & {\small{}76.4} & {\small{}64.0} & {\small{}57.5} & {\small{}21.1} & {\small{}47.4}\tabularnewline
 & {\small{}(20.1)} & {\small{}(20.6)} & {\small{}(17.3)} & {\small{}(19.0)} & {\small{}(12.9)} & {\small{}(12.3)} & {\small{}(8.7)} & {\small{}(67.8)}\tabularnewline
 &  &  &  &  &  &  &  & \tabularnewline
{\small{}N} & {\small{}5,005} & {\small{}4,402} & {\small{}11,010} & {\small{}11,120} & {\small{}8,293} & {\small{}3,284} & {\small{}3,480} & {\small{}1,362}\tabularnewline
\bottomrule
\end{tabular*}{\small\par}

\medskip{}

\begin{singlespace}
\textbf{\footnotesize{}Note:}{\footnotesize{} 2SLS estimation of equations (\ref{eq:second stage-1})--\ref{eq:firststages-1} results in a matrix of payoffs to field $j$ as compared to $k$ for those who prefer $j$ and have $k$ as next-best field. Each cell is a 2SLS estimate (with standard errors in parenthesis) of the payoff to a given pair of preferred field and next-best field. The rows represent completed fields and the columns represent next-best fields.}{\footnotesize\par}
\end{singlespace}
\end{sidewaystable}

\begin{sidewaystable}

\caption{Payoff estimates Denmark (\$1,000), 13 years after application\label{tab:Payoffs Dk t13}}

{\small{}}%
\begin{tabular*}{1\textwidth}{@{\extracolsep{\fill}}lcccccccc}
\toprule
 & \multicolumn{8}{c}{{\small{}Next-best field}}\tabularnewline
\cmidrule{2-9} \cmidrule{3-9} \cmidrule{4-9} \cmidrule{5-9} \cmidrule{6-9} \cmidrule{7-9} \cmidrule{8-9} \cmidrule{9-9}
 & {\small{}Science} & {\small{}Business} & {\small{}Social Science} & {\small{}Teaching} & {\small{}Humanities} & {\small{}Other Health} & {\small{}Technology} & {\small{}Law}\tabularnewline
\midrule
{\small{}Completed field} &  &  &  &  &  &  &  & \tabularnewline
{\small{}Science} &  & {\small{}-17.6} & {\small{}19.8} & {\small{}33.7} & {\small{}0.7} & {\small{}-2.7} & {\small{}-24.3} & {\small{}-24.3}\tabularnewline
 &  & {\small{}(13.5)} & {\small{}(23.4)} & {\small{}(15.9)} & {\small{}(10.1)} & {\small{}(9.2)} & {\small{}(12.1)} & {\small{}(11.9)}\tabularnewline
{\small{}Business} & {\small{}75.0} &  & {\small{}42.7} & {\small{}76.9} & {\small{}81.4} & {\small{}54.7} & {\small{}74.1} & {\small{}47.0}\tabularnewline
 & {\small{}(27.4)} &  & {\small{}(24.1)} & {\small{}(36.6)} & {\small{}(29.9)} & {\small{}(32.7)} & {\small{}(28.6)} & {\small{}(21.9)}\tabularnewline
{\small{}Social Science} & {\small{}-15.2} & {\small{}-22.4} &  & {\small{}0.6} & {\small{}12.8} & {\small{}4.7} & {\small{}-15.8} & {\small{}-16.5}\tabularnewline
 & {\small{}(11.6)} & {\small{}(12.8)} &  & {\small{}(5.0)} & {\small{}(5.6)} & {\small{}(5.9)} & {\small{}(9.8)} & {\small{}(5.4)}\tabularnewline
{\small{}Teaching} & {\small{}-31.5} & {\small{}-43.1} & {\small{}-35.0} &  & {\small{}-8.3} & {\small{}-23.1} & {\small{}-34.8} & {\small{}-40.6}\tabularnewline
 & {\small{}(12.4)} & {\small{}(13.3)} & {\small{}(10.5)} &  & {\small{}(10.9)} & {\small{}(9.8)} & {\small{}(11.9)} & {\small{}(11.4)}\tabularnewline
{\small{}Humanities} & {\small{}-29.1} & {\small{}-21.0} & {\small{}-20.6} & {\small{}-0.1} &  & {\small{}-0.4} & {\small{}-22.4} & {\small{}-31.2}\tabularnewline
 & {\small{}(9.5)} & {\small{}(10.1)} & {\small{}(12.5)} & {\small{}(5.9)} &  & {\small{}(11.3)} & {\small{}(12.0)} & {\small{}(10.3)}\tabularnewline
{\small{}Other Health} & {\small{}-13.7} & {\small{}-35.9} & {\small{}-12.4} & {\small{}-5.9} & {\small{}10.0} &  & {\small{}-13.3} & {\small{}-49.4}\tabularnewline
 & {\small{}(11.9)} & {\small{}(14.8)} & {\small{}(10.9)} & {\small{}(3.8)} & {\small{}(6.7)} &  & {\small{}(10.7)} & {\small{}(14.1)}\tabularnewline
{\small{}Technology} & {\small{}-4.5} & {\small{}-45.5} & {\small{}-39.3} & {\small{}-19.8} & {\small{}3.5} & {\small{}-12.0} &  & {\small{}-25.5}\tabularnewline
 & {\small{}(10.6)} & {\small{}(13.0)} & {\small{}(12.7)} & {\small{}(19.8)} & {\small{}(9.2)} & {\small{}(15.8)} &  & {\small{}(17.6)}\tabularnewline
{\small{}Law} & {\small{}22.7} & {\small{}3.4} & {\small{}12.5} & {\small{}21.8} & {\small{}30.1} & {\small{}16.3} & {\small{}25.1} & \tabularnewline
 & {\small{}(8.1)} & {\small{}(7.6)} & {\small{}(8.9)} & {\small{}(8.7)} & {\small{}(7.2)} & {\small{}(8.7)} & {\small{}(12.1)} & \tabularnewline
{\small{}Medicine} & {\small{}14.6} & {\small{}7.8} & {\small{}14.5} & {\small{}46.1} & {\small{}27.1} & {\small{}30.5} & {\small{}13.7} & {\small{}2.4}\tabularnewline
 & {\small{}(6.5)} & {\small{}(11.2)} & {\small{}(9.0)} & {\small{}(15.6)} & {\small{}(6.7)} & {\small{}(5.1)} & {\small{}(5.5)} & {\small{}(7.2)}\tabularnewline
 &  &  &  &  &  &  &  & \tabularnewline
{\small{}N} & {\small{}1,477} & {\small{}1,197} & {\small{}1,701} & {\small{}1,493} & {\small{}3,324} & {\small{}1,667} & {\small{}824} & {\small{}520}\tabularnewline
\bottomrule
\end{tabular*}{\small\par}

\medskip{}

\begin{singlespace}
\textbf{\footnotesize{}Note:}{\footnotesize{} 2SLS estimation of equations (\ref{eq:second stage-1})--\ref{eq:firststages-1} results in a matrix of payoffs to field $j$ as compared to $k$ for those who prefer $j$ and have $k$ as next-best field. Each cell is a 2SLS estimate (with standard errors in parenthesis) of the payoff to a given pair of preferred field and next-best field. The rows represent completed fields and the columns represent next-best fields.}{\footnotesize\par}
\end{singlespace}
\end{sidewaystable}

\end{document}